\documentclass{article}
\usepackage[utf8]{inputenc}

\usepackage[margin=1in,marginparwidth=65pt,marginparsep=5pt]{geometry}
\usepackage{color}
\usepackage{graphicx}
\usepackage{dsfont}
\usepackage{epsfig}
\usepackage{epstopdf}
\usepackage{booktabs}
\usepackage{subfig}
\usepackage{multirow}
\usepackage[font=small,skip=0pt]{caption}
\usepackage{tabularx}
\usepackage{fnpct}

\usepackage[dvipsnames]{xcolor}

\usepackage{amsmath}
\usepackage{lipsum}
\usepackage{amssymb}
\usepackage{multirow}
\usepackage{amsthm}

\usepackage{tikz}
\usetikzlibrary{matrix,decorations.pathreplacing}
\usetikzlibrary{calc}
\usepackage{array}

\usepackage{rotating}
\usepackage{verbatim}
\usepackage[colorlinks,
            linkcolor=red,
            anchorcolor=blue,
            citecolor=blue
            ]{hyperref}
\usepackage{bm}
\usepackage[normalem]{ulem}
\usepackage[sort,authoryear]{natbib}
\usepackage{bbm}
\usepackage{array}
\usepackage{enumerate}
\usepackage{float}

\usepackage[flushleft]{threeparttable}
\usepackage{tablefootnote}
\usepackage[linesnumbered,ruled,vlined]{algorithm2e}

\usepackage{ifthen}
\usepackage{url}
\usepackage{mathtools}
\usepackage{arydshln}
\usepackage{authblk}
\title{Diversifying conformal selections}
\author[1]{Yash Nair}
\author[2]{Ying Jin}
\author[1]{James Yang}
\author[1,3]{Emmanuel Cand\`es}
\affil[1]{Department of Statistics, Stanford University}
\affil[2]{Harvard Data Science Initiative}
\affil[3]{Department of Mathematics, Stanford University}
\date{}

\newtheorem{theorem}{Theorem}[section]

\newtheorem{proposition}{Proposition}[section]
\newtheorem{definition}{Definition}[section]
\newtheorem{example}{Example}[section]
\newtheorem{corollary}{Corollary}[section] 
\newtheorem{assumption}{Assumption}[section]
\newtheorem{lemma}[theorem]{Lemma}
\newtheorem{remark}{Remark}

\newtheorem*{setup*}{Setup}
\newtheorem*{nullHypothesis*}{Null hypothesis}
\newtheorem*{constraints*}{Resampling constraints}

\newcommand{\running}{\text{run}}

\newcommand{\sharpe}{\textnormal{Sharpe}}
\newcommand{\markowitz}{\textnormal{Markowitz}}
\newcommand{\Above}{\textnormal{calib-above}}

\newcommand{\bx}{\mathbf{x}}
\newcommand{\bS}{\mathbf{S}}

\newcommand{\bv}{\boldsymbol{v}}

\newcommand{\rsc}{\text{RSC}}
\newcommand{\mc}{\text{MC}}
\newcommand{\coarse}{\text{coarse}}

\newcommand{\loo}{i\rightarrow \star}

\newcommand{\test}{\textnormal{test}}
\newcommand{\Sigmatest}{\Sigma^{\test}}
\newcommand{\testbelow}{\textnormal{test-below}}

\newcommand{\bd}{\textbf{d}}
\newcommand{\sub}{\text{sub}}
\newcommand{\RR}{\mathbb{R}}

\newcommand{\clip}{\textnormal{clip}}

\newcommand{\bb}{\mathbf{b}}

\newcommand{\exch}{\textnormal{exch}}

\newcommand{\bW}{\textbf{W}}

\newcommand{\bxi}{\boldsymbol{\xi}}

\newcommand{\underrep}{\textnormal{Underrep}}
\newcommand{\bmu}{\boldsymbol{\mu}}

\newcommand{\by}{\mathbf{y}}
\newcommand{\low}{\text{low}}
\newcommand{\high}{\text{high}}
\newcommand{\bz}{\mathbf{z}}
\newcommand{\bZ}{\mathbf{Z}}

\newcommand{\RN}[1]{%
  \textup{\uppercase\expandafter{\romannumeral#1}}%
}
\newcommand{\proj}{\text{proj}}

\newcommand{\sharpediversityR}{\varphi^{\sharpe}(\bZ_{\mathcal{S}})}

\newcommand{\markowitzdiversityR}{\varphi^{\markowitz}(\bZ_{\mathcal{S}})}

\newcommand{\underrepdiversityR}{\varphi^{\underrep}(\bZ_{\mathcal{S}})}

\newcommand{\bE}{\mathbb{E}}

\DeclareMathOperator*{\argmax}{arg\,max}
\DeclareMathOperator*{\argmin}{arg\,min}

\newcommand{\bchi}{\boldsymbol{\chi}}

\newcommand{\Zsort}{\bZ^{()}}

\newcommand{\be}{\mathbf{e}}

\newcommand{\bh}{\textnormal{BH}}
\newcommand{\sort}{\textnormal{sort}}

\newcommand{\esort}{\varepsilon^{(t)}}
\definecolor{lightblue}{rgb}{0.4,0.6,1}
\newcommand{\fdr}{\textnormal{FDR}}

\newcommand{\bP}{\mathbb{P}}

\newcommand{\bchirelaxedt}{\bchi^{*,\text{relaxed},\tau^*}}
\newcommand{\chirelaxedt}{\chi^{*,\text{relaxed},\tau^*}}

\newcommand{\rrelaxedt}{\mathcal{R}^{*,\text{relaxed}}_{\tau^*}}
\newcommand{\bzero}{\boldsymbol{0}}
\newcommand{\ind}{\mathds{1}}

\makeatletter
\renewcommand{\paragraph}{%
  \@startsection{paragraph}{4}%
  {\z@}{1.25ex \@plus 1ex \@minus .2ex}{-1em}%
  {\normalfont\normalsize\bfseries}%
}
\makeatother

\SetKwInput{KwInput}{Input}                
\SetKwInput{KwOutput}{Output}              

\counterwithout*{footnote}{section}
\counterwithout*{footnote}{subsection}
\counterwithout*{footnote}{subsubsection}

\usepackage{color-edits}
\addauthor[Ying]{ying}{magenta}
\addauthor[Yash]{yash}{red}

\begin{document}

\maketitle

\begin{abstract}
When selecting from a list of potential candidates, it is important to ensure not only that the selected items are of high quality, but also that they are sufficiently dissimilar so as to both avoid redundancy and to capture a broader range of desirable properties. In drug discovery, scientists aim to select potent drugs from a library of unsynthesized candidates, but recognize that it is wasteful to repeatedly synthesize highly similar compounds. In job hiring, recruiters may wish to hire candidates who will perform well on the job, while also considering factors such as socioeconomic background, prior work experience, gender, or race.
We study the problem of using any prediction model to construct a maximally diverse selection set of candidates while controlling the false discovery rate (FDR) in a model-free fashion.
Our method, \emph{diversity-aware conformal selection (DACS)}, achieves this by designing a general optimization procedure to construct a diverse selection set subject to a simple constraint involving conformal e-values which depend on carefully chosen stopping times. The key idea of DACS is to use optimal stopping theory to adaptively choose the set of e-values which (approximately) maximizes the expected diversity measure.
We give an example diversity metric for which our procedure can be run exactly and efficiently. We also develop a number of computational heuristics which greatly improve its running time for generic diversity metrics. We demonstrate the empirical performance of our method both in simulation and on job hiring and drug discovery datasets.
\end{abstract}

\section{Introduction}

\subsection{Motivation}\label{motivation}
We study the problem of selecting from a list of potential candidates a subset whose as-yet-unknown performance is \emph{good} 
while also being \emph{diverse}. 
This problem arises in a number of real-world applications: two particularly salient examples are in drug discovery and job hiring. In the former setting, scientists seek to select, from a large library of unsynthesized compounds, a subset of drugs for further investigation. 
It is important to identify a shortlist that contains a high proportion of potent compounds, so that the researcher does not waste resources further investigating and/or synthesizing (the process of which is extremely expensive) inactive compounds. At the same time, scientists also emphasize the importance of selecting structurally \emph{diverse} candidates~\citep{nakamura2022selecting,maggiora2005practical,glen2006similarity,lambrinidis2021multi,gillet2011diversity,jang2024llmsgeneratediversemolecules}, as screening many highly similar compounds is inefficient given that they will tend to share biological mechanisms. In the job hiring setting, it is important to ensure that the selected candidates are not only of a high caliber, but also diverse as measured by, e.g., undergraduate university, gender, race, etc.~\citep{marra2024addressing, roshanaei2024towards,houser2019can,jora2022role}.

In these settings, practitioners typically have on hand intricate pretrained machine learning (ML) models designed to predict the quality of any candidate given the candidate's features (e.g., a drug's chemical structure or a job applicant's prior work qualifications) \citep{huang2020deeppurpose,askr2023deep,mahmoud2019performance,sharma2021novel}. Using this pretrained \emph{blackbox} ML model, as well as a \emph{reference} dataset of candidates whose performance has already been measured, we aim to deliver a selected set of candidates which (provably) contains mostly high-performing candidates \emph{and} is also diverse. We emphasize once again that focusing on only one of these objectives \emph{alone} may lead to undesirable results: returning a set of candidates which is mostly high quality but is not diverse may be either inefficient (in drug discovery settings) or inequitable (in the case of diversity-aware job hiring) while returning a diverse set of candidates that contains too many low quality candidates is also clearly unacceptable. 

Apart from some recent works \citep{wu2024optimal,huo2024real} that study diverse selection (we discuss these in detail in Section~\ref{related-work}), the literature on candidate selection and multiple hypothesis testing has nearly exclusively been focused on the goal of returning selection sets that are as large as possible subject to various statistical reliability criteria such as false discovery rate (FDR) control \citep[e.g.,][]{lee2024boostingebhconditionalcalibration,xu2023more,bai2024optimizedconformalselectionpowerful,lee2025fullconformalnoveltydetectionpowerful}. It is worth emphasizing that this is different than---and often at odds with---our goal of constructing a \emph{diverse} selection set subject to the same FDR constraint. 
Indeed, large/powerful selection sets---such as the conformal selection (CS) set \citep{jin2023selection}---may exhibit poor diversity due to the inclusion of promising yet similar instances. In such cases, it is desirable to ``prune down'' such selection sets into smaller, more diverse ones while still controlling the FDR. This is a non-trivial task: pruning often requires removing promising candidates to increase diversity, which, when naively done, can compromise error control. 
In general, one must carefully balance diversity with FDR, much like the familiar tradeoff between power and Type-I error in any statistical inference problem. In some cases, it may be appropriate to allow for a more liberal guarantee for the sake of increased diversity, while in others maintaining a low FDR may be more important than diversity. Figure~\ref{fig:init-tradeoff} illustrates this tradeoff curve on a job hiring dataset from \cite{roshan2020campus} considered in Section~\ref{expers:hiring}. We discuss these tradeoffs in more detail in Section~\ref{contribution}.

\begin{figure}
    \centering
    \hspace*{1.5cm}\includegraphics[scale=0.5]{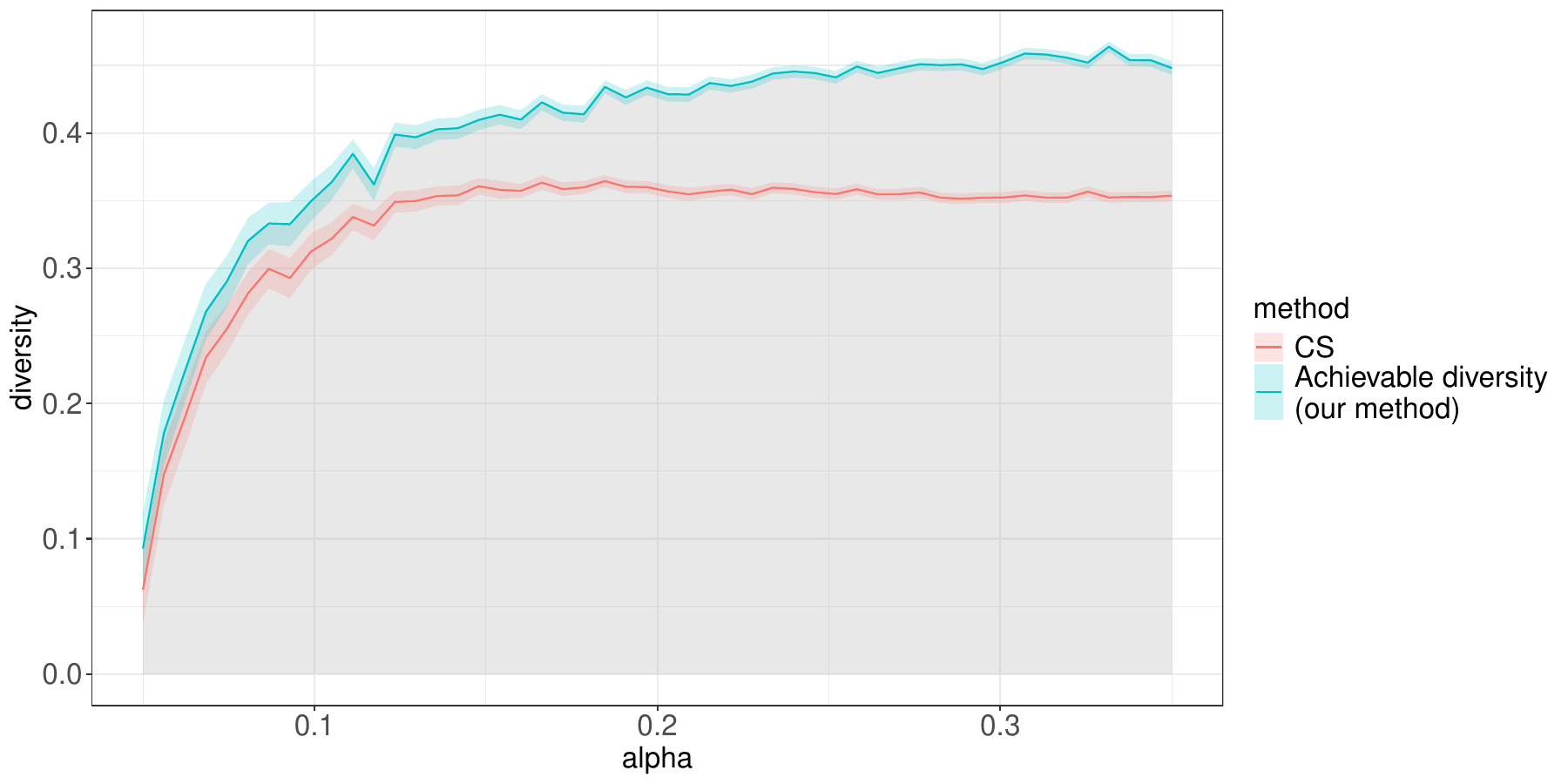}\hspace*{-1.5cm}
    \caption{Tradeoff between diversity and nominal FDR level $\alpha$ on a job hiring dataset discussed in Section~\ref{expers:hiring}. As $\alpha$ increases, the achieved diversity by our method (blue) increases and, in particular, dominates that of the conformal selection (CS) \citep{jin2023selection} method (red) at the same nominal level.}
    \label{fig:init-tradeoff}
\end{figure}

Finally, we note that many of the tools developed in this paper also apply broadly to settings where diversity may not be the desired objective (though diversity will be our main focus throughout this paper). Our core methodology can be used in settings in which a practitioner is interested in optimizing any generic (i.e., non-power-based) criterion as a function of the selection set. Section~\ref{contribution} provides a more detailed discussion.

\subsection{Setup}\label{sec:setup}
We build upon the conformal selection method introduced in \cite{jin2023selection}. Before presenting the methodological framework, we first introduce the basic setup and inferential targets. We assume access to a set of labeled calibration samples $\{(Z_i,X_i,Y_i)\}_{i=1}^n \subseteq \mathcal{Z} \times \mathcal{X} \times \mathbb{R}$ drawn i.i.d.~from some distribution $\bP := P_Z \times P_{X \mid Z} \times P_{Y \mid Z,X}$, a set of unlabeled test samples $\{(Z_{n+i},X_{n+i})\}_{i=1}^m \subseteq \mathcal{Z} \times \mathcal{X}$ drawn i.i.d.\footnote{Our method and theory also easily extend to the case wherein the calibration and test datasets are exchangeable.}~from $P_Z \times P_{X \mid Z}$, and a pretrained ML model $\hat{\mu}: \mathcal{X} \rightarrow \mathbb{R}$ which predicts the quality $Y$ of a candidate given (only) its $X$-covariates; the unobserved test labels are denoted by $Y_{n+1}, \ldots, Y_{n+m}$. 
We distinguish between covariates $Z$ (used to define diversity) and $X$ (used for prediction) to accommodate settings where we wish to diversify on a different set of features than those from which $\hat{\mu}$ bases its quality predictions. For example, in hiring, $Z$ might represent protected demographic attributes for measuring diversity, while $X$ consists of prior job performance metrics used by $\hat{\mu}$ to predict quality (it is typically impermissible to base these predictions off of such demographic characteristics). In other cases---such as drug discovery---$Z$ and $X$ may coincide, e.g., both representing molecular structure.

We will constrain the selected set to not contain too many low-quality candidates. To this end, we define a candidate $i$ as \emph{low-quality} if $Y_i \leq 0$; conversely, candidates with $Y_i > 0$ are deemed \emph{high-quality}.\footnote{More generally, one may set candidate-specific and possibly non-zero thresholds $c_i$ and define candidate $i$ as low- or high-quality if $Y_i \leq c_i$ of $Y_i > c_i$, respectively. Provided that the triples $(X_i, Y_i, c_i)$ are exchangeable across calibration and test datasets, this setting is readily accommodated by simply redefining the response as $Y_i - c_i$ and applying the zero threshold.} We then require that the false discovery rate (FDR) of our selected set $\mathcal{R} \subseteq [m]$ be controlled below a user-specified level $\alpha \in (0,1)$:
\begin{equation}\label{fdr-eq}\fdr := \bE\left[\frac{\sum_{i=1}^m \mathds{1}\{i \in \mathcal{R}, Y_{n+i} \leq 0\}}{1 \vee |\mathcal{R}|}\right] \leq \alpha.\end{equation} 
The FDR measures the average fraction of resources that would be wasted on low-quality instances in the selection set. 
Following \cite{jin2023selection}, it will at times be helpful to view the problem as a \emph{random} multiple testing problem, in the sense that the hypotheses $H_{i}: Y_{n+i} \leq 0, i = 1, \ldots, m$ are random events, as opposed to the standard multiple testing paradigm in which the hypotheses are not random events, but rather describe properties of the joint distribution of $(X,Y)$. 

As previously mentioned, the aim of this paper is to maximize some notion of diversity of the selection set $\mathcal{R}$, subject to finite-sample FDR control. Formally, let $\bZ := (Z_{1}, \ldots, Z_{n+m})$ denote the list of all diversification variables. Given any user-specified function $\varphi(\bZ_{\mathcal{S}})$ which measures the diversity of any subset $\mathcal{S} \subseteq [n+m]$ in terms of its $Z$-variables $\bZ_{\mathcal{S}} := \{Z_i: i \in \mathcal{S}\}$, our goal is to produce a selection set $\mathcal{R}\subseteq [m]$ that (approximately) maximizes $\bE[\varphi(\bZ^{\test}_{\mathcal{R}})]$ while obeying~\eqref{fdr-eq}, where $\bZ^{\test} := (Z_{n+1}, \ldots, Z_{n+m})$ is the list of test diversification variables.
In general, we place no constraints on the diversity metric $\varphi$: the user is free to choose any metric they wish---though in Section~\ref{dacs-heuristics}, we will need to make some assumptions about $\varphi$ for the sake of computational efficiency (but not FDR control). To provide more context, we preview some concrete examples of diversity measures and discuss application settings in which practitioners might use them as a natural objective.

Our first two example objectives, the \emph{Sharpe ratio} and \emph{Markowitz objective} of Examples~\ref{sharpe-ex} and \ref{markowitz-ex} depend on a positive definite matrix $\Sigma \in \mathbb{R}^{(n+m) \times (n+m)}$ which encodes pairwise similarities between $(Z_{1}, \ldots, Z_{n+m})$.
For example, if the $Z$'s take continuous values in $\mathbb{R}^d$ then taking $\Sigma$ to be the radial-basis function (RBF) kernel matrix may be a reasonable choice. Another example is the Tanimoto coefficient matrix \citep{tanimoto1958elementary,bouchard2013proof}, which is commonly used in cheminformatics and drug discovery to measure similarity between molecular fingerprints.

\begin{example}[Sharpe ratio]\label{sharpe-ex}
    The \emph{Sharpe ratio} \citep{sharpe1966mutual}, commonly used in the finance literature, is defined as \begin{equation}\label{sharpe}
        \sharpediversityR := \begin{cases}
            \frac{|\mathcal{S}|}{\sqrt{\mathbf{1}_{\mathcal{S}}^\top \Sigma \mathbf{1}_{\mathcal{S}}}}, & \text{ if } |\mathcal{S}| > 0,\\
            0, & \text{ if } \mathcal{S} = \emptyset.
        \end{cases}
    \end{equation} A selection set $\mathcal{R}$ achieving a high Sharpe ratio effectively trades off between making many selections while ensuring that these selections are, on the whole, not too similar.
\end{example}

\begin{example}[Markowitz objective]\label{markowitz-ex}
    Like the Sharpe ratio, the \emph{Markowitz objective} \citep{markowitz}, is also used in the portfolio optimization and diversification literature, and encourages selection sets which balance between the two competing objectives of size and total diversity. Rather than a ratio, the Markowitz objective achieves this balance by using a weighted combination of these objectives: \begin{equation}\label{markowitz}
        \markowitzdiversityR :=
            |\mathcal{S}| - \frac{\gamma}{2} \cdot \mathbf{1}_{\mathcal{S}}^\top \Sigma \mathbf{1}_{\mathcal{S}},
        \end{equation} where $\gamma > 0$ is a tuning parameter controlling the trade off between diversity and set-size (small values of $\gamma$ correspond to large but non-diverse sets being preferable while large values of $\gamma$ encode desirability of small but diverse sets).
\end{example}

Our last example diversity measure, the \emph{underrepresentation index}, is applicable in settings where the diversification variable $Z$ is categorical and one would like to construct a selection set of high-performing candidates without underrepresenting any category. This diversity metric may be appropriate, for example, in the job hiring or admissions settings discussed in Section~\ref{motivation}.

\begin{example}[Underrepresentation index]\label{underrep-ex}
    Suppose that the diversification variable $Z$ takes on at most $C\in\mathbb{N}^+$ categorical values, i.e., $Z_i \in [C]$ for all $i = 1, \ldots, n+m$. The \emph{underrepresentation index} of a set is the proportion of its least-represented category. More formally, for any set $\mathcal{S} \subseteq [n+m]$, define $N_c(\mathcal{S}) := \sum_{i \in \mathcal{S}} \mathds{1}\{Z_{i} = c\}$ to be the number of selections whose $Z$-value belongs to category $c$. We then define the underrepresentation index to be \begin{equation}\label{underrep-def}\underrepdiversityR := \begin{cases}
        \min_{c\in[C]} \frac{N_c(\mathcal{S})}{|\mathcal{S}|}, & \text{ if } |\mathcal{S}| > 0,\\
        -\frac{1}{C}, & \text{ if } \mathcal{S} = \emptyset.
    \end{cases}\end{equation}

    The underrepresentation index is closely related to the Berger-Parker dominance \citep{berger1970diversity}, a popular biodiversity index which measures the proportion of the \emph{most}-represented category (a sample is diverse if the Berger-Parker dominance is small). As we will see in Section~\ref{dacs-underrep}, the specific form of the underrepresentation index enables a fast and exact implementation of our method, though our framework applies generally to other discrete diversity indices such as the Gini-Simpson index \citep{simpson1949measurement} and entropy.
\end{example}

\subsection{Our contributions}\label{contribution}
We introduce \emph{diversity-aware conformal selection} (DACS), a framework for constructing a selection set that (approximately) maximizes a generic diversity metric while controlling the false discovery rate (FDR) at a pre-specified level $\alpha\in(0,1)$ in finite samples.

Our idea is to \emph{prune} a large and non-diverse selection set---in particular the one returned by conformal selection (CS) \citep{jin2023selection}---into a diverse one which still controls the FDR at the same nominal level. This is a non-trivial task since, in general, subsets of an FDR-controlling set will fail to control FDR at the same level.

There are two key ideas underlying DACS: 
\vspace{-0.5em}
\begin{enumerate}[(1)]\setlength\itemsep{0em}
    \item The exchangeability between the calibration and test data enables the construction of a family of e-values \citep{vovk2021values} based on any black-box prediction model $\hat\mu$. Such e-values are indexed by stopping times $\tau$ and, as we show, any choice of stopping time $\tau$ yields valid e-values.
    \item Any selection set that satisfies a simple constraint with respect to a set of valid e-values is guaranteed to control the FDR. Therefore, for any stopping time, we can run an ``e-value optimization program'' to build an FDR-controlling selection set that maximizes diversity subject to the constraint imposed by the e-values at that stopping time. 
\end{enumerate} 
\vspace{-0.5em}
Leveraging these two ideas, DACS---which we introduce in Section~\ref{dacs-main-sec}---uses optimal stopping theory to find an (approximately) optimal choice of the stopping time $\tau^*$ whose associated e-values will lead to a selection set that is diverse and controls FDR. The final selection set is then constructed by running the e-value optimization program with the stopping time $\tau^*$. 
It is worth noting that this core methodology is completely agnostic to the fact that our goal is diversity. Indeed, our method applies much more broadly to settings in which the user is interested in optimizing any general function $\varphi$ of the selected data.

In Section~\ref{dacs-underrep} we develop a fast and exact implementation of our method when $\varphi$ is the underrepresentation index. For a generic diversity metric, though, implementing DACS may be non-trivial as it typically requires solving \emph{many} e-value optimization problems---each an integer program---to compute the optimal stopping time. To address this, we develop, in Section~\ref{dacs-heuristics}, several strategies that substantially reduce runtime:
\begin{itemize}
    \item In Section~\ref{lp-relax}, we introduce a technique to relax the integer e-value optimization programs while still approximately controlling FDR. Without this relaxation, running our method with generic diversity metrics would be infeasible for all but the smallest problems. This relaxation technique may also be of independent interest in standard e-value-based multiple testing contexts.
    \item Section~\ref{warm-start} develops custom projected gradient descent solvers for the (relaxed) Sharpe ratio and Markowitz objectives, along with a warm-start heuristic to further accelerate the optimal stopping computation. We implement efficient versions of these solvers in C++ and they can substantially reduce runtime, sometimes by orders of magnitude. 
    \item We devise, in Section~\ref{skipping}, a heuristic that reduces the total number of time points over which the optimal stopping computation is performed. This gives users nearly complete control over computational cost, allowing them to trade off runtime against quality of the selection set (as measured by diversity).
\end{itemize} 
We also illustrate the effectiveness of our framework via empirical evaluations on datasets from job hiring and drug discovery in Section~\ref{expers} and in simulation in Section~\ref{sims}. All code and data used in these empirical studies is publicly available at \url{https://github.com/Yashnair123/diverseSelect}.

As we discussed briefly in Section~\ref{motivation}, returning a subset of the CS selection set is desirable in situations where repeatedly processing similar instances is wasteful or processing each candidate is infeasible. 
In general, though, if the practitioner wishes to construct a diverse selection set that is \emph{not} a subset of the CS level-$\alpha$ selection set, then they will have to run DACS at a nominal FDR level greater than $\alpha$. 
It is up to the practitioner to decide how to make this tradeoff
by weighing several competing considerations: 
\begin{itemize}  \setlength\itemsep{0em}
    \item Whether they want to potentially make selections that go beyond (i.e., are not contained in) the CS level-$\alpha$ selection set
    \item How important diversity of the selection set is to them
    \item At what nominal level they want their selection set's FDR to be controlled
\end{itemize}
Both the decision to run DACS and the nominal level at which to run it will depend on the relative importance the user places on each of these factors.

\subsection{Related work}\label{related-work}

\paragraph{Conformal selection and outlier detection}
Our work builds on the model-free selection literature \citep{jin2023modelfreeselectiveinferencecovariate,jin2023selection,bai2024optimizedconformalselectionpowerful}, particularly the conformal selection (CS) framework introduced in \citet{jin2023selection}. Recent work \citep{rava2024burdensharedburdenhalved,bashari2024derandomized,mary2022semi,lee2024boostingebhconditionalcalibration} has shown how conformalized multiple testing procedures can be reformulated as stopping-time-based methods in the deterministic multiple hypothesis testing setting. These works, however, have been primarily focused on one or two specific stopping times related to the BH procedure and, from this vantage point, one important methodological contribution of our work is to observe that different stopping times can result in selection sets with varying diversity while still controlling FDR. We choose a stopping time using optimal stopping theory to (approximately) maximize diversity.

\paragraph{FDR control using side information}
This work is conceptually related to the extensive literature on using side information in multiple testing while maintaining (approximate) FDR control \citep{li2019multiple,lei2018adapt,ren2023knockoffs,lei2017star,ignatiadis2021covariate,ignatiadis2016data}. The primary goal in these works is improving power, not diversity---though \citet{lei2017star} also consider structural coherence, guided by domain knowledge, as a motivating objective. In contrast, our work leverages the side information given by the diversification variables $Z_1, \ldots, Z_{n+m}$ specifically to increase selection diversity.
Relatedly, \citet{benjamini2007false,sun2015false,sesia2020multi,gablenz2024catch,katsevich2019multilayer,foygel2015p,spector2024controlled} use spatial side information to perform cluster-level testing, possibly across multiple resolutions. Again, these approaches use the spatial side information to either boost power or enhance interpretability (e.g., enforcing spatial locality of the signal) and aim for various notions of cluster-level error control. Again, our goal is fundamentally different: rather than identifying clusters of candidates, we aim to select a diverse set---ideally one that avoids overrepresenting any single cluster if it exists.

\paragraph{Multiple hypothesis testing with e-values}
Our ability to optimize for diversity hinges on the use of e-values, which enable flexible multiple testing procedures \citep{gablenz2024catch,wang2022false,xu2021unified}. We contribute a novel method that jointly optimizes (1) the stopping rule used to define e-values and (2) the strategy for selecting a diverse set under the resulting e-value constraint. This joint optimization is carried out via optimal stopping theory. Indeed, this idea of (approximately) optimally choosing which e-values to use before applying an e-value-based selection procedure may be of interest in other contexts, like those studied in \citet{gablenz2024catch,wang2022false,xu2021unified}.

\paragraph{Diverse candidate selection} The closest papers to this work are \citet{huo2024real,wu2024optimal}, who also study the problem of selecting a diverse set of candidates subject to various notions of Type-I error control. We mention three key differences between these related works and ours. 

\begin{enumerate}
    \item \textbf{Different settings:} \citet{wu2024optimal} maximize diversity under a fixed selection size budget and control a variant of the FDR asymptotically. In contrast we require no such budget and target exact finite sample FDR control. \citet{huo2024real} operate in an online setting, controlling FDR asymptotically at any stopping time and treating diversity as a constraint (rather than an objective). In comparison our problem setting is offline, and our method \emph{optimizes} for diversity.
    \item \textbf{Assumptions}: Both \citet{huo2024real,wu2024optimal} make strong smoothness and boundedness assumptions on the densities of the conditional distributions $\hat{\mu}(X) \mid Y \leq 0$ and $\hat{\mu}(X) \mid Y > 0$ to achieve asymptotic error control. Even if these strong assumptions hold, their methods may still be anti-conservative in finite samples. In contrast, our method makes no assumptions on the conditional distributions of the ML predictions and is exactly valid in finite samples. While we introduce heuristics in Section~\ref{dacs-heuristics} that may inflate FDR by a known factor, this can be corrected without assumptions---though we observe no inflation in practice and recommend against making this adjustment.
    \item \textbf{Diversity metrics:}  \citet{huo2024real,wu2024optimal}’s methods apply to a limited class of diversity metrics, namely pairwise similarity measures such as the kernel similarity used in the Sharpe ratio and Markowitz objective. These metrics cannot capture more holistic notions of diversity (like our underrepresentation index does, for instance) that consider the selection set as a whole and cannot be written in terms of pairwise similarities alone. Our method, on the other hand, accommodates any user-specified diversity metric, including those not expressible through pairwise relationships.
\end{enumerate}

\section{Diversity-aware conformal selection}\label{dacs-main-sec}

\subsection{Reinterpreting conformal selection}\label{method-review}

We begin with a recap of~\cite{jin2023selection}'s conformal selection (CS) method. We then discuss how the CS procedure can be written as the output of the eBH filter \citep{wang2022false} applied to certain stopping-time-based e-values. The flexible structure in these e-values will be crucial in building our method.

Using the pre-trained machine-learning algorithm $\hat{\mu}\colon \mathcal{X} \to \mathbb{R}$, the CS procedure defines a \emph{score} function $V: \mathcal{X} \times \mathbb{R} \rightarrow \mathbb{R}$ that is \emph{monotone} in its second argument, i.e., $V(x,y)\leq V(x,y')$ for any $x\in \mathcal{X}$ and any $y\leq y'$, $y,y'\in \mathbb{R}$. 
As shown in \cite{jin2023selection}, a particularly powerful choice is the ``clipped'' score function:\footnote{Any monotone score function $V(x,y)$ can be ``clipped'' to $\infty \ind\{y>0\}+V(x,0)\ind\{y\leq 0\}$, for which all techniques in this paper continue to apply; see~\cite{jin2023selection,bai2024optimizedconformalselectionpowerful} for related discussion.} \[V^{\clip}(x,y) = \begin{cases}
    \infty, & \text{ if } y > 0,\\
    -\hat{\mu}(x), & \text{ if } y \leq 0.
\end{cases}\] Throughout this paper, we will focus solely on this clipped score function and require only that the analyst use a fitted quality-prediction model $\hat{\mu}$ for which the values $\hat{\mu}(X_1), \ldots, \hat{\mu}(X_{n+m})$ are almost surely distinct (this can always be achieved by, e.g., adding exogenous noise to $\hat{\mu}$'s outputs).

Given the score function $V^{\clip}$, CS computes calibration scores $V_i := V^\clip(X_i, Y_i)$ for each $i = 1,\ldots, n$ as well as imputed test scores $\widehat{V}_i := V^\clip(X_{n+i}, 0)$ for each $i = 1, \ldots ,m$ and constructs p-values \begin{equation}\label{conformal-pval}p_i := \frac{1 + \sum_{j=1}^n \mathds{1}\{\widehat{V}_i \geq V_j\}}{n+1} \quad\text{ for } i = 1, \ldots, m.\end{equation} The selection set returned is then the output of the Benjamini-Hochberg (BH) filter \citep{benjamini1995controlling} applied to the p-values $p_1, \ldots, p_m$ defined in~\eqref{conformal-pval}.\footnote{Because the hypotheses in our setting are \emph{random}, the variables $p_1, \ldots, p_m$ are not p-values in the standard sense. Rather, as is shown in \cite{jin2023selection}, they satisfy the condition that $\bP(p_i \leq \alpha, Y_{n+i} \leq 0) \leq \alpha$ for all $i = 1, \ldots, m$ and $\alpha \in (0,1)$.}

The first key observation \citep{lee2024boostingebhconditionalcalibration} is that the CS selection set can be written as the output of the eBH filter \citep{wang2022false} applied to certain e-values.  A similar viewpoint has also been discussed in \cite{bashari2024derandomized} for the related (deterministic) multiple testing problem of outlier detection. This perspective will allow us to exploit the martingale structure underlying a \emph{general} class of e-values---in particular, not just those relating to CS---to optimize for diversity.
Formally, let $\bW := (W_1, \ldots, W_{n+m}) := (V_1, \ldots, V_n, \widehat{V}_1, \ldots, \widehat{V}_m)$ denote the complete list of calibration and test scores. We let $W_{(1)} \leq \cdots \leq W_{(n+m)}$ denote their sorted values (ties among scores equal to infinity may be broken arbitrarily), with $W_{(0)} := -\infty$. Let $\pi_{\sort}$ be the permutation for which $(W_{\pi_{\sort}(1)}, \ldots, W_{\pi_{\sort}(n+m)}) = (W_{(1)}, \ldots, W_{(n+m)})$ and define $B_s := \mathds{1}\{\pi_{\sort}(s) \leq n\}$ to be the indicator that the $s^{\text{th}}$ smallest score belongs to the calibration set. Then, for each $t \in [n+m]$, let $N^{\Above}_t$ be the number of calibration points strictly above $W_{(t)}$: $N^{\Above}_t := \sum_{s=t+1}^{n+m} B_s.$ Finally, for each $t \in [n+m]$, define \begin{equation}\label{stopping-time-based-evals}e^{(t)}_i := \frac{(n+1) \mathds{1}\{\widehat{V}_i \leq W_{(t)}\}}{1 + n - N^{\Above}_t} \quad \text{ for all } i = 1, \ldots, m \text{ and }t = 1, \ldots, n+m.\end{equation}

\begin{figure}
    \centering
    \includegraphics[scale=0.5]{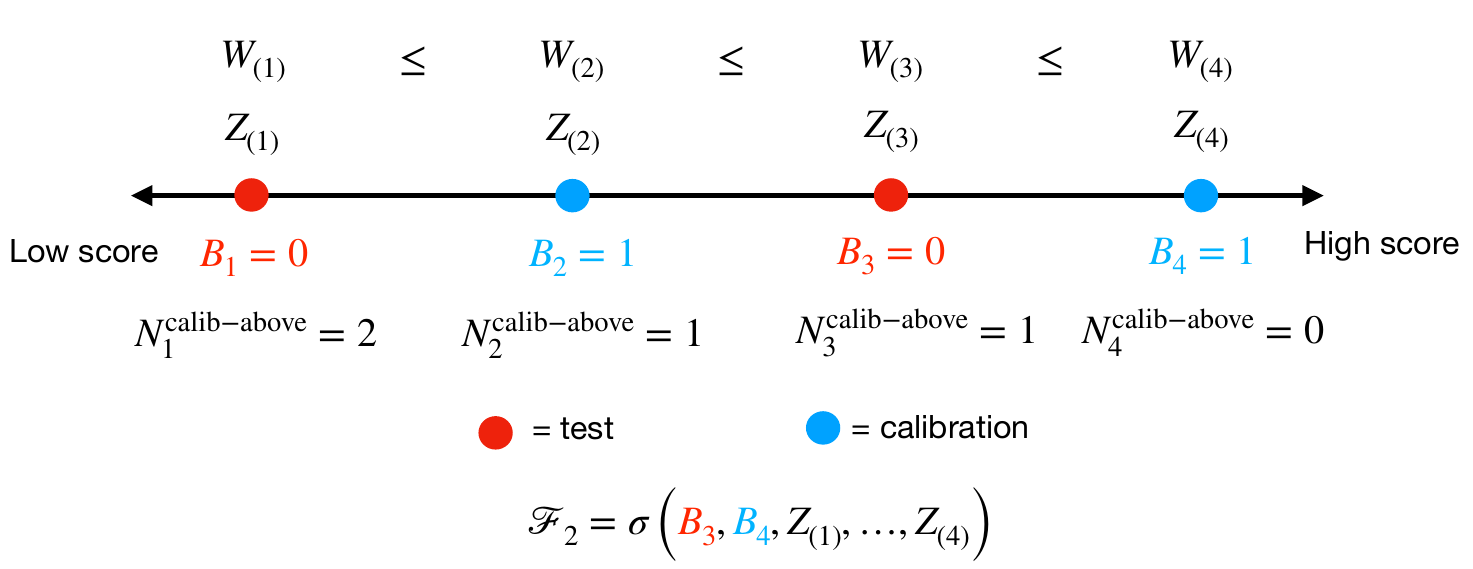}
    \caption{Illustration of the variables involved in defining the filtration as well as variables defining filtration at time $2$, shown for an example with $n=2$ calibration points (blue) and $m=2$ test points (red).}
    \label{fig:key-notation}
\end{figure}

We then define a backwards filtration, which at time $t \in [0:n+m]$ is given by \[\mathcal{F}_t := \sigma\left(B_{t+1}, \ldots, B_{n+m}, \Zsort\right),\] the $\sigma$-algebra generated by the indicators of the $(t+1)^{\text{th}}$ through $(n+m)^{\text{th}}$ smallest scores belonging to the calibration set as well as $\Zsort$, where $\Zsort := (Z_{(1)}, \ldots, Z_{(n+m)})$ is the list of $\bW$-sorted diversification variables $Z_i, i \in [n+m]$ (i.e., $Z_{(t)} := Z_{\pi_{\sort}(t)}$). In a nutshell, as $t$ decreases, the filtration gradually reveals whether the $t^{\text{th}}$ sorted score belongs to the calibration or test set; Figure~\ref{fig:key-notation} provides a visual illustration of the notation used to define the filtration. As is shown in \citet[][Proposition 8, part (2)]{lee2024boostingebhconditionalcalibration}, the selection set returned by CS is identical to the output of the eBH procedure applied to the variables $ (e^{(\tau_{\bh})}_1, \dots, e^{(\tau_{\bh})}_m )$ for a particular stopping time $\tau_{\bh}$ which we call the BH stopping time. In Appendix~\ref{appendix:ebh-evalues} we both give a formal definition of the BH stopping time and establish that the variables $(e^{(\tau_{\bh})}_1, \dots, e^{(\tau_{\bh})}_m )$ are e-values (in the sense of equation~\eqref{null-eval} below).

\subsection{High-level idea: a layered optimization framework}\label{sec:layered}

There are two key observations on which DACS is built:
\begin{enumerate}
    \item Choosing any stopping time $\tau$ different than $\tau_{\bh}$ still produces valid e-values $e^{(\tau)}_1, \dots, e^{(\tau)}_m$. The following theorem establishes this result formally.

    \begin{theorem}\label{stopping-time-thm}
    Let $\tau$ be any stopping time with respect to the reverse filtration $\left(\mathcal{F}_t\right)_{t=0}^{n+m}$. Then the variables $e^{(\tau)}_1, \dots, e^{(\tau)}_m$ are e-values in the sense that they satisfy \begin{equation}\label{null-eval}
        \bE[e^{(\tau)}_i \mathds{1}\{Y_{n+i} \leq 0\}] \leq 1 \text{ for all } i = 1, \ldots, m.
    \end{equation}
\end{theorem}
    
    Theorem~\ref{stopping-time-thm} builds on results of \cite{bashari2024derandomized,lee2024boostingebhconditionalcalibration}, who show the validity of conformal e-values in the (deterministic) outlier detection setting, and extends them to our selection problem. The proof of Theorem~\ref{stopping-time-thm} is given in Appendix~\ref{appendix:main-thm-proof}. Though part of the proof applies the same supermartingale given in \citet{lee2024boostingebhconditionalcalibration}, considerable care must be taken to account for the fact that our problem is a random multiple hypothesis testing problem. This necessitates a careful leave-one-out-type analysis before the supermartingale argument can be invoked; indeed, this leave-one-out analysis is why we restrict attention to \emph{only} the clipped score function, as mentioned at the outset of this section.

    \item A sufficient condition for any multiple testing procedure to enjoy FDR control in our random multiple testing problem is that it satisfies e-value \emph{self-consistency} \citep{wang2022false}.

    \begin{definition}
        A selection set $\mathcal{R}$ is self-consistent with respect to the variables $e_1, \ldots, e_m$ if \begin{equation}\label{sc}
        e_i \geq \frac{m}{\alpha|\mathcal{R}|} \text{ for all } i \in \mathcal{R}.
    \end{equation}
    \end{definition}
    
    \begin{theorem}\label{thm:self-consistency-validity}
        Let $e_1, \ldots, e_m$ be e-values in the sense that $e_i\geq 0$ and  $\bE[e_i\mathds{1}\{Y_{n+i} \leq 0\}] \leq 1$ for all $i \in[m]$. 
        Then any selection set that is self-consistent with respect to $e_1, \ldots, e_m$ controls FDR at level $\alpha$.
    \end{theorem} Both the statement and proof of Theorem~\ref{thm:self-consistency-validity} are essentially identical to \citet[][Proposition 2]{wang2022false}, which establishes the same FDR control statement for self-consistent selection sets in the standard deterministic multiple testing setting. Theorem~\ref{thm:self-consistency-validity} says that FDR control continues to hold in our random multiple hypothesis testing setting; see Appendix~\ref{appendix:self-consistency} for a proof.
    
\end{enumerate}
Using these two observations, our method can be summarized at a high level as approximately solving the following layered optimization problem:
\begin{itemize}
    \item \emph{Inner optimization: E-value-based optimization.} For any given time $t$, we can maximize the diversity measure $\varphi(\bZ^{\test}_{\mathcal{R}})$ among all selection sets $\mathcal{R}$ that are self-consistent with respect to the time $t$ e-values $e^{(t)}_1, \dots, e^{({t})}_m$. The resulting maximum diversity value is denoted as $O_t$. 
    \item \emph{Outer optimization: Optimal stopping.} We search for the optimal stopping time $\tau^*$ that (approximately) maximizes $\bE[O_{\tau}]$, the expected diversity of the maximally diverse self-consistent selection set at time $\tau$, among all stopping times $\tau$. We achieve this by applying tools from optimal stopping theory.
\end{itemize}

Once the optimal stopping time $\tau^*$ has been computed, the final step of our procedure is to run an optimization procedure which maximizes diversity (again, as measured by $\varphi$) subject to constraint~\eqref{sc} to determine the optimally diverse selection set that is self-consistent with respect to $e^{(\tau^*)}_1, \dots, e^{(\tau^*)}_m$.

\subsection{Making diverse selections with DACS}\label{precise-method}
Throughout this section and for the remainder of the paper, we adopt a convention of referring to time in decreasing index order, reflecting the fact that we work with a \emph{backwards} filtration $\left(\mathcal{F}_t\right)_{t=1}^{n+m}$: specifically, we will say that ``time $t$ is after time $s$'' if $t \leq s$.

\subsubsection{Finding the optimal stopping time}

Because DACS prunes the CS selection set into a diverse subset, we restrict our search for the optimal stopping time $\tau^*$ to only those stopping times which stop \emph{after} the BH stopping time $\tau_{\bh}$. This is the optimal stopping problem \citep{chowgreat} of maximizing $\mathbb{E}[O_{\tau}]$ across all stopping times $\tau \leq \tau_{\bh}$. 
Should the distribution of data (more precisely, the joint distribution of the indicator variables $(B_1, \ldots, B_{\tau_{\bh}})$ conditional on the stopped filtration $\mathcal{F}_{\tau_{\bh}}$) be known, the optimal stopping time $\tau^*$ that maximizes $\mathbb{E}[O_{\tau}]$ can be found in two steps: (1) computing  $\mathbb{E}[O_t\,|\,\mathcal{F}_t]$ for each $t\leq \tau_{\bh}$ and (2) constructing a sequence of random variables known as the \emph{Snell envelope} \citep{snell1952applications} via  dynamic programming.

However, the conditional distribution $(B_1, \ldots, B_{\tau_{\bh}}) \mid \mathcal{F}_{\tau_{\bh}}$ is not known in general, and this in turn impedes our ability to compute certain conditional expectations (such as $\mathbb{E}[O_t\,|\,\mathcal{F}_t]$) which are necessary to perform steps (1) and (2) above. This is because the observed calibration scores use the true outcomes $V_i=V^{\text{clip}}(X_i,Y_i)$ whereas the test scores are imputed: $\widehat{V}_i=V^{\text{clip}}(X_i,0)$. While the imputation is crucial for FDR control, it renders the calibration and test scores non-exchangeable, thus making the conditional distribution $(B_1, \ldots, B_{\tau_{\bh}}) \mid \mathcal{F}_{\tau_{\bh}}$ unknown. 
To address this issue, we will perform our optimal stopping calculations under a distribution which approximates the true distribution $\bP$. In particular, we approximate the true data-generating distribution $\bP$ by (any) distribution, denoted $\bP_{\exch}$, for which \[\bP_{\exch}\left(B_1 = b_1, \ldots, B_t=b_t, t \leq \tau_{\bh} \mid \mathcal{F}_{t}\right) \overset{a.s.}{=} \bP_{\exch}\left(B_1 = b_{\pi(1)}, \ldots, B_t=b_{\pi(t)}, t \leq \tau_{\bh} \mid \mathcal{F}_t\right),\] for any permutation $\pi$ of $[t]$ and any $t \in [n+m]$. In words, $\bP_{\exch}$ makes all indicators past the BH stopping time conditionally exchangeable.
While this conditional exchangeability does not generally hold under the true data-generating distribution, it will be approximately true in realistic situations, for instance when $\hat{\mu}(X) \overset{d}{\approx} \hat{\mu}(X) \mid Y \leq 0$, since then the non-infinite calibration scores (which are the only calibration scores that can appear past the BH stopping time) will come from approximately the same distribution as the test scores. Such an approximate equality in distribution occurs if $\bP(Y \leq 0) \approx 1$. This corresponds to the realistic and challenging situation in which most candidates are nulls, which is likely to be the case in the drug discovery and hiring applications that we consider.
We find in our empirical studies in Sections~\ref{expers} and \ref{sims} that $\bP_{\exch}$ serves as a good approximation in many other scenarios as well; in particular, DACS continues to produce selection sets with substantial diversity even when $\bP(Y \leq 0)$ is not close to one. Finally, while our optimal stopping computations are performed under a distribution that differs from the true one, this discrepancy impacts only the optimality (i.e., the achieved diversity) of the resulting selection set. Crucially, as we will soon see, it \emph{does not} compromise FDR control.

We now explain, at a high level, how to find $\tau^*$ by performing the two steps discussed in the beginning of this section with respect to the exchangeable measure $\bP_{\exch}$. 
We will begin at the BH stopping time $\tau_{\bh}$ and condition on all information contained in the filtration at that time, i.e. $\mathcal{F}_{\tau_{\bh}}$ (recall that we enforce $\tau^* \leq \tau_{\bh}$ and so conditioning on $\mathcal{F}_{\tau_{\bh}}$ amounts to conditioning on all available information; see Remark~\ref{global-conditioning} below for details). Operating conditionally on $\mathcal{F}_{\tau_{\bh}}$, we then follow a standard approach to finding the optimal stopping time \citep[][Chapters 1 and 2]{ferguson2006optimal} which we recapitulate here.

Letting $\bE_{\exch}$ denote expectation with respect to the exchangeable distribution $\bP_{\exch}$, the first step to compute $\tau^*$ is to define \emph{reward} variables that are adapted to the reverse filtration $\left(\mathcal{F}_t\right)_{t=1}^{n+m}$ \begin{equation}\label{reward-construction}
    R_t := \bE_{\exch}[O_t \mid \mathcal{F}_t] \text{ for } t = 1, \ldots, \tau_{\bh}.
\end{equation}

Using the rewards $R_t$, we then inductively define: \begin{equation}\label{snell}
    E_1 := R_1 \text{ and } E_t := \max\left(R_t, \bE_{\exch}[E_{t-1} \mid \mathcal{F}_t]\right) \text{ for } t= 2, \ldots, \tau_{\bh}.
\end{equation}
The sequence $(E_1, \ldots, E_{\tau_{\bh}})$ is called the Snell envelope.
The optimal stopping time is then given by \begin{equation}\label{opt-stopping-time-def}
    \tau^* := \max\{t \in [\tau_{\bh}]: R_t \geq E_t\},
\end{equation} the first time $t$ after the BH stopping time at which $R_t$ is at least $\bE_{\exch}[E_{t-1} \mid \mathcal{F}_t]$, the optimal value (in expectation) if one were to continue on past time $t$. After $\tau^*$ has been computed, we solve the e-value optimization problem of maximizing diversity subject to self-consistency with respect to the e-values $e_1^{(\tau^*)}, \ldots, e_m^{(\tau^*)}$ to obtain the selection set $\mathcal{R}^*_{\tau^*}$. As shown in \cite{gablenz2024catch}, this optimization program can be written as an integer program with linear constraints; in Section~\ref{dacs-underrep} we show how to solve it efficiently for the underrepresentation index and in Section~\ref{dacs-heuristics} we show how it can be relaxed for a generic diversity metric. Figure~\ref{fig:overview} gives a high-level overview of our procedure and the following corollary to Theorems~\ref{stopping-time-thm} and \ref{thm:self-consistency-validity} establishes its validity:

\begin{corollary}\label{corollary:dacs-validity}
    The selection set $\mathcal{R}^*_{\tau^*}$ returned by DACS enjoys FDR control at level $\alpha$.
\end{corollary}

\begin{figure}
    \centering
    \includegraphics[scale=0.9]{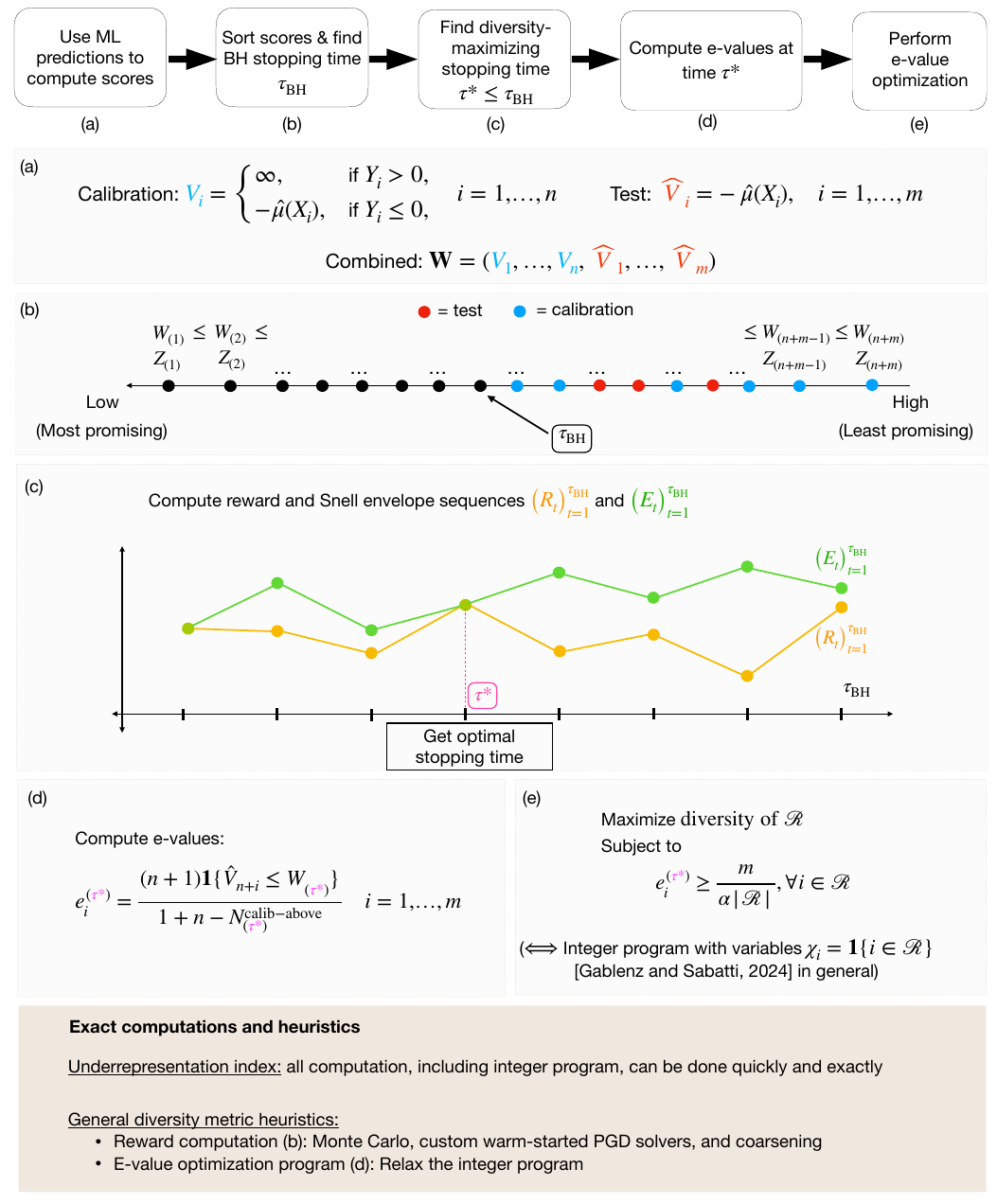}
    \caption{High-level overview of DACS procedure. Flowchart illustrates the five main steps of our method, with panels (a)--(e) detailing each step. The most computationally intensive parts are the reward computation (panel c) and e-value optimization (panel e). As shown in the brown panel, both steps can be performed exactly and efficiently for the underrepresentation index; for general diversity metrics, we introduce heuristics to accelerate computation.}
    \label{fig:overview}
\end{figure}

\subsubsection{Reward and Snell envelope computation under the exchangeable distribution}\label{reward-snell-section}
In this section, we explain in detail how to compute both the reward and Snell envelope sequences involved in panel (c) of Figure~\ref{fig:overview}. We begin with a remark regarding conditioning.

\begin{remark}\label{global-conditioning}
The conditional expectations in displays~\eqref{reward-construction} and \eqref{snell} are all with respect to $\sigma$-algebras that contain $\mathcal{F}_{\tau_{\bh}}$. This is owed to the fact that the optimal stopping time defined by equation~\eqref{opt-stopping-time-def} occurs no earlier than $\tau_{\bh}$, and thus all calculations required to find it may treat the information contained in $\mathcal{F}_{\tau_{\bh}}$ as fixed.
Importantly, this allows us to condition on the sequence $\Zsort$ and treat it as fixed.
\end{remark}

Notice that both $R_t$ and $E_t$ are $\mathcal{F}_t$-measurable. Because we view $\Zsort$ as fixed, per Remark~\ref{global-conditioning}, this means that the only randomness in $R_t$ and $E_t$ is through the randomness in $B_{t+1}, \ldots, B_{n+m}$. The exchangeability guaranteed under $\bP_{\exch}$ ensures that, for $t \leq \tau_{\bh}$, $R_t$ and $E_t$ do not depend on the ordering of $B_{t+1}, \ldots, B_{n+m}$ and hence depend \emph{only} on $N_t^{\Above} = B_{t+1} + \cdots + B_{n+m}$. To make this dependence explicit, we adopt the function notation $R_t := R_t(N_t^{\Above})$ and $E_t := E_t(N_t^{\Above})$. The Snell envelope update in equation~\eqref{snell} is then performed pointwise over the domain of these functions---that is, it requires evaluating $R_t(s_t)$ and $E_t(s_t)$ for each $s_t$ in the support of the random variable $N_t^{\Above}$. 
As such, it is important to characterize the support of $N_t^{\Above}$.

\begin{proposition}\label{prop:support}
    For each $t = 1, \ldots, \tau_{\bh}$, define the sets
    \[\Omega_t := \left[\max\big(N^{\Above}_{\tau_{\bh}}, n-t\big) : \min\big(n, \tau_{\bh}-t+N_{\tau_{\bh}}^{\Above}\big)\right].\]

    Then $\Omega_t$ is precisely the support of $N_t^{\Above}$ conditional on $\mathcal{F}_{\tau_{\bh}}$ under $\bP_{\exch}$.
\end{proposition}

According to Proposition~\ref{prop:support}, whose proof is given in Appendix~\ref{appendix:support}, it suffices to compute only the values $R_t(s_t)$ and $E_t(s_t)$ for $s_t\in \Omega_t$. We focus on the $\mathcal{F}_{\tau_{\bh}}$-\emph{conditional} supports of the variables $N_t^{\Above}$ since Remark~\ref{global-conditioning} permits us to condition on $\mathcal{F}_{\tau_{\bh}}$.

Turning first to reward computation, the following proposition provides an explicit formula to compute  $R_t(s_t)$ under $\bP_{\exch}$; see Appendix~\ref{appendix:reward-prob} for a proof.

\begin{proposition}\label{global-null-prop}
\begin{enumerate}[(i)]
    \item For any $\bb \in \{0,1\}^t$, define 
    $\esort_i(\bb) := \frac{(1-b_{i})(n+1)}{1+b_1+\cdots+b_t}$ for each $i \in [t]$. In other words, if $\bb \in \{0,1\}^t$ denotes the indicator vector of which, among the smallest $t$ scores, belong to the calibration set, then 
    $\esort_i(\bb)$ is the time $t$ e-value corresponding to the $i^{\text{th}}$ sorted score.

Viewing $\bZ^{()}$ as fixed, per Remark~\ref{global-conditioning}, the optimal value $O_t$ depends only on $(B_1, \ldots, B_t)$. Formally, for any $\bb\in \{0,1\}^t$, when $(B_1,\dots,B_t)=\bb$, we have $O_t=O_t(\bb)$ where \begin{equation}\label{go-opt-def}O_t(\bb) = \max\left\{\varphi\big(\Zsort_{\mathcal{R}}\big): \mathcal{R} \subseteq [t] \text{ is self-consistent w.r.t.~}\esort_1(\bb), \ldots, \esort_t(\bb)\right\}.\end{equation}

 \item For each $t \in [\tau_{\bh}]$ and $s \in [n-t:n]$ define $\mathcal{B}(t,s) := \{\bb \in \{0,1\}^t: b_1 + \cdots + b_t = n-s\}$ to be the set of Boolean vectors of length $t$ summing to $n-s$ (this is precisely the set of possible values of $(B_1, \ldots, B_t)$ given that $N_t^{\Above} = s$) and let $C(t,s) := \big|\mathcal{B}(t,s)\big| = \binom{t}{n-s}$ be the cardinality of $\mathcal{B}(t,s)$. The reward functions $R_t(s_t)$ are given by $R_1(s_1) = O_1(n-s_1)$ and \begin{equation}\label{mble-rewards}R_t(s_t) = C(t,s_t)^{-1}\sum_{\bb \in \mathcal{B}(t,s_t)} O_t(\bb) \text{ for }s_t \in \Omega_t \text{ and }t = 2, \ldots, \tau_{\bh}.\end{equation}
\end{enumerate}

\end{proposition}

We now show the computation required to compute the Snell envelope; see Appendix~\ref{appendix:snell-construction} for a proof of Proposition~\ref{snell-construction-prop}.

\begin{proposition}\label{snell-construction-prop}
    The function $E_1(\cdot)$ is given by $E_1(s_1) = R_1(s_1)$ for each $s_1 \in \Omega_1$. Furthermore, for each $t = 2, \ldots, \tau_{\bh}$ the functions $E_t(\cdot)$ are inductively given by \begin{equation}\label{envelope-construction}
        E_t(s_t) = \max\left\{ R_t(s_t), \frac{t-n+s_t}{t}\cdot E_{t-1}(s_t) + \frac{n-s_t}{t} \cdot E_{t-1}(s_t+1)\right\} \quad \text{for each } s_t \in \Omega_t.\footnote{If $s_t = n - t$ or $s_t = n$, then $s_t \notin \Omega_{t-1}$ or $s_t + 1 \notin \Omega_{t-1}$, respectively. Consequently, either $E_{t-1}(s_t)$ or $E_{t-1}(s_t + 1)$ is not well-defined in these cases. However, this does not pose an issue, as the probability weights associated with these terms will be zero in either case.}
    \end{equation}
\end{proposition}

Equation~\eqref{envelope-construction} is precisely the inductive Snell envelope update given in~\eqref{snell}, but written for each element of the conditional support of $N_t^{\Above}$ given $\mathcal{F}_{\tau_{\bh}}$. Performing this update is computationally fast. It is worth noting that this is in quite sharp contrast with most optimal stopping problems, in which it is generally computationally infeasible to construct the Snell envelope since the supports of the variables $E_t$ are typically either infinite or grow exponentially with $t$ (see, e.g., \cite{del2011robustness}).

\begin{algorithm}[!ht]
  
  \KwInput{BH stopping time $\tau_{\bh}$, number of calibration points above BH stopping time $N_{\tau_{\bh}}^{\Above}$}
 \For{$t = 1, \ldots, \tau_{\bh}$}{
  Compute $R_t(s_t)$ for each pair $s_t \in \Omega_t$ using equation~\eqref{mble-rewards}
 }
  Set $E_1(s_1) \gets R_1(s_1)$ for each pair $s_1 \in \Omega_1$\\
  \For{$t = 2, \ldots, \tau_{\bh}$}{
  Compute $E_t(s_t)$ for each $s_t \in \Omega_t$ using equation~\eqref{envelope-construction}
  }
  Set $t \gets \tau_{\bh}$\\
  \While{$R_t(N_t^{\Above}) < E_t(N_t^{\Above})$}{
  Set $t \gets t-1$
  }
  Set $\tau^* \gets t$\\
  Compute e-values $e^{(\tau^*)}_1, \ldots, e^{(\tau^*)}_m$ and maximize $\varphi(\bZ^{\test}_{\mathcal{R}})$ across all selection sets $\mathcal{R} \subseteq [m]$ for which \[e_i^{(\tau^*)} \geq \frac{m}{\alpha |\mathcal{R}|} \text{ for all }i \in \mathcal{R},\] letting $\mathcal{R}^*_{\tau^*}$ denote the optimal feasible set \\
  \KwOutput{ Selected set $\mathcal{R}^*_{\tau^*}$} 

\caption{Diversity-aware conformal selection}
\label{dacs-algo}
\end{algorithm}

We give pseudocode for our entire procedure in Algorithm~\ref{dacs-algo} and end this section by simply stating its optimality under the exchangeable distribution, which is a direct consequence of the optimality of the Snell-envelope-based stopping time, e.g., \citet[][Theorem 3.2]{chowgreat}.
\begin{theorem}\label{optimality}
    The rejection set $\mathcal{R}^*_{\tau^*}$ returned by DACS is optimally diverse in expectation under $\bP_{\exch}$ among stopping-time-based rejection sets. That is, if $\tau$ is any $\left(\mathcal{F}_t\right)_{t=0}^{n+m}$-stopping time and $\mathcal{R}_{\tau}$ is any self-consistent selection set with respect to the e-values $e^{(\tau)}_1, \ldots, e^{(\tau)}_m$, then $\bE_{\exch}[\varphi(\bZ^{\test}_{\mathcal{R}^*_{\tau^*}})] \geq \bE_{\exch}[\varphi(\bZ^{\test}_{\mathcal{R}_{\tau}})].$
\end{theorem}

\section{Implementation of DACS: exact computation and heuristics}
\label{sec:implement}

The main computational challenges of DACS are in panels (c) and (e) of Figure~\ref{fig:overview}, specifically: (1) computing the rewards via equation~\eqref{mble-rewards} (i.e., lines 1--2 of Algorithm~\ref{dacs-algo}), and (2) performing the e-value optimization based on the e-values $e^{(\tau^*)}_1, \ldots, e^{(\tau^*)}_m$ (i.e., line 10 of Algorithm~\ref{dacs-algo}).
Accordingly, our focus in this section is restricted to the computational aspects of items (1) and (2) above. Section~\ref{dacs-underrep} presents an efficient computation strategy for the underrepresentation index, while Section~\ref{dacs-heuristics} introduces heuristics to accelerate computation for general diversity metrics, like the Sharpe ratio and Markowitz objective.

\subsection{Exact computation for the underrepresentation index}\label{dacs-underrep}

\subsubsection{Computing the rewards}\label{underrep-reward}
The first step towards efficient computation of the rewards (panel (c) of Figure~\ref{fig:overview}) via equation~\eqref{mble-rewards} is to observe that the optimal values $O_t$ for the underrepresentation index admit a closed form solution. This is formalized in the following proposition, whose proof can be found in Appendix~\ref{app:underrep-opt-value-closed-form}.

\begin{proposition}\label{underrep-opt-value-closed-form}
    For each $c \in [C]$, define $N_t^{c, \testbelow} := \sum_{s=1}^t \mathds{1}\left\{B_s = 0, Z_s = c\right\}$ to be the number of test points past time $t$ belonging to category $c$. Then \begin{equation}\label{eq:underrep-reward-closed-form}
        O_t = \begin{cases}
            \frac{1}{C}, & \text{ if } \mathop{\min}_{c\in [C]}  N_t^{c, \testbelow} \geq \frac{K_t}{C},\\
             -\frac{1}{C}, & \text{ if } n-N_t^{\Above} > \rho_t,\\
             \frac{\mathop{\min}_{c\in [C]} N_t^{c, \testbelow}}{K_t}, & \text{ otherwise},
        \end{cases}
    \end{equation} where $K_t := \left\lceil \frac{m(1+n-N_t^{\Above})}{\alpha (n+1)} \right\rceil$ is $\mathcal{F}_t$-measurable and $\rho_t := \frac{\alpha t(n+1) - m}{\alpha(n+1)+m}$ is deterministic.
\end{proposition}

By the exchangeability guaranteed under $\bP_{\exch}$, the random variables $(N_t^{c, \testbelow})_{c \in [C]}$ are, conditional on $\mathcal{F}_t$, a sample from a multivariate hypergeometric distribution for each $t \leq \tau_{\bh}$. Since equation~\eqref{eq:underrep-reward-closed-form} depends only on the minimum of these variables, it suffices to study the distribution of the minimum of a multivariate hypergeometric sample in order to compute the reward functions $R_t(\cdot)$. As shown in~\cite{lebrun2013efficient}, the survival function of this minimum can be computed exactly and efficiently using a fast Fourier transform-based approach; as we show in Appendix~\ref{appendix:underrep-additional-reward}, this enables efficient calculation of the rewards.

\subsubsection{E-value optimization}\label{underrep-evalue}
To efficiently implement DACS for the underrepresentation index, it remains to efficiently perform the e-value optimization in line 10 of Algorithm~\ref{dacs-algo}. The key idea is that the self-consistency constraint~\eqref{sc} for the e-values $e^{(\tau^*)}_i$, $i = 1,\ldots, m$, is equivalent to two simpler conditions: (1) the selection set $\mathcal{R}^*_{\tau^*}$ must only include candidates whose sorted score occurs after time $\tau^*$, and (2) the size of the selection set must be at least a certain budget. Specifically, a direct calculation shows that $\mathcal{R} \subseteq [m]$ is self-consistent with respect to the e-values $\big(e^{(\tau^*)}_i\big)_{i \in [m]}$ if and only if
\begin{equation}\label{sc'}
    \widehat{V}_i \leq W_{(\tau^*)} \quad \text{for all } i \in \mathcal{R} \quad \text{and} \quad |\mathcal{R}| \geq K_{\tau^*},
\end{equation}
where $K_t$ is defined in Proposition~\ref{underrep-opt-value-closed-form}.
Thus, the e-value optimization reduces to finding a subset $\mathcal{R}$ of test candidates with scores below $V_{(\tau^*)}$, of size at least $K_{\tau^*}$, that maximizes the underrepresentation index.

Algorithm~\ref{underrep-algo} in Appendix~\ref{app:underrep-evalue-opt} gives pseudocode for an efficient procedure to compute $\mathcal{R}_{\tau^*}^*$ and Proposition~\ref{appendix:e-value-underrep-optimality} in the same Appendix section establishes its correctness. We summarize the main idea of the algorithm here. 
As notation, define for each $c \in [C]$ the set $\mathcal{S}_{\tau^*}^{c, \testbelow} := \big\{i \in [m]: \widehat{V}_i \leq W_{(\tau^*)} \text{ and } Z_{n+i} = c\big\}$ of test candidates, past time $\tau^*$, which belong to category $c$. Furthermore, let $\mathcal{S}_{\tau^*}^{(1), \testbelow},$ $ \ldots,$ $ \mathcal{S}_{\tau^*}^{(C), \testbelow}$ denote this list of sets sorted by increasing size.

 \paragraph{Greedy algorithm} The algorithm proceeds as follows:
\begin{itemize}
    \item If $n - N_{\tau^*}^{\Above}$ is sufficiently large, then no non-empty set is self-consistent (see the proof of Proposition~\ref{underrep-opt-value-closed-form}), and the algorithm returns the empty set.
    \item Otherwise, it attempts to select $|\mathcal{S}_{\tau^*}^{(1), \testbelow}|$ candidates from each $\mathcal{S}_{\tau^*}^{(c), \testbelow}$ for $c \in [C]$, provided doing so satisfies the budget constraint $|\mathcal{R}| \geq K_{\tau^*}$.
    \item If this is not possible, the algorithm grows the selection set $\mathcal{R}_{\tau^*}^*$ sequentially. At each step $c=1,\dots,C$, it checks whether selecting enough candidates from the remaining sets $\mathcal{S}_{\tau^*}^{(c), \testbelow}, \ldots, \mathcal{S}_{\tau^*}^{(C), \testbelow}$ can meet the budget lower bound with equality and without changing the least-represented category. If so, it adds the selections to $\mathcal{R}_{\tau^*}^*$ and terminates. If not, it adds all candidates from $\mathcal{S}_{\tau^*}^{(c), \testbelow}$ to $\mathcal{R}_{\tau^*}^*$ and proceeds to the next step.
\end{itemize}
Further details can be found in Appendix~\ref{app:underrep-evalue-opt}.

\subsection{Monte Carlo and heuristics: relaxation, warm starting, and coarsening}\label{dacs-heuristics}
For a generic user-specified diversity metric $\varphi$, it is nontrivial to efficiently compute the reward functions $R_t(\cdot)$ and to perform the e-value optimization involved in panels (c) and (e) of Figure~\ref{fig:overview}. In general, to compute the average defining the reward $R_t(\cdot)$ in equation~\eqref{mble-rewards}, one will need to compute the optimal value function $O_t(\bb)$ for combinatorially many values of $\bb$. In turn, computing even a single value $O_t(\bb)$ typically involves solving an e-value optimization program similar to the final one in line 10 of Algorithm~\ref{dacs-algo}.

Following \cite{gablenz2024catch}, we first rewrite the e-value optimization program as a linearly constrained integer program (LCIP). More precisely, recall that the optimization problem in line 10 of Algorithm~\ref{dacs-algo} is to maximize $\varphi(\bZ^{\test}_{\mathcal{R}})$ across all $\big(e_i^{(\tau^*)}\big)_{i \in [m]}$-self-consistent selection sets $\mathcal{R}$. \cite{gablenz2024catch} show that the optimal solution can be written as $\mathcal{R}^*_{\tau^*} = \{i \in [m]: \chi^*_i = 1\}$ where $\bchi^*$ is the solution to the following LCIP:
\begin{equation}\label{program:e-val}
\begin{split}
    \text{maximize} \quad & \varphi(\bchi; \bZ^{\test}) \\
    \text{s.t.} \quad & \chi_i \leq \frac{\alpha e_i^{(\tau^*)}}{m}\sum_{j=1}^m\chi_j, \quad i = 1, \dots, m \\
    & \chi_i \in \{0,1\}, \quad i = 1, \dots, m
\end{split}
\end{equation} 
 where we overload the notation $\varphi(\bchi; \bZ^{\test})$ to mean $\varphi\big(\bZ^{\test}_{\left\{i: \chi_i = 1\right\}}\big)$. In Algorithm~\ref{dacs-algo}, the program~\eqref{program:e-val} must be solved to return the final output in line 10 and, in general, similar programs must be solved to calculate the optimal values $O_t(\bb)$ used for the reward computation in line 2. As such, using DACS ``out-of-the-box'' for a generic diversity metric $\varphi$ will require solving combinatorially many LCIPs which is, of course, computationally prohibitive. 

In this part, we introduce several strategies to lighten the computational load, including (i) replacing the average in equation~\eqref{mble-rewards} by a Monte Carlo (MC) approximation, (ii) relaxing the integer constraints in~\eqref{program:e-val},  (iii) implementing custom warm-started projected gradient descent (PGD) solvers to further speed up MC approximation, and (iv) coarsening the dynamic programming to allow for fewer reward computations. These strategies complete the design of DACS and we use them in all of our numerical experiments involving the Sharpe ratio and Markowitz objective.

 \subsubsection{Monte Carlo approximation}\label{mc-approx}
 Our first strategy is to simply approximate the averages which define $R_t(\cdot)$ in equation~\eqref{mble-rewards} via MC, eliminating the need to solve combinatorially many LCIPs. For any $s_t$, we approximate $R_t(s_t)$ by 
 $$
 R^{\mc}_t(s_t) := L^{-1} \sum_{\ell=1}^{L} O_t(\bb^{(\ell)}),\quad  \bb^{(1)}, \ldots, \bb^{(L)} \overset{\text{i.i.d.}}{\sim}\text{Unif}\left(\mathcal{B}(t, s_t)\right),
 $$
 where, as defined in Proposition~\ref{global-null-prop}, $\mathcal{B}(t,s_t)$ is the entire set of values that $(B_1,\dots,B_t)$ can possibly take on given that $N_t^{\Above}=s_t$. 
 We then use the MC-approximated reward functions to compute the Snell envelope sequence $(E^{\mc}_t(\cdot))_{t \in [\tau_{\bh}]}$ just as in equation~\eqref{envelope-construction}, and the stopping time is computed as usual: 
$\tau^{*} := \max\left\{t \in [\tau_{\bh}]: R^{\mc}_t \geq E^{\mc}_t\right\}.$
Since $\tau^{*}$ remains a stopping time (for a fixed MC random seed), the resulting e-values $e_1^{(\tau^{*})}, \ldots, e_m^{(\tau^{*})}$ are still valid by Theorem~\ref{stopping-time-thm}, preserving FDR control. The approximate optimality guarantees of our method, of course, will depend on the accuracy of the MC approximation, and in particular on the size of $L$.
 
 While MC approximation reduces the number of LCIPs that need to be solved, it is still important to be able to solve each LCIP efficiently, since a modest number of MC samples will be needed to ensure that our approximation is of reasonable quality. The remainder of this section is devoted to heuristics which greatly speed up this computation.

 \subsubsection{Relaxation and randomized rejection}\label{lp-relax}
Solving an LCIP, let alone an integer linear program (ILP), is in general NP complete \citep{schrijver1998theory}. 
We utilize a commonly used heuristic, which is to relax the integer constraint to an interval constraint (i.e., two inequality constraints) \citep{boyd2004convex,schrijver1998theory}. 
Because the solution to this relaxed program may no longer be a binary vector indicating which candidates to and not to select, we develop a randomized selection approach, based on the non-integer solution, to obtain a final selection set which enjoys approximate FDR control.

Allowing for non-integer solutions to the program~\eqref{program:e-val} necessitates that the function $\bchi \mapsto \varphi(\bchi; \bZ^{\test})$, defined only for $\bchi \in \{0,1\}^m$, be extendable to any $\bchi \in [0,1]^m$.
Below, we concretely show how this extension works for the Sharpe ratio and Markowitz objective. Throughout, it will be helpful to define $\Sigma^{\test}$ to be the bottom right $m \times m$ submatrix of the full similarity matrix $\Sigma$, denoting just the pairwise similarities between the test diversification variables $
\bZ^{\test}$.

\begin{example}[Extending Sharpe ratio and Markowitz diversity objectives]
    The Sharpe ratio and Markowitz diversity objectives are given by $\varphi^{\sharpe}\left(\bchi; \bZ^{\test}\right) = \begin{cases}
            \frac{\bchi^\top \mathbf{1}}{\sqrt{\bchi^\top \Sigmatest\bchi}}, & \text{ if } \bchi \neq \bzero,\\
            0, & \text{ if } \bchi = \bzero,
        \end{cases}$ and $\varphi^{\markowitz}\left(\bchi; \bZ^{\test}\right) = \bchi^\top \mathbf{1} - \frac{\gamma}{2} \cdot \bchi^\top \Sigmatest \bchi$ for $\bchi \in \{0,1\}^m$. These functions are immediately extendable to $[0,1]^m$ as the equations defining them continue to be well-defined for any $\bchi \in [0,1]^m$.
\end{example}

For the remainder of this section, when we write $\varphi$ we refer to the extension of the diversity metric to the unit cube $[0,1]^m$. Using this extension, the relaxation of program~\eqref{program:e-val} is a linearly constrained optimization problem:
\begin{align}
    \label{relaxedsc0} \text{maximize} \quad & \varphi(\bchi; \bZ^{\test})\\
    \label{relaxedsc1} \text{s.t.} \quad & \chi_i \leq \frac{\alpha e^{(\tau^*)}_i}{m}\sum_{j=1}^m\chi_i, \quad i = 1, \dots, m \\
    &\label{relaxedsc2} 0 \leq \chi_i \leq 1, \quad i = 1, \dots, m
\end{align}

We call a vector satisfying constraints~\eqref{relaxedsc1}--\eqref{relaxedsc2} \emph{relaxed self-consistent} (RSC) and refer to the entire program~\eqref{relaxedsc0}--\eqref{relaxedsc2} as a \emph{relaxed} e-value optimization program. The solution $\bchirelaxedt$ to this program will not, in general, be integer-valued (indeed, in our simulations involving the Sharpe ratio and Markowitz objective, sometimes even $90\%$ or more of its elements are non-binary) and therefore, to construct the final selection set $\rrelaxedt$ we randomly select with independent Bernoulli probability as specified by $\bchirelaxedt$. That is, $\rrelaxedt := \{i: \xi_i = 1\}$ where $\xi_i \overset{\text{ind}}{\sim} \text{Bern} (\chirelaxedt_i )$, for $i = 1, \ldots, m$. The resulting selection set  $\rrelaxedt$ enjoys approximate FDR control:

\begin{proposition}\label{relaxed-fdr}
    The FDR of $\rrelaxedt$ is at most $1.3\alpha$.
\end{proposition}

Proposition~\ref{relaxed-fdr} is a special case of Proposition~\ref{appendix:rsc-prop} in Appendix~\ref{appendix:rsc} which establishes the same approximate FDR control guarantee for any selection set constructed via randomized Bernoulli sampling according to any vector $\bchi \in [0,1]^m$ that is RSC with respect to any set of valid e-values. Furthermore, the same FDR bound holds true for standard deterministic e-value-based multiple hypothesis testing; see Appendix~\ref{appendix:rsc-det} for a precise statement and proof. We note that the upper bound $1.3\alpha$ is often conservative. 
Indeed, in our empirical studies in Sections~\ref{expers} and \ref{sims}, the FDR of $\rrelaxedt$ \emph{never} exceeds the nominal level $\alpha$.

Proposition~\ref{relaxed-fdr} allows us to replace the final LCIP in line 10 of Algorithm~\ref{dacs-algo} (panel (e) of Figure~\ref{fig:overview}) by the relaxed e-value optimization program, yielding the selection set $\rrelaxedt$. Because we construct the final selection set $\rrelaxedt$ via relaxed e-value optimization, it no longer makes sense to define the optimal value function $O_t(\cdot)$ in equation~\eqref{go-opt-def} as the solution to a (non-relaxed) e-value optimization problem (which, as discussed the beginning of this subsection, is difficult to compute). Instead, we work with a \emph{relaxed} optimal value function $\widehat{O}_t(\cdot)$, which represents the expected diversity of the randomized selection set based on the relaxed e-value optimization problem at time $t$. In particular, recalling that $\varepsilon_i^{(t)}(\bb)$ is the e-value of the $i^{\text{th}}$ sorted score when the membership of the sorted scores in the calibration fold is given by the indicator vector $\bb$, we define: \begin{equation}\label{relaxed-go-def}\widehat{O}_t(\bb) := \bE_{\xi_i \overset{\text{ind}}{\sim} \text{Bern}(\chi^*_i),\\ i \in [t]}\left[\varphi\big(\boldsymbol{\xi}; \Zsort_{1:t}\big)\right]\text{ where } \bchi^* = \argmax_{\bchi \in [0,1]^t: \, \bchi \text{ is RSC w.r.t.~} (\esort_i(\bb))_{i=1}^t} \varphi\big(\bchi; \Zsort_{1:t}\big)
.\end{equation} This definition mirrors that of $O_t(\bb)$ in equation~\eqref{go-opt-def}, but with the selection vector $\bchi$ now only required to satisfy \emph{relaxed} self-consistency. 
In general, the expectation in~\eqref{relaxed-go-def} can be approximated using MC (analogously to Section~\ref{mc-approx}, doing so will not affect the FDR control guarantee of Proposition~\ref{relaxed-fdr})---this is what we do for the Sharpe ratio---though in some cases an analytical formula exists, as is the case for the Markowitz objective.

The final ingredient is to ensure that the extension $\varphi$ is concave, or at least that the program~\eqref{relaxedsc0}--\eqref{relaxedsc2} is equivalent (e.g., up to a change-of-variables) to a convex optimization program. Below, we show that this is indeed true for the Sharpe ratio and Markowitz objective.

\begin{example}[Sharpe ratio and Markowitz objective, continued]\label{ex:relaxed-programs-sharpe-markowitz}
    The Markowitz objective is concave and hence the corresponding relaxed e-value optimization program~\eqref{relaxedsc0}--\eqref{relaxedsc2} is a quadratic program. As for the Sharpe ratio, a standard trick---see, e.g., \citet[][pp.~101--103]{cornuejols2018optimization}---shows that any solution to the quadratic program
\begin{align}
     \label{sharpe0}\text{minimize} \quad & \bchi^\top \Sigma^{\test}\bchi\\
     \label{sharpe1}\text{s.t.} \quad & \chi_i \leq \frac{\alpha e_i^{(\tau^*)}}{m}\sum_{j=1}^m\chi_i, \quad i = 1, \dots, m \\
      \label{sharpe2}  & \bchi^\top \mathbf{1} = 1 \\
     \label{sharpe3}   & 0 \leq \chi_i \leq 1, \quad i = 1, \dots, m
\end{align} is also a solution to the relaxed e-value optimization program~\eqref{relaxedsc0}--\eqref{relaxedsc2} for the Sharpe ratio.\footnote{If the program~\eqref{sharpe0}--\eqref{sharpe3} is infeasible, we simply return $\bchi^* = \boldsymbol{0}$. Additionally, there will generally be infinitely many solutions to the Sharpe ratio e-value optimization program. When constructing the final selection set in line 10 of Algorithm~\ref{dacs-algo}, we return the solution which yields the largest (expected) number of selections by
replacing any solution $\bchi^* \neq \mathbf{0}$ to~\eqref{sharpe0}--\eqref{sharpe3} with $\frac{\bchi^*}{\|\bchi^*\|_\infty}$ which is also an optimal solution,
but yields a greater expected number of selections.} Hence, solving the relaxed e-value optimization programs for both the Sharpe ratio and Markowitz objectives requires only to solve quadratic programs.
\end{example}

  Using this relaxed e-value optimization heuristic, computing all the reward function values in panel (c) of Figure~\ref{fig:overview} requires solving $L$ convex optimization programs for each $s_t \in \Omega_t, t \in [\tau_{\bh}]$; solving the final relaxed e-value optimization program in panel (e) of Figure~\ref{fig:overview} requires solving only one more convex optimization problem.

 \subsubsection{Warm-started PGD solvers using coupled Monte Carlo samples}\label{warm-start}
The MC approximation and relaxation of the integer constraint given in Sections~\ref{mc-approx} and \ref{lp-relax} enable us to compute the \emph{relaxed MC} approximate reward function $\widehat{R}^{\mc}_t(\cdot)$ as an MC average of the relaxed-self-consistent optimal values: $\widehat{R}^{\mc}_t(s_t) := L^{-1} \sum_{\ell=1}^{L} \widehat{O}_t(\bb^{(\ell)})$. This requires computing, for each  $s_t \in \Omega_t, t \in [\tau_{\bh}],$ and $\ell \in [L]$, the value $\widehat{O}_t(\bb^{(\ell)})$ for the random sample $\bb^{(\ell)}$. Naively, this involves \emph{independently} solving 
$L\sum_{t=1}^{\tau_{\bh}} |\Omega_t |$ 
many relaxed e-value optimization programs, which can still be expensive in large-scale problems even though each single problem can now be  solved efficiently. To address this, we develop an approach that exploits the structural connections between these problems to accelerate computation through warm-starting.

\paragraph{Warm-starting} The relaxed e-value optimization programs given in equation~\eqref{relaxed-go-def} are highly related. To fix ideas, let us consider the computation required for the optimal values $\widehat{O}_{t}(\bb )$ and $\widehat{O}_{t+1}(\bb' )$ for any two indicator vectors $\bb  \in \mathcal{B}(t,s_{t+1}+1)$ and $ \bb'  \in \mathcal{B}(t+1, s_{t+1})$ for which $s_{t+1} \in \Omega_{t+1}, s_{t+1}+1 \in \Omega_t$. 
As we show in Appendix~\ref{appendix:warm-start}, the feasible set of the relaxed e-value optimization program defining $\widehat{O}_{t+1}(\bb' )$ is contained in the feasible set of that defining $\widehat{O}_{t}(\bb )$.\footnote{Technically, this requires some (very natural) assumptions on the extension of the diversity metric (see Assumptions~\ref{perm-invar} and \ref{zero-doesnt-affect} in Appendix~\ref{appendix:warm-start}), which are satisfied by both the Sharpe ratio and Markowitz objective.} This guarantees that the optimal  solution to the former is feasible for the latter. Furthermore, Proposition~\ref{warm-start-prop} in the same Appendix section shows that the objective functions of these programs will be similar if $\bb$ and $\bb'$ agree on most of their first $t$ coordinates. Hence, if $\bb$ and $\bb'$ are close, the solution defining $\widehat{O}_{t+1}(\bb')$ may also be close to that defining $\widehat{O}_{t}(\bb)$. Based on this observation, we implement a projected gradient descent (PGD)\footnote{The PGD solver requires projection onto the constraint set of the relaxed problem. For the Sharpe ratio constraint set~\eqref{sharpe1}--\eqref{sharpe3}, efficient projection algorithms are known~\citep[e.g.,][]{kiwiel2008breakpoint,helgason1980polynomially,held1974validation}. For the more general constraint set~\eqref{relaxedsc1}--\eqref{relaxedsc2} (which is the one that arises in the relaxed Markowitz problem) we develop, in Appendix~\ref{app:warm-start-implement}, an efficient projection algorithm.} solver with warm starts: after solving the relaxed e-value optimization problem for $\widehat{O}_{t+1}(\bb')$, we initialize the solver for $\widehat{O}_{t}(\bb)$ using this solution. See Appendix~\ref{app:warm-start-implement} for implementation details.

\paragraph{Ensuring similarity by coupling} Stepping back to reward computation, calculating the values $\widehat{R^{\mc}_t}(s_{t+1}+1)$ and $\widehat{R}^{\mc}_{t+1}(s_{t+1})$, requires solving the optimization programs defining $\widehat{O}_t(\bb^{(\ell)})$ and $\widehat{O}_{t+1}(\bb'{}^{(\ell)})$, respectively, for $\bb^{(1)}, \ldots, \bb^{(L)} \overset{\text{i.i.d.}}{\sim} \text{Unif}\left(\mathcal{B}(t,s_{t+1}+1)\right)$ and $\bb'{}^{(1)}, \ldots, \bb'{}^{(L)} \overset{\text{i.i.d.}}{\sim} \text{Unif}\left(\mathcal{B}(t+1,s_{t+1})\right)$. As mentioned, our warm-starting procedure becomes effective if $\bb^{(\ell)}$ and $\bb'{}^{(\ell)}$ are similar. To achieve this, we couple the MC samples. We start at time $t+1$ and in particular, with the binary vector $\bb'{}^{(\ell)}$. Given this vector, first deterministically set $\bb^{(\ell)} := \bb'{}^{(\ell)}_{1:t}$ and then, if and only if $b'^{(\ell)}_{t+1} = 0$, sample an index $J \sim \text{Unif}\left(\{j \in [t]: b'_j = 1\}\right)$ and flip $\bb^{(\ell)}$'s $J^{\text{th}}$ coordinate by setting $b^{(\ell)}_J = 0$. This scheme ensures that $\bb^{(\ell)}$ and $\bb'^{(\ell)}_{1:t}$ agree on either all or all but one of their entries, thus ensuring the similarity of $\bb^{(\ell)}$ and $\bb'^{(\ell)}$ needed to make the warm-starting procedure efficient. Furthermore, if $\bb'{}^{(\ell)} \sim \text{Unif}\left(\mathcal{B}(t+1,s_{t+1})\right)$, then $\bb^{(\ell)} \sim \text{Unif}\left(\mathcal{B}(t,s_{t+1}+1)\right)$ and hence the resulting MC estimates of the reward values $\widehat{R}^{\mc}_t(s_{t+1}+1)$ and $\widehat{R}^{\mc}_{t+1}(s_{t+1})$ are still unbiased and independent across $\ell = 1, \ldots, L$ (though they are no longer independent across $(s_t,t)$ pairs). Figure~\ref{fig:warm-start} gives a visual example of the paths of $(s_t,t)$ pairs along which we apply the warm-starting heuristic. In Appendix~\ref{app:warm-start-ablation} we empirically demonstrate that both this warm-starting procedure and our custom PGD solvers can significantly improve computation time. For example, in our simulations involving the Sharpe ratio in Section~\ref{sims:sharpe-markowitz}, the PGD solver on its own (i.e., without using warm starts) yields speedups that range (depending on the simulation setting, $\alpha$, and diversity metric) from $13\times$ to over $100\times$  faster than when using the MOSEK solver \citep{aps2022mosek}, and the use of warm starts can, in some cases, yield an additional $20\%$ reduction in computation time.

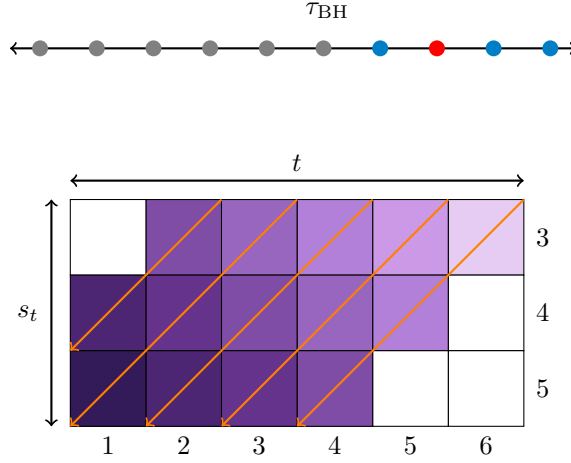
\begin{figure}
    \centering

\begin{tikzpicture}

\begin{scope}[shift={(-4,10)}]
    \draw[<->, thick] (0,0) -- (7.5,0);

    \foreach \i in {0,...,9} {
      \pgfmathsetmacro{\x}{0.4 + 0.75*\i}

      \ifnum\i<6
        \fill[gray] (\x,0) circle (3pt);
      \else
        \ifnum\i<7
          \fill[cyan!60!blue] (\x,0) circle (3pt);
        \else
          \ifnum\i=7
            \fill[red] (\x,0) circle (3pt);
          \else
            \fill[cyan!60!blue] (\x,0) circle (3pt);
          \fi
        \fi
      \fi
    }
\end{scope}


\definecolor{p0}{rgb}{0.9,0.8,0.95}
\definecolor{p1}{rgb}{0.8,0.6,0.9}
\definecolor{p2}{rgb}{0.7,0.5,0.85}
\definecolor{p3}{rgb}{0.6,0.4,0.75}
\definecolor{p4}{rgb}{0.5,0.3,0.65}
\definecolor{p5}{rgb}{0.4,0.2,0.55}
\definecolor{p6}{rgb}{0.3,0.15,0.45}
\definecolor{p7}{rgb}{0.2,0.1,0.35}
\definecolor{p8}{rgb}{0.1,0.05,0.25}

\pgfkeys{
  /colorbydiag/0/.initial=p0,  
  /colorbydiag/1/.initial=p1,
  /colorbydiag/2/.initial=p2,
  /colorbydiag/3/.initial=p3,
  /colorbydiag/4/.initial=p4,
  /colorbydiag/5/.initial=p5,
  /colorbydiag/6/.initial=p6,
  /colorbydiag/7/.initial=p7,
  /colorbydiag/8/.initial=p8   
}

\begin{scope}[shift={(6.8,10)}]
\begin{scope}[xscale=-1, shift={(4,0)}]
\begin{scope}[rotate=270]
\foreach \i in {0,...,2} {  
  \foreach \j in {0,...,5} {  
    \pgfmathtruncatemacro{\row}{6 - \j}  
    \pgfmathtruncatemacro{\col}{\i + 1}

    \pgfmathsetmacro{\iswhite}{
    (\row == 1 && (\col==1)) ||
      (\row == 4 && ( \col==4)) ||
      (\row==5 && (\col==3 || \col==4)) ||
      (\row==6 && (\col==2 || \col==3 || \col==4)) 
    }

    \ifdim \iswhite pt=1pt
      \fill[white] (\i+2,\j) rectangle ++(1,1);
    \else
      \pgfmathtruncatemacro{\diag}{\i + \j}
      \pgfkeysgetvalue{/colorbydiag/\diag}{\colorname}
      \fill[\colorname] (\i+2,\j) rectangle ++(1,1);
    \fi

    \draw (\i+2,\j) rectangle ++(1,1);
  }
}

\draw[<->, thick] (1.75,6) -- (1.75,0) node[midway, above] {\( t \)};

\draw[<->, thick] (2,6.25) -- (5,6.25) node[midway, left=1pt] {\( s_t \)};  

\foreach \i in {6,...,1} {
  \node at (5.25,6.5-\i) {\( \i \)};
}

\foreach \i in {3,...,5} {
  \node at (1.5+\i-3+1,-0.25) {\( \i \)};
}


\node at (-0.5, 2.6) {\( \tau_{\text{BH}} \)};

% ...

\draw[->,orange, thick] (2, 4) -- (4, 6);
\draw[->,orange, thick] (2, 3) -- (5, 6);
\draw[->,orange, thick] (2, 2) -- (5, 5);
\draw[->,orange, thick] (2, 1) -- (5, 4);
\draw[->,orange, thick] (2, 0) -- (5, 3);

\end{scope}
\end{scope}
\end{scope}
\end{tikzpicture}

\caption{Example paths for warm starting in a problem with $n=5$ calibration points and $m=5$ test points. Upper panel illustrates which scores, after sorting, are known to belong to the calibration set (blue), the test set (red), or are unknown (grey) at the BH stopping time $\tau_{\bh}$. The lower panel visualizes the table of $R_t(s_t)$ values for each $s_t \in \Omega_t, t \in [\tau_{\bh}]$. White squares indicate $(s_t,t)$ pairs for which $s_t \not\in \Omega_t$ and hence the value need not be computed. Otherwise, our warm starting heuristic fills out cells by following the orange arrows along the purple gradients.}
\label{fig:warm-start}
\end{figure}

\subsubsection{Reward computations on a coarsened grid of time-points}\label{skipping}
Our final heuristic is to ``coarsen'' the grid of timesteps on which we need to compute the reward function. More specifically, the analyst may choose a grid of times $1 = t_1 < \cdots < t_Q = \tau_{\bh}$ (both the size of and the elements in the grid are allowed to depend on the information in $\mathcal{F}_{\tau_{\bh}}$) and compute \emph{only} the reward functions $R_{t_q}(\cdot), q \in [Q]$ at these coarsened grid points as opposed to all time steps $t\in [\tau_{\bh}]$. The Snell envelope on this coarsened grid is then computed by the updates \begin{equation}\label{coarse-envelope-construction}
        E^{\coarse}_{t_q}(s) := \max \bigg\{ R_{t_q}(s), \sum_{s' \in \Omega_{t_{q-1}}}E_{t_{q-1}}(s')p_{t_q, n-s, t_q-t_{q-1}}(s'-s) \bigg\}
    \end{equation} for each $q \in [Q]$ where $p_{T,S,N}(\cdot)$ is the PMF of a $\text{Hypergeom}(T,S,N)$ random variable.\footnote{We parametrize  $\text{Hypergeom}(T,S,N)$ to be the distribution of number of successes when drawing, uniformly without replacement, $N$ items out of a total population of size $T$ containing $S$ many successfully labeled items.} The intuition behind equation~\eqref{coarse-envelope-construction} is that it is a generalization of the original Snell envelope update~\eqref{envelope-construction} to the setting wherein one only has access to the coarsened process $N_{t_Q}^{\Above}, N_{t_{Q-1}}^{\Above}, \ldots, N_{t_1}^{\Above}$. In the original setting, $N_{t+1}^{\Above} - N_t^{\Above} \mid N_t^{\Above}$ follows a Bernoulli distribution under the exchangeable distribution, and hence the Snell envelope update required a weighted average of two terms corresponding to a simple Bernoulli expectation. Analogously, the multi-term weighted average in equation~\eqref{coarse-envelope-construction} corresponds to an expectation of a hypergeometric distribution which is precisely the distribution of $N_{t_{q+1}}^{\Above} - N_{t_q}^{\Above} \mid (N_{t_q}^{\Above}, \mathcal{F}_{\tau_{\bh}})$ under $\bP_{\exch}$. Just as with the original update~\eqref{envelope-construction}, equation~\eqref{coarse-envelope-construction} is still straightforward from a computational viewpoint.

    Having defined the coarsened reward and Snell envelope processes, the coarsened optimal stopping time is given by the first time in the grid that the former dominates the latter: \[\tau^{*,\coarse} := \max\big\{t \in \{t_1, \ldots, t_Q\}: R_{t}(N_t^{\Above}) \geq E^{\coarse}_{t}(N_t^{\Above}) \big\}.\] Because $\tau^{*,\coarse}$ is still a stopping time, Theorem~\ref{stopping-time-thm} implies we can construct the final selection set in line 10 of Algorithm~\ref{dacs-algo} with respect to its e-values while leaving FDR control intact.

The key upshot of the coarsening heuristic is that we \emph{only} need to compute the reward functions $R_{t_q}(\cdot)$ for each $t_q$ in the grid. Because the grid is user-specified, the analyst is free to choose $Q$ as small as they would like, thereby substantially reducing the number of reward function evaluations and hence the number of optimization programs that must be solved. Of course, this comes at the (statistical) cost of degrading the quality of the stopping $\tau^{*, \coarse}$ as compared to $\tau^{*}$. This heuristic can be used in combination with the MC approximation and relaxation heuristic of Sections~\ref{mc-approx} and \ref{lp-relax} as well as with a generalized version of the warm-start heuristic of Section~\ref{warm-start}. For this generalized warm-start heuristic, we use the same path structure as in Figure~\ref{fig:warm-start}, except paths now skip timesteps (columns) not in the grid $\{t_1, \ldots, t_Q\}$; coupled MC sampling along these paths is performed similarly to the approach described in Section~\ref{warm-start}.

 \section{Real data experiments}\label{expers}

 In this section we apply DACS on real datasets using the underrepresentation index as well as the Sharpe ratio and Markowitz objective. In these experiments, we compare the selection set returned by DACS to that of CS \citep{jin2023modelfreeselectiveinferencecovariate} (which does not optimize for diversity).

 \subsection{Hiring dataset: underrepresentation index}\label{expers:hiring}
We first illustrate our method on a Kaggle job hiring dataset \citep{roshan2020campus}. This is the same dataset on which \cite{jin2023modelfreeselectiveinferencecovariate} apply their conformal selection method and we use it because recruitment datasets from companies are often proprietary. The dataset consists of $215$ candidates, each described by covariates such as gender, final exam results, specialization, prior work experience, etc.; the response is the binary indicator that the applicant is hired. We apply DACS with the underrepresentation index to diversify on gender. The machine learning algorithm we use is a gradient boosting model\footnote{To avoid ties among $\hat{\mu}(X)$'s we add small independent Gaussian noise to its outputs.} with hyperparameters that were found to have performed well on the dataset \citep{kaushik2023recruitment}. We allow it access to all covariates in the dataset \emph{except} gender (i.e., so in this case the covariates $X$ do not contain the diversification variable $Z$). We consider two different problem settings: in Setting 1, we take $n=125$ calibration points and $m=45$ test points (e.g., corresponding to situation in which a moderately-sized company is selecting applicants from a small pool) and in Setting 2 we take $n=45$ and $m = 125$ (e.g., corresponding to the situation in which the company is new and is selecting from a large list of potential applicants).  Our results are averaged over $1000$ random train/calibration/test splits and error bars denote $\pm 1.96$ standard errors.

We plot, in Figure~\ref{hiring:main-result}, the average proportions of clusters selected by DACS as compared to CS conditional on the selection set being non-empty (left) as well as the number of selections made by each method (right). More specifically, recalling that $N_c(\mathcal{R})$ denotes the number of candidates in the selection set $\mathcal{R}$ belonging to category $c$, the blue and red bars in any given cell in the left panel of Figure~\ref{hiring:main-result} represent (approximately) the vector of expected values $\left(\bE\left[\frac{N_c(\mathcal{R})}{|\mathcal{R}|} \mid \mathcal{R} \neq \emptyset\right]\right)_{c=1}^C$ for the DACS and CS selection sets, respectively, for that setting and nominal level $\alpha$. Each cell in the left panel additionally contains blue and red horizontal lines at the (approximate) value $\bP(\mathcal{R} = \emptyset)$, the probability that the selection set is empty (for larger values of $\alpha$, these lines are almost barely visible because they overlap and are both equal to $0$).

\begin{figure}
    \centering
    \includegraphics[scale=0.43]{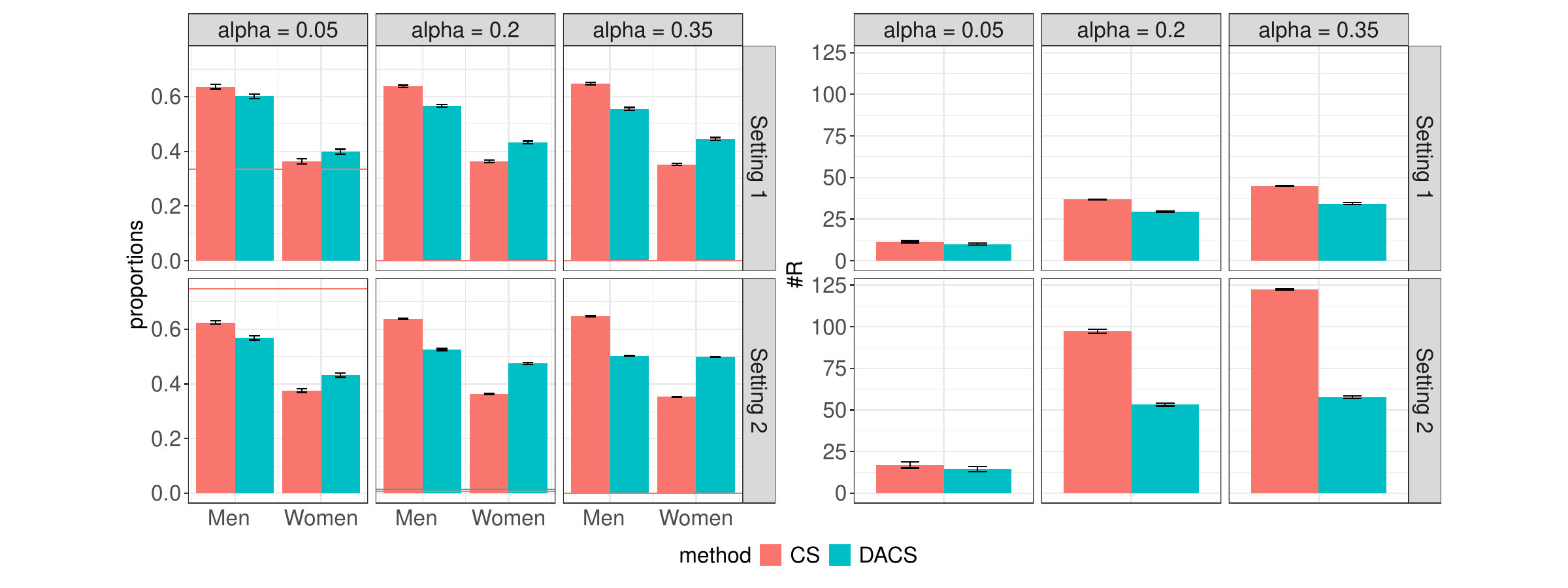}
    \caption{Left: Conditional average cluster proportions given selection set is non-empty for DACS selection set (blue) compared to CS selection set (red) for job hiring dataset at various nominal levels $\alpha$ and different settings (i.e., values of $n$ and $m$); horizontal lines, when visible, denote the average proportion of selection sets that are empty. Right: average number of rejections made by each method for various levels $\alpha$ and different settings of $n$ and $m$.}
    \label{hiring:main-result}
\end{figure}

Figure~\ref{hiring:main-result} shows that DACS is able to make more diverse selections in both settings across all levels of $\alpha$. The improvement in diversity is greatest for larger values of $\alpha$ in both settings. Additionally, the increase in diversity achieved by DACS involves, in some cases, substantially shrinking the selection set size as is demonstrated in the right-hand panel of Figure~\ref{hiring:main-result}. We report further results regarding FDR and power, as well as further details about the data in Appendix~\ref{app:exper-job-hire}.

 \subsection{Diverse drug discovery: Sharpe ratio and Markowitz objective}\label{expers:drug}
 We apply DACS to a G protein-coupled receptor (GPCR) drug-target interaction (DTI) dataset \citep{drugdata}. The dataset consists of $2874$ drug-target pairs along with a measured binding affinity of each pair. Our results are aggregated over $250$ different training/calibration/test splits: we set aside $2156$ examples for training our machine learning model and use $n = 503$ calibration samples and $m = 215$ test samples. We encode the ligands and targets via CNN and Transformer architectures, respectively, using the DeepPurpose package \citep{huang2020deeppurpose}. The machine learning model we use to train $\hat{\mu}$ is a 3-layer MLP trained for $10$ epochs, using the same package. We take the diversification variable $Z$ to be the (encoded) ligand and construct $\Sigma$ as the similarity matrix of Tanimoto coefficients between each pair of ligands \citep{bajusz2015tanimoto}. Following \cite{jin2023modelfreeselectiveinferencecovariate}, we label the $j^{\text{th}}$ drug-target candidate as ``high-quality'' if its binding affinity exceeds a threshold $c_j$, defined as the median binding affinity among training examples with the same target protein.\footnote{If a target does not appear in the training set,  we take $c_j$ to be the median binding affinity across the \emph{entire} training set.} We report results for $\alpha = 0.5$ (CS makes very few selections at smaller nominal levels), utilizing all heuristics from Section~\ref{dacs-heuristics}, in particular using $L = 500$ MC samples and a grid of $Q \approx 100$ time points.\footnote{Technically, the size of our grid is data-dependent and depends on certain divisibility properties to ensure that timestep $1$ is always included. In particular, we set $Q = 100$ if $\frac{\tau_{\bh}-1}{99} = \lceil \frac{\tau_{\bh}}{100}\rceil$ and to $101$ otherwise. Also, if $\tau_{\bh} \leq 100$, we take $Q = \tau_{\bh}$.}

Since the possible ranges of the Sharpe ratio and Markowitz objective vary significantly across problems, their raw values are difficult to interpret. Instead, we evaluate DACS and CS relative to a baseline that reflects the range of achievable diversity values.\footnote{Even relative error could be misleading: if the objective values lie in a narrow interval $[A,B]$ with large $A$ and small $B - A$, then the relative error between DACS and CS will appear small, even if the actual difference is meaningful when compared to the range of realizable selection set diversities. Our baseline accounts for this by incorporating only achievable diversity values.} To compute this baseline, we consider all of the possible diversity-optimized selection sets that could have been achieved had we stopped at a different stopping time. 
In particular, for the given diversity metric $\varphi$ (i.e., either the Sharpe ratio or the Markowitz objective for some value of $\gamma$), we do the following: for each $t \leq \tau_{\bh}$, perform the randomized relaxation of Section~\ref{lp-relax} some $D$ number of times (we use $D = 25$ in our experiments and simulations) with respect to the optimal solution to the relaxed e-value optimization program at time $t$ (i.e., we use the same deterministic $\bchi^{*,\text{relaxed},t}$ but consider $D$ different Bernoulli-randomized selection sets based on it) to obtain selection sets $\mathcal{R}_{t}^{(1)}, \ldots, \mathcal{R}_{t}^{(D)}$. The baseline distribution with respect to which we will make all diversity comparisons is then $\hat{F}^{\varphi}_{\text{baseline}}$, defined to be the empirical CDF of the diversity values $\left\{\varphi\left(\bZ^{\test}_{\mathcal{R}_{t}^{(i)}}\right)\right\}_{i=1,t=1}^{D,\tau_{\bh}}$. In particular, we compare the CDF values, under $\hat{F}^{\varphi}_{\text{baseline}}$, of DACS' diversity score to CS'. These ``normalized diversities'' have the advantage that they are in $[0,1]$, allowing them to be more easily compared and interpreted. 
Furthermore, because the baseline accounts for all possible stopping times, it is inherently difficult to outperform---that is, achieving normalized CDF values near $1$ is challenging---while still maintaining FDR control (even approximate control, as ensured by Proposition~\ref{relaxed-fdr}). Indeed, a procedure that selects the most diverse stopping-time-based selection set in a post-hoc manner may fail to control FDR.
Since the DACS and CS selection sets are based on two possible times at which we could have stopped, it makes sense to compare them relative to the distribution of diversity-optimized selection sets across all valid stopping times.

  Finally, we note that the $\gamma$ parameters for which we present the Markowitz objective results were chosen to be near the average value $2/\lambda_{\max}(\Sigma)$ across $1000$ splits as we found that larger values tended to result in the objective favoring the empty selection set while smaller values often resulted in the CS selection set being optimal---both expected outcomes given the structure of the Markowitz objective. In practice, the analyst is free to choose $\gamma$ by inspecting $\lambda_{\max}(\Sigma)$ (or, indeed, any permutation-invariant function of $\Sigma$), without breaking FDR control because it leaves the exchangeability between calibration and test samples intact.

 \begin{figure}
    \centering
    \includegraphics[scale=0.4]{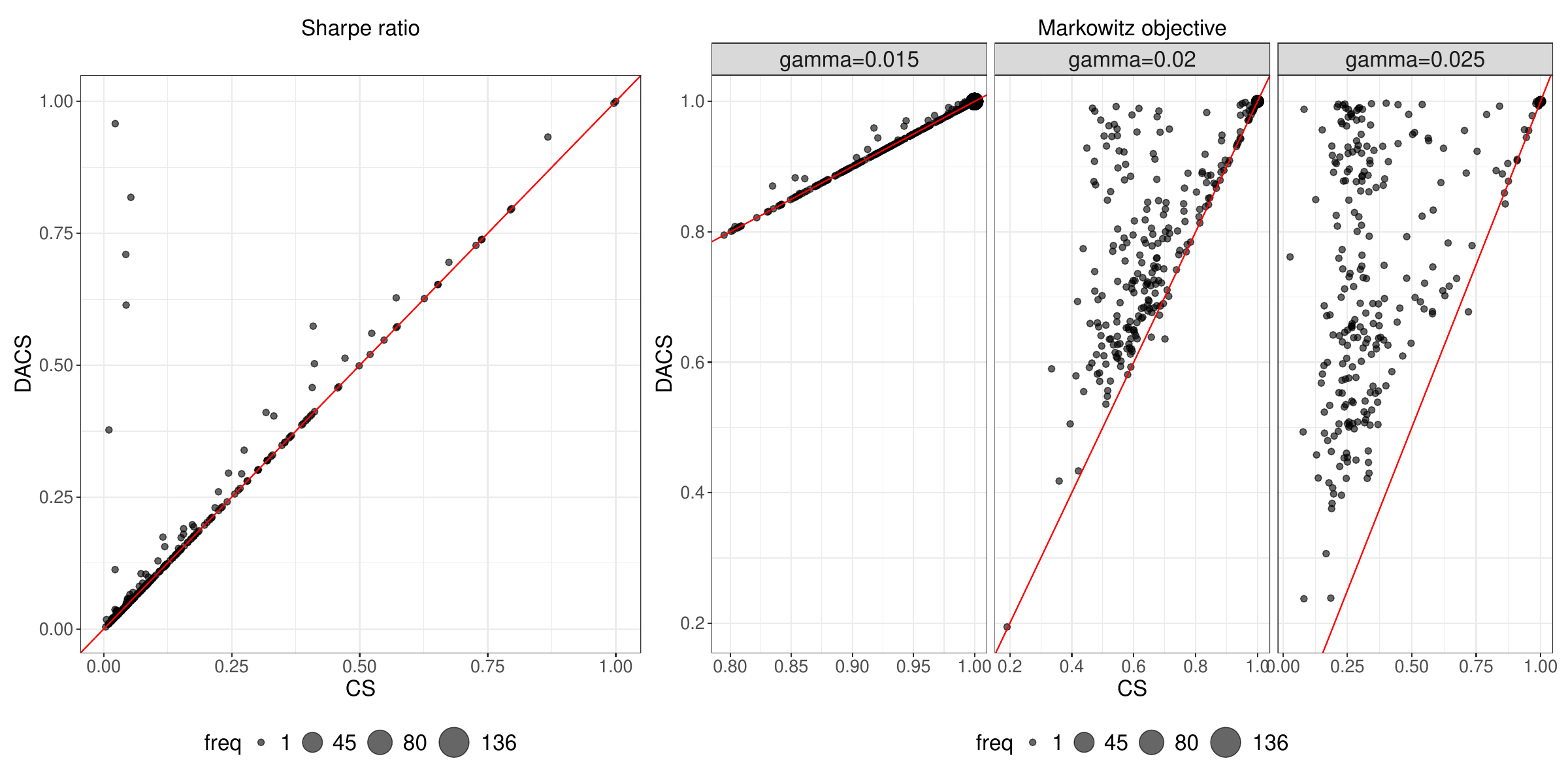}
    \caption{DACS' normalized diversity under $\hat{F}^{\varphi}_{\text{baseline}}$ compared to CS' for both the Sharpe ratio and Markowitz objective for GPCR dataset. Each scatter point represents a different data split. Red lines are $y=x$.}
    \label{drug-exper:main-result}
\end{figure}

Figure~\ref{drug-exper:main-result} presents the main results on the GPCR dataset using both the Sharpe ratio and the Markowitz objective. Overall, DACS achieves marginally higher normalized Sharpe ratio diversity values than CS, though the two methods frequently make selections exhibiting the same diversity. In this setting, both methods tend to produce selection sets with low normalized diversity values. For the Markowitz objective, DACS' performance relative to CS appears to become superior for larger values of $\gamma$. This makes sense: as $\gamma$ gets larger it becomes more and more desirable to return smaller selection sets (indeed as $\gamma \rightarrow \infty$, the empty set is optimal) and consequently the CS selection set becomes worse. In Appendix~\ref{app:exper-drug-discover}, we report the FDR, power, and number of selections made by each method. Although Proposition~\ref{relaxed-fdr} guarantees FDR control only at level $1.3\alpha$, Figure~\ref{fig:drugexper-fdr-power} in Appendix~\ref{app:exper-drug-discover} shows that it is controlled below the nominal level $\alpha$ for all objectives.

\section{Simulations}\label{sims}
  Just as in Section~\ref{expers}, we compare the selection sets returned by DACS to those returned by CS in a variety of different simulation settings to further investigate their performance. All simulation results are aggregated over $250$ independent trials and, when reported, error bars denote $\pm 1.96$ standard errors.

\subsection{Underrepresentation index}\label{sims:underrep}
We consider $6$ different simulation settings (i.e., data-generating processes) all with $n = 500$ calibration samples and $m = 300$ test samples. We take $\hat{\mu}$ to be a fitted $3$-hidden-layer multi-layer perceptron (MLP), trained on a training set of $1000$ independent $(X_i,Y_i)$ pairs---though we also report additional results when ordinary least-squares (OLS) and support vector machine (SVM) regressions are used to train $\hat{\mu}$ in Appendix~\ref{app:add-sim-results-underrep}. While we defer a precise description of the $6$ different simulation settings to Appendix~\ref{app:add-sim-detail-underrep}, we provide a high-level discussion here. All simulation settings in this section involve a hierarchical Gaussian mixture model in which the diversification variable $Z$ is a categorical variable, denoting the upper-level cluster index. An (unobserved) second-level index label $Z^{\sub}$ is then drawn from another categorical distribution, conditional on $Z$ and then $X$ is sampled from the Normal distribution $\mathcal{N}(\mu_{Z^{\sub}}, 1)$, where the parameters $\bmu$ are setting-dependent. Finally, $Y$ is distributed according to a (setting-dependent) regression function of $X$ plus homoscedastic standard Gaussian noise.

\begin{figure}
\centering
    \includegraphics[scale=0.6]{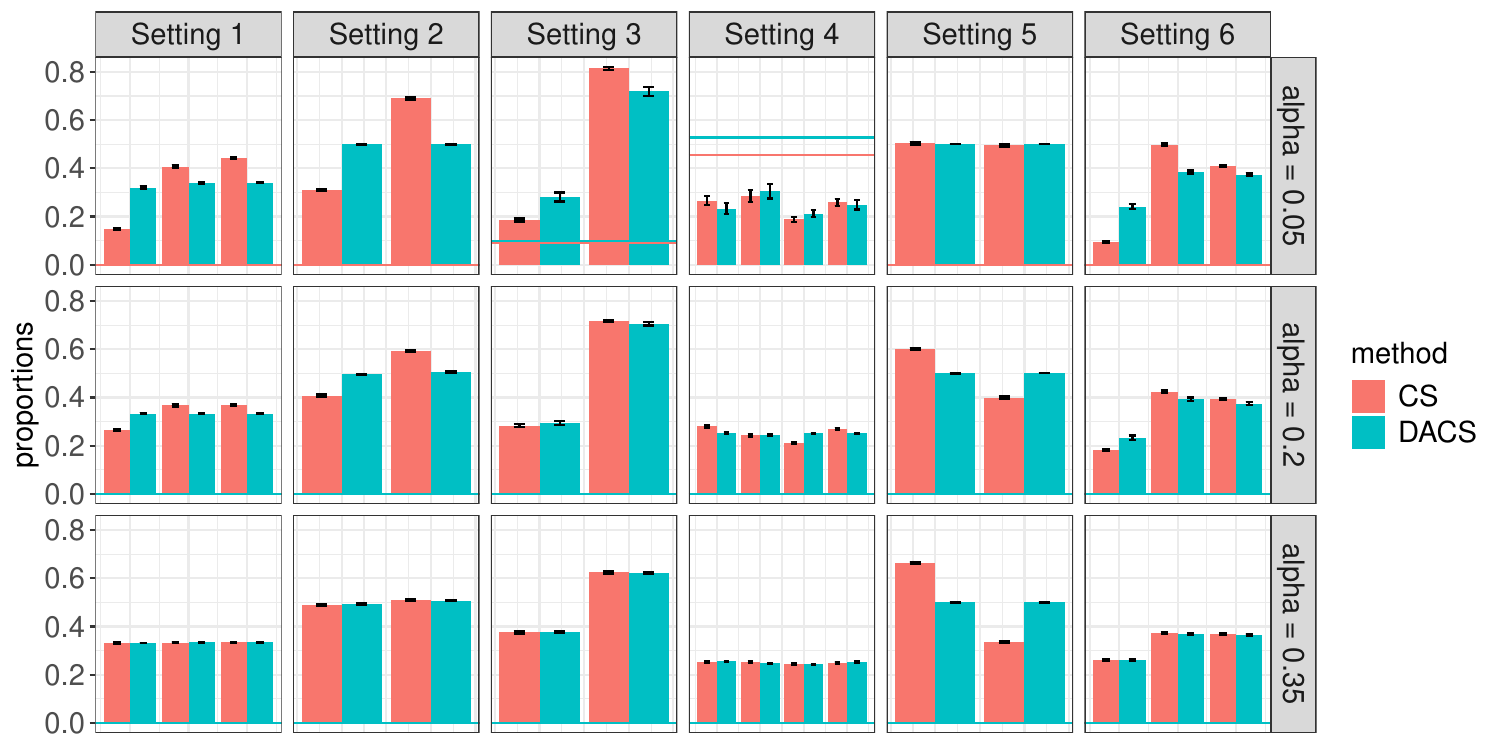}
    \caption{Average cluster proportions conditional on being non-empty for DACS selection set (blue) compared to CS selection set (red) for various simulation settings and nominal levels $\alpha$; $\hat{\mu}$ is a fitted MLP. Within any given cell, bars which are closer to uniform denote a more diverse selection set. Horizontal lines, when visible, denote the average proportion of sets which are empty.}
    \label{fig:underrep-result1}
\end{figure}

\begin{figure}
\centering
    \includegraphics[scale=0.6]{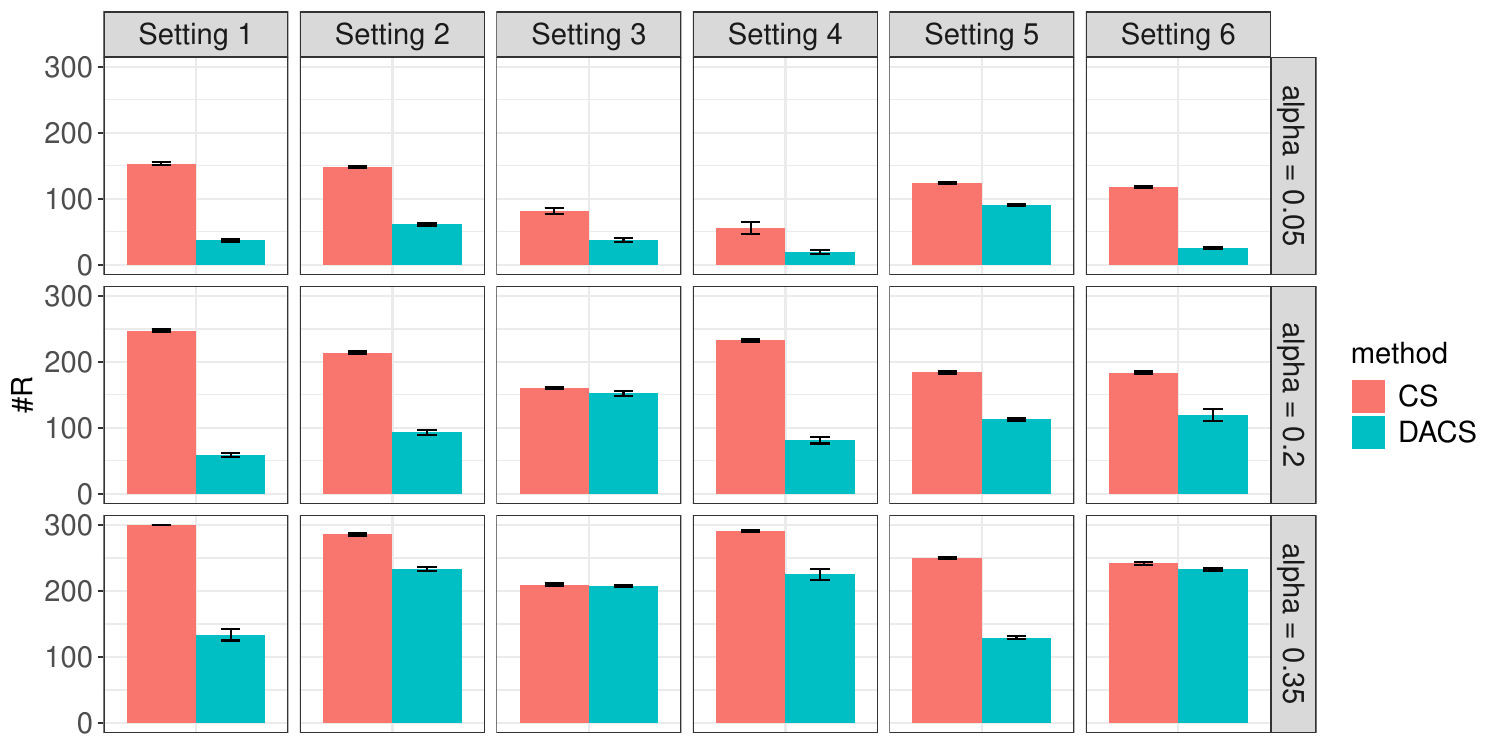}
    \caption{Average numbers of selections made by DACS selection set (blue) compared to CS selection set (red) for various simulation settings and nominal levels $\alpha$; $\hat{\mu}$ is a fitted MLP.}
    \label{fig:underrep-result2}
\end{figure}

Figure~\ref{fig:underrep-result1} shows that, conditional on being non-empty, the selection sets produced by DACS with the underrepresentation index are more diverse, on average, than those produced by CS (though sometimes DACS does return the empty selection set more frequently than CS which results in a lower average underrepresentation index; we discuss this further in the next paragraph). In particular, the figure shows that DACS produces selections which are closer to---and, in many cases, almost exactly---uniform across clusters (thereby indicating greater diversity) than CS. Appendix~\ref{app:add-sim-results-underrep} contains further simulation results for the underrepresentation index. It is also interesting to compare the diversities of DACS and CS for differing values of $\alpha$ since, as alluded to in Section~\ref{contribution}, we may be willing to accept a selection set with a weaker FDR control guarantee for the sake of greater diversity. In Settings 3, 4, and 6, we see that increasing the nominal level $\alpha$ enables DACS to construct more and more diverse selection sets (in the other settings, DACS is already constructing essentially maximally diverse selection sets even at $\alpha = 0.05$). Figure~\ref{fig:underrep-result2}, which displays the average number of selections for the same simulation settings and nominal levels as Figure~\ref{fig:underrep-result1}, also shows that, in each of these settings, DACS at level $\alpha = 0.35$ makes the same or more rejections than CS at level $\alpha = 0.05$. Thus, by increasing the nominal level, DACS can construct more diverse selection sets without sacrificing the selection size.

To investigate scenarios where DACS exhibits less impressive performance, we examine its performance in Setting 4 with $\alpha = 0.05$, where it produces selection sets with lower underrepresentation index than CS (see Figure~\ref{fig:diversity-result-mlp-cluster} in Appendix~\ref{app:avg-div-underrep}). Similar results are seen for OLS and SVM regressors in the same setting and level, and for OLS in Settings 3 and 4 at $\alpha = 0.2$. For OLS, we attribute the drop in performance to poor predictive quality. For the other regressors, Figure~\ref{fig:underrep-result-most_rep} in Appendix~\ref{app:most-rep} shows that the most-represented category in the CS set is often small, making it challenging to prune it into a more diverse set. We observe a similar effect when reducing the number of test points to $m=100$: DACS yields little or no diversity improvement compared to CS in many settings (see Appendix~\ref{app:further-underrep-100}). We conjecture that this may also be due to the smaller size of the CS selection set in such settings, which limits the ability to improve its underrepresentation index.

\subsection{Sharpe ratio and Markowitz objective}\label{sims:sharpe-markowitz}
In this section we study the empirical performance of DACS using the Sharpe ratio and Markowitz objectives in simulation. We consider four simulation settings, all with $n = 500$ calibration samples, $m = 100$ test samples, and we take $\hat{\mu}$ to be a fitted a 2-hidden-layer MLP trained on an independent set of $1000$ $(X_i,Y_i)$ pairs. In this section, we set $Z = X$, so that we are interested in diversifying on the same covariates from which $\hat{\mu}$ bases its quality predictions; the similarity matrix $\Sigma$ is the RBF kernel matrix. Finally, we use the MC approximation of rewards (using $L = 300$ MC samples) as well as all of the heuristics developed in Section~\ref{dacs-heuristics}, including coarsening to a grid of $Q \approx 50$ evenly spaced time points.\footnote{Just as in Section~\ref{expers:drug}, the approximation is due to the fact that $Q = 50$ or $51$ depending on divisibility properties of $\tau_{\bh}$. Again, if $\tau_{\bh} \leq 50$, we take $Q = \tau_{\bh}$.} We defer the details regarding the precise specifications of each simulation setting to Appendix~\ref{app:add-sim-detail-sharpe-markowitz}.

\begin{figure}
    \centering
    \includegraphics[scale=0.4]{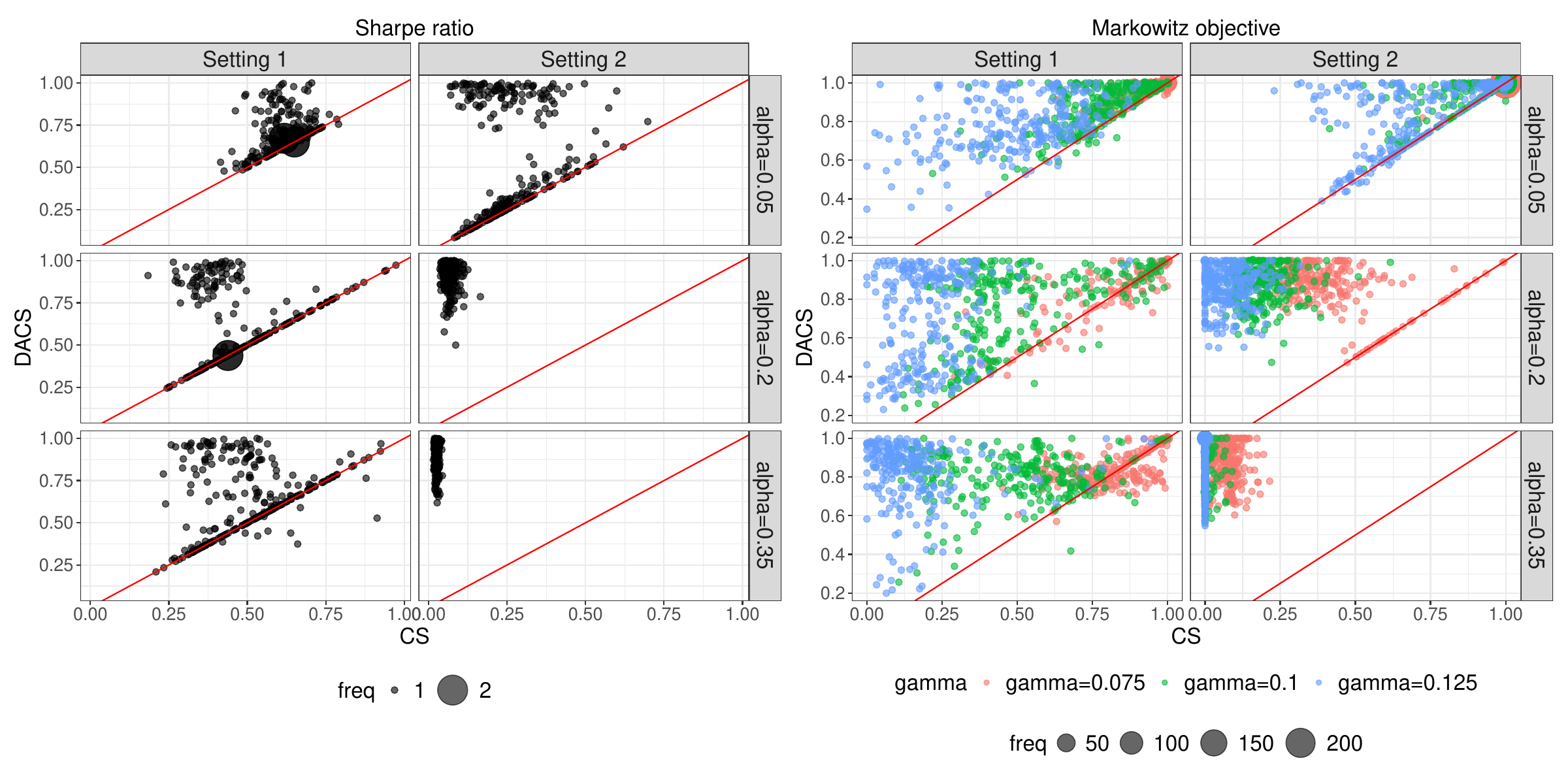}
    \caption{DACS' normalized diversity under $\hat{F}^{\varphi}_{\text{baseline}}$ compared to CS' for both the Sharpe ratio and Markowitz objective for various values of $\gamma$. Each scatter point represents a different simulation replicate; results reported for two simulation settings and three nominal levels $\alpha$. Red line is $y=x$.}
    \label{sharpe-markowitz:main-sim}
\end{figure}

Figure~\ref{sharpe-markowitz:main-sim} shows results, in the first two settings, comparing the diversity---again using the same normalization with respect to a baseline distribution as described in Section~\ref{expers:drug}---of DACS' selection sets compared to CS at the same nominal level $\alpha$ for various simulation settings, values of $\alpha$, and values $\gamma$ (in the case of the Markowitz objective). The high-level takeaway from all of these plots is the same: DACS using both the Sharpe ratio and the Markowitz objective constructs, on the whole, more diverse (as measured by that metric) selection sets than CS in these settings. 
In Appendix~\ref{app:add-sim-results-sharpe-markowitz}, we also report the FDR, power, and number of rejections made by DACS as well as its computation time for these settings. Though DACS is guaranteed to control FDR only at level $1.3\alpha$, Figure~\ref{fig:sharpe-markowitz-fdr-power} in Appendix~\ref{app:add-sim-results-sharpe-markowitz-fdr-power} shows that FDR is controlled below $\alpha$ across both simulation settings, mirroring our results in Section~\ref{expers:drug}. In Appendix~\ref{app:add-sim-results-sharpe-markowitz}, we show that our custom PGD solvers significantly reduce computation time compared to MOSEK\footnote{In all our simulations---as well as in the experiments in Section~\ref{expers:drug}---we use MOSEK to solve the final relaxed e-value optimization program in line 10 of Algorithm~\ref{dacs-algo}, even if the PGD solver is used for reward computation.}: for the Sharpe ratio and Markowitz objectives, the speedup in these settings ranges from $67$--$76\times$ and $13$--$111\times$, respectively. The same appendix also shows that warm-starting further reduces runtime: in Setting 1, by roughly $8$--$13\%$ for the Sharpe ratio and $17$--$21\%$ for the Markowitz objective. In Setting 2, warm-starting and coupling lower solver time for the Sharpe ratio, but reduce overall computation time only slightly---likely due to the overhead of constructing coupled Monte Carlo samples. For the Markowitz objective in Setting 2, warm-starting yields an additional $9$--$21\%$ reduction in runtime.

\begin{figure}
    \centering
    \includegraphics[scale=0.4]{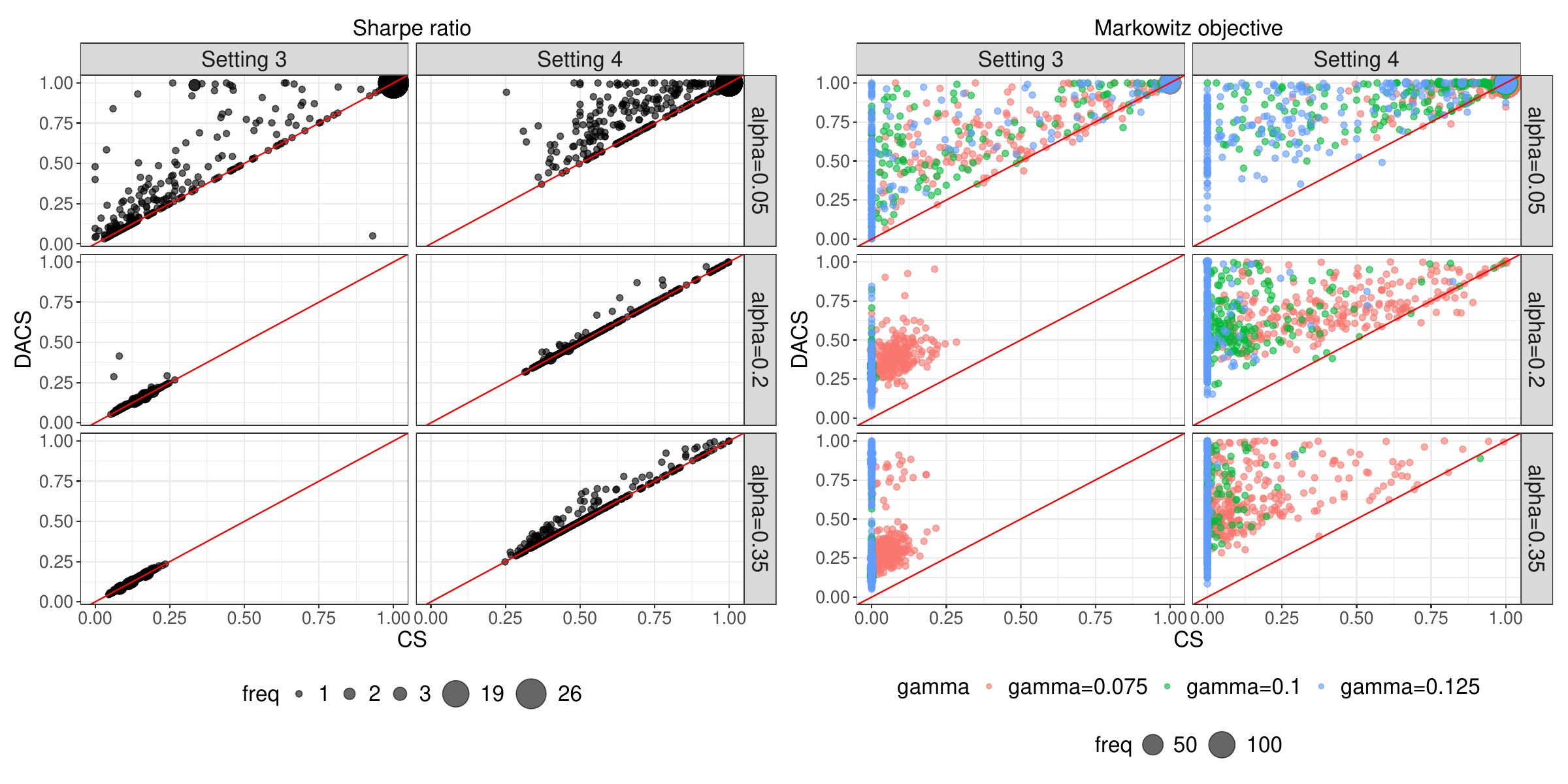}
    \caption{Additional scatter plots of normalized diversities of DACS versus CS for two more simulation settings.}
    \label{sharpe-markowitz:extra-sim}
\end{figure}

Figure~\ref{sharpe-markowitz:extra-sim} shows normalized diversity plots for two additional simulation settings. The left facet highlights challenges DACS sometimes faces in diversifying significantly beyond CS for the Sharpe ratio. While DACS produces selection sets that are no less diverse than CS, on the whole--and indeed, does produce substantially more diverse selection sets for some problem settings---its advantage is sensitive to both the underlying setting and the nominal level $\alpha$. This sensitivity is particularly pronounced in Setting 3 for large $\alpha$, wherein DACS fails to achieve greater diversity than CS and also attains low absolute normalized Sharpe ratio scores. This indicates that more diverse selection sets could have been produced had different stopping times been chosen (though, we again remind the reader that choosing the stopping time in such a post-hoc way may compromise FDR control). For the Markowitz objective, DACS consistently produces more diverse selection sets than CS across all settings. However, again in Setting 3 at larger $\alpha$ values, we observe a similar trend: for some values of $\gamma$, the normalized diversity scores---while clearly improved over CS---are at times somewhat low.

\section*{Acknowledgements} The authors would like to thank John Cherian, Val\'erie Ho, Etaash Katiyar, Ginnie Ma, Emma Pierson, Chiara Sabatti, Henry Smith, Anav Sood, Asher Spector, and Skyler Wu for helpful discussions. Y.N.~acknowledges support by a Graduate Research Fellowship from the National Science Foundation. E.J.C.~was supported by the Office of Naval Research grant N00014-24-1-2305, the National Science Foundation grant DMS-2032014, and the Simons Foundation under award 814641. Y.J.~acknowledges support by the Wojcicki-Troper Postdoctoral Fellowship at Harvard Data Science Initiative. Some of the computing for this project was performed on the Sherlock cluster. We would like to thank Stanford University and the Stanford Research Computing Center for providing computational resources and support that contributed to these research results.

\bibliographystyle{plainnat}
\bibliography{references.bib}

\newpage
\appendix
\section{Proofs of main results}
\subsection{Proof of Theorem~\ref{stopping-time-thm}}\label{appendix:main-thm-proof}
We must show that \begin{equation}\label{eq:e-value-equation}\bE\left[\frac{(n+1)\mathds{1}\{\widehat{V}_i \leq W_{(\tau)}\}\mathds{1}\{Y_{n+i} \leq 0\}}{1 + n - N_{\tau}^{\Above}}\right] \leq 1\end{equation} for any $\left(\mathcal{F}_t\right)_{t=0}^{n+m}$-stopping time $\tau$ and each $i \in [m]$. The idea is to relate the above expectation to one involving a stopping time and filtration obtained by replacing the $i^{\text{th}}$ imputed score $\widehat{V}_i$ with the (unobserved) oracle score $V^\star_i := V^{\clip}(X_{n+i}, Y_{n+i})$. On this $i^{\text{th}}$ ``oracle version,'' we show that quantity inside the expectation, viewed as a sequence in $t$, is a backwards supermartingale, and then conclude by the optional stopping theorem.

To facilitate the analysis, we will consider $i^{\text{th}}$ oracle versions of the quantities which define the filtration $(\mathcal{F}_t)$. More precisely, define \[\bW(\loo) := (V_1, \ldots, V_n, \widehat{V}_1, \ldots, \widehat{V}_{i-1}, V^\star_i, \widehat{V}_{i+1}, \ldots, \widehat{V}_m),\] with $W_{(1)}(\loo) \leq \cdots \leq W_{(n+m)}(\loo)$ denoting the sorted values of $\bW(\loo)$ and $W_{(0)}(\loo) := -\infty$. Also, let $\pi_{\sort}^{\loo}$ denote the permutation for which $(W_{\pi_{\sort}^{\loo}(1)}(\loo), \ldots, W_{\pi_{\sort}^{\loo}(n+m)}(\loo)) = (W_{(1)}(\loo), \ldots, W_{(n+m)}(\loo))$. Before proceeding, we pause to discuss precisely how $\pi_{\sort}^{\loo}$ breaks ties among scores equal to infinity. 

\begin{remark}[Ties among $i^{\text{th}}$ oracle scores]\label{rem:pi-sort}
    
    In the main text, ties in the observed score list $\bW$ were allowed to be broken arbitrarily. This is because the assumption of no ties among the $\hat{\mu}(X)$'s guaranteed that only calibration scores could be tied (in particular, only those calibration scores equal to infinity). Turning to the $i^{\text{th}}$ oracle scores, however, there may now be ties between calibration scores equal to infinity and the $i^{\text{th}}$ oracle test score and we must define the sorting permutation $\pi_{\sort}^{\loo}$ more carefully. In particular, we define $\pi_{\sort}^{\loo}$ so that it maintains the same ordering as $\pi_{\sort}$ on the calibration scores (i.e., it breaks ties among calibration scores equal to infinity in the same way that $\pi_{\sort}$ does) but then breaks ties with $V^\star_i$ randomly. More specifically, if $V^\star_i$ along with $k$ calibration scores $V_{\ell_1}, \ldots, V_{\ell_k}$ are all equal to infinity, then we sample $U \sim \text{Unif}\{n+m-k-1, \ldots, n+m\}$ and set $(\pi_{\sort}^{\loo})^{-1}(n+i) = U$ and break ties among $V_{\ell_1}, \ldots, V_{\ell_k}$ according to the rule $(\pi_{\sort}^{\loo})^{-1}(\ell_j) \leq (\pi_{\sort}^{\loo})^{-1}(\ell_{j'}) \iff (\pi_{\sort})^{-1}(\ell_{j}) \leq (\pi_{\sort})^{-1}(\ell_{j'})$ for $j,j' \in [k]$. 
\end{remark}

Having defined the sorted $i^{\text{th}}$ oracle scores, set $B_s(\loo) := \mathds{1}\left\{\pi_{\sort}^{\loo}(s) \leq n\right\}$ and $N_t^{\Above}(\loo) = \sum_{s=t+1}^{n+m}B_s(\loo)$. Finally, we define the $i^{\text{th}}$ oracle filtration at time $t$ as \[\mathcal{F}_t(\loo) := \sigma\left(B_{t+1}(\loo), \ldots, B_{n+m}(\loo), \bZ^{()}(\loo)\right),\] where analogously to the main text, $\bZ^{()}(\loo)$ is the list of $\bW(\loo)$-sorted diversification variables $Z_i, i \in [n+m]$ (i.e., $Z_{(t)}(\loo) = Z_{\pi_{\sort}^{\loo}(t)}(\loo)$).

We will now construct a $\left(\mathcal{F}_t(\loo)\right)_{t=0}^{n+m}$-stopping time $\tau(\loo)$ from the $\left(\mathcal{F}_t\right)_{t=0}^{n+m}$-stopping time $\tau$. To do so, recall that since $\tau$ is a $\left(\mathcal{F}_t\right)_{t=0}^{n+m}$-stopping time, the indicators $\mathds{1}\left\{\tau < s\right\}$ are $\mathcal{F}_s$-measurable for each $s \in [n+m]$. Recalling that \[\mathcal{F}_t := \sigma\left(B_{t+1}, B_{t+2}, \ldots, B_{n+m}, \Zsort\right),\] this implies that there exist measurable functions $g_s: \{0,1\}^{n+m-s} \times \mathcal{Z}^{n+m} \rightarrow \{0,1\}$ such that \[\mathds{1}\left\{\tau < s\right\} = g_s(B_{s+1}, \ldots, B_{n+m}, \Zsort)\] for each $s \in [n+m]$. To construct $\tau(\loo)$, we will simply replace the arguments to $g_s$ with their $i^{\text{th}}$ oracle versions and then take a running minimum. More specifically, defining the functions \[h_s(B_{s+1}(\loo), \ldots, B_{n+m}(\loo), \Zsort(\loo)) := \min_{t=s}^{n+m}g_t(B_{t+1}(\loo), \ldots, B_{n+m}(\loo), \Zsort(\loo)),\] we construct the $i^{\text{th}}$ oracle stopping time $\tau(\loo)$ as \[\tau(\loo) := \max\left\{s \in [n+m]: h_s(B_{s+1}(\loo), \ldots, B_{n+m}(\loo), \Zsort(\loo)) = 0\right\}.\] By construction, we have that \[\mathds{1}\left\{\tau(\loo) < s\right\} = h_s(B_{s+1}(\loo), \ldots, B_{n+m}(\loo), \Zsort(\loo)), \forall s \in [0:n+m].\] Because the right-hand-side is $\mathcal{F}_s(\loo)$-measurable, $\tau(\loo)$ is a $\left(\mathcal{F}_t(\loo)\right)_{t=0}^{n+m}$-stopping time.

We will now relate $\tau$ to $\tau(\loo)$ on the event that $\{Y_{n+i} \leq 0\}$. Because we are using the clipped score, we have that $\widehat{V}_i = V^\star_i$ on this event and hence $\bW(\loo) = \bW$. Consequently, on the event that $\{Y_{n+i} \leq 0\}$, \[(B_1(\loo), \ldots, B_{n+m}(\loo), \Zsort(\loo)) = (B_1, \ldots, B_{n+m}, \Zsort),\] and the random variables which generate $\mathcal{F}_s$ are equal to their $i^{\text{th}}$ oracle counterparts which generate $\mathcal{F}_s(\loo)$, for each $s \in [0:n+m]$. There are two important consequences of this fact on the event that $\{Y_{n+i} \leq 0\}$: 

\begin{subequations}\label{eq:subeqns}
     \begin{align}
      \tau(\loo) =  \max\left\{s \in [n+m]: g_s(B_{s+1}, \ldots, B_{n+m}, \Zsort) = 0\right\} = \tau, \label{eq:subeq1}\\
   \text{and}~~   N_{\tau(\loo)}^{\Above}(\loo) = N_{\tau}^{\Above}. \label{eq:subeq2}
     \end{align}
    \end{subequations}

As a result of equalities~\eqref{eq:subeq1}--\eqref{eq:subeq2}, we obtain the deterministic equality \[\frac{(n+1)\mathds{1}\{\widehat{V}_i \leq W_{(\tau)}\}\mathds{1}\{Y_{n+i} \leq 0\}}{1 + n - N_{\tau}^{\Above}} = \frac{(n+1)\mathds{1}\{V^\star_i \leq W_{(\tau(\loo))}(\loo)\}\mathds{1}\{Y_{n+i} \leq 0\}}{1 + n - N_{\tau(\loo)}^{\Above}(\loo)}.\footnote{Because we are using the clipped score, it is possible that both $ W_{(\tau(\loo))}(\loo) = \infty$  and $V^\star_i = \infty$ (though the second equality is impossible on the event that $Y_{n+i} \leq 0$). As such, the inequality ``$V^\star_i \leq W_{(\tau(\loo))}$'' should be read as saying: $(\pi_{\sort}^{\loo})^{-1}(n+i) \leq \tau(\loo)$. We opt for the former notation at various points throughout this Appendix section for ease of readability.}\]

The right-hand-side is further upper-bounded upon removing the indicator that $Y_{n+i} \leq 0$, and hence it suffices to show that \[\mathbb{E}\left[\frac{(n+1)\mathds{1}\{V^\star_i \leq W_{(\tau(\loo))}(\loo)\}}{1 + n - N_{\tau(\loo)}^{\Above}(\loo)}\right] \leq 1.\]

Because $N_{n+m}^{\Above}(\loo) = 0$ and $\tau(\loo)$ is a $\left(\mathcal{F}_t(\loo)\right)_{t=0}^{n+m}$-stopping time, this will follow immediately by the optional stopping theorem if it can be shown that the sequence \begin{equation}\label{loo-mg}M_t := \frac{(n+1)\mathds{1}\{V^\star_i \leq W_{(t)}(\loo)\}}{1 + n - N_{t}^{\Above}(\loo)}\end{equation} is a backwards supermartingale with respect to a backwards filtration containing $\left(\mathcal{F}_t(\loo)\right)_{t=0}^{n+m}$. But~\eqref{loo-mg} is precisely the backwards supermartingale that arises in the (deterministic) outlier detection problem, and the proof of Proposition 11 of \cite{lee2024boostingebhconditionalcalibration}, applied to the negative of $\bW(\loo)$, establishes that it is a backwards supermartingale.

\subsection{Proof that $e_1^{(\tau_{\bh})}, \ldots, e_m^{(\tau_{\bh})}$ are e-values}\label{appendix:ebh-evalues}
In this section, we define the BH stopping time and establish that the variables $\left(e_1^{(\tau_{\bh})}, \ldots, e_m^{(\tau_{\bh})}\right)$ are e-values. The BH stopping time is the FDP-estimator-based stopping time given by the equation \begin{equation}\label{def:bh-st}\tau_{\bh} := \max\left\{t \in [n+m]: \frac{m}{n+1} \cdot \frac{1 + \sum_{s=1}^t B_s}{1 \vee \sum_{s=1}^t (1-B_s)} \leq \alpha\right\},\end{equation} where $\tau_{\bh} := 0$ if the above set is empty. This stopping time viewpoint of the BH procedure for conformal p-values is from \cite{mary2022semi}.

In view of Theorem~\ref{stopping-time-thm}, it suffices just to show that $\tau_{\bh}$ is a $\left(\mathcal{F}_t\right)_{t=0}^{n+m}$-stopping time.

\begin{corollary}
    The variables $\left(e_1^{(\tau_{\bh})}, \ldots, e_m^{(\tau_{\bh})}\right)$ are e-values.
\end{corollary}
\begin{proof}
    As discussed, it suffices to show that $\tau_{\bh}$ is a $\left(\mathcal{F}_{t=0}^{n+m}\right)_t$-stopping time. Recall that \[\tau_{\bh} := \max\left\{t \in [n+m]: \frac{m}{n+1} \cdot \frac{1 + \sum_{s=1}^t B_s}{1 \vee \sum_{s=1}^t (1-B_s)} \leq \alpha\right\}.\] Rewriting \[\frac{m}{n+1} \cdot \frac{1 + \sum_{s=1}^t B_s}{1 \vee \sum_{s=1}^t (1-B_s)} = \frac{m}{n+1} \cdot \frac{1 + n - N_t^{\Above}}{1 \vee (t - n + N_t^{\Above})},\] we find that this quantity is $\left(\mathcal{F}\right)_{t=0}^{n+m}$-measurable and hence $\tau_{\bh}$ is a $\left(\mathcal{F}_t\right)_{t=0}^{n+m}$-stopping time. The result then follows by Theorem~\ref{stopping-time-thm}.
\end{proof}

\subsection{Proof of Theorem~\ref{thm:self-consistency-validity}}\label{appendix:self-consistency}
The proof is essentially identical to that of \citet[][Proposition 2]{wang2022false}. Suppose that the selection set $\mathcal{R}$ is self-consistent. Then its FDR is

\begin{align*}
    \bE\left[\frac{\sum_{i=1}^m \mathds{1}\{i \in \mathcal{R}\}\mathds{1}\{Y_{n+i} \leq 0\}}{1 \vee |\mathcal{R}|}\right] &\leq \bE\left[\frac{\sum_{i=1}^m \alpha e_i\mathds{1}\{i \in \mathcal{R}\}\mathds{1}\{Y_{n+i} \leq 0\}}{m}\right]\\
    &\leq \frac{\alpha}{m}\sum_{i=1}^m \bE[e_i\mathds{1}\{Y_{n+i} \leq 0\}]\\
    &\leq \alpha
\end{align*} where the first inequality is by self-consistency and the last inequality is because $e_1, \ldots, e_m$ are e-values.

\subsection{Proof of Proposition~\ref{prop:support}}\label{appendix:support}
First observe that the set $\Omega_t$ is $\mathcal{F}_{\tau_{\bh}}$-measurable since both $N_{\tau_{\bh}}^{\Above}$ and $\tau_{\bh}$ are. First we show that $\bP_{\exch}(N_t^{\Above} \in \Omega_t \mid \mathcal{F}_{\tau_{\bh}}) = 1$ for which it suffices to show that \[\bP_{\exch}(N_t^{\Above}  = s_t \mid N_{\tau_{\bh}}^{\Above}, \tau_{\bh}) = 0\] for any $s_t$ belonging to one of the four following cases: 1.~$s_t < N_{\tau_{\bh}}^{\Above}$, 2.~$s_t < n-t$, 3.~$s_t > n$, 4.~$s_t > \tau_{\bh} - t + N_{\tau_{\bh}}^{\Above}$. We handle these four cases presently:
\begin{enumerate}
    \item The sequence $N_t$ is monotone non-increasing in $t$ and hence $s_t < N_{\tau_{\bh}}^{\Above}$ is impossible as $t \leq \tau_{\bh}$.
    \item Observe that $N_t^{\Above} + t$ is an upper-bound on $n$, the total number of calibration data, thereby implying that $s_t < n-t$ is impossible.
    \item The number of calibration points above $W_{(t)}$, $N_t^{\Above}$, cannot be more than the total number of calibration points $n$.
    \item For $t \leq \tau_{\bh}$, the largest possible increase from $N_{\tau_{\bh}}^{\Above}$ to $N_{t}^{\Above}$ occurs if $B_{t+1} = \cdots = B_{\tau_{\bh}} = 1$, implying that $N_t^{\Above} \leq N_{\tau_{\bh}}^{\Above} + \tau_{\bh} - t$.
\end{enumerate}

Now, we show that $\bP_{\exch}(N_t^{\Above} = s_t \mid \mathcal{F}_{\tau_{\bh}}) > 0$ for all $s_t \in \Omega_t$. To do we must exhibit a sequence $(b_{1}, \ldots, b_{\tau_{\bh}})$ such that $s_t = N_{\tau_{\bh}}^{\Above} + b_{t+1} + \cdots + b_{\tau_{\bh}}$ and $P(B_1=b_1, \ldots, B_{\tau_{\bh}} =b_{\tau_{\bh}} \mid \mathcal{F}_{\tau_{\bh}}) > 0$. By the exchangeability guaranteed under the exchangeable distribution $\bP_{\exch}$, we have that $(B_1, \ldots, B_{\tau_{\bh}})$ is an exchangeable binary vector containing $n - N_{\tau_{\bh}}^{\Above}$ many $1$'s. Thus, the vector \[(b_1, \ldots, b_{\tau_{\bh}}) = \big(\underbrace{1, \ldots, 1}_{n - s_t \text{ times}}, \underbrace{0, \ldots, 0}_{\tau_{\bh} - n + N_{\tau_{\bh}}^{\Above}}, \underbrace{1, \ldots, 1}_{s_t - N_{\tau_{\bh}}^{\Above} \text{ times}}\big)\] occurs with positive probability. Because $s_t \in \Omega_t$, we have that $b_t = 0$ or $b_t$ is the last $1$ in the first block of $1$'s and consequently $s_t = N_{\tau_{\bh}}^{\Above} + b_{t+1} + \cdots + b_{\tau_{\bh}}$ with $P(B_1=b_1, \ldots, B_{\tau_{\bh}} =b_{\tau_{\bh}} \mid \mathcal{F}_{\tau_{\bh}}) > 0$ as desired.

\subsection{Proof of Proposition~\ref{global-null-prop}}\label{appendix:reward-prob}

\begin{enumerate}[(i)]
    \item The variables $\varepsilon_i^{(t)}(\bb)$ are the e-values for $\pi_{\sort}(1), \ldots, \pi_{\sort}(t)$ if we stop at time $t$ and have $B_1 = b_1, \ldots, B_t = b_t$. Since all other e-values will be equal to zero, per the definition~\eqref{stopping-time-based-evals}, the right-hand-side of equation~\eqref{go-opt-def} is indeed equal to the optimal diversity value among all selection sets which are self-consistent with respect to the e-values $e_1^{(t)}, \ldots, e_m^{(t)}$ when $B_1 = b_1, \ldots, B_t = b_t$. 

    \item By the exchangeability guaranteed by the exchangeable distribution, the vector $(B_1, \ldots, B_t)$ is uniform on the set $\mathcal{B}(t,s_t)$ when $N_t^{\Above} = s_t$. Equation~\eqref{mble-rewards} is then precisely the expectation with respect to this uniform distribution on $(B_1, \ldots, B_t)$.
\end{enumerate}

\subsection{Proof of Proposition~\ref{snell-construction-prop}}\label{appendix:snell-construction}
By the definition of the Snell envelope update given in equation~\eqref{snell}, it suffices to show that $\bP_{\exch}(B_t=1 \mid \mathcal{F}_{t}) = \frac{n-N_t^{\Above}}{t}$ for each $ t \leq \tau_{\bh}$. Because, for each $t \leq \tau_{\bh}$, $(B_1, \ldots, B_{t})$ are exchangeable conditional on $\mathcal{F}_t$ and independent from $\Zsort$ under the exchangeable distribution $\bP_{\exch}$, this follows immediately by applying \citet[][Lemma E.2]{mary2022semi} (which is restated below in Lemma~\ref{lemma:special-case} for the reader's convenience) conditional on $\mathcal{F}_{\tau_{\bh}}$.

\begin{lemma}[Special case of {\citet[][Lemma E.2]{mary2022semi}}]\label{lemma:special-case}
    Let $(u_1, \ldots, u_{n+1}) \in \{0,1\}$ be an exchangeable vector containing $n$ many $1$'s. Let $V = \sum_{s=1}^{t+1} u_s$. Then \[P(u_{t+1} = 1\mid u_{t+2}, \ldots, u_{n+1}) = \frac{V}{t+1}.\]
\end{lemma}

\section{Auxiliary results and proofs}
\subsection{Proofs for Section~\ref{dacs-underrep}: underrepresentation index}
In this section, we present proofs for results in the main text regarding the application of DACS to the underrepresentation index as discussed in Section~\ref{dacs-underrep}.
\subsubsection{Proof of Proposition~\ref{underrep-opt-value-closed-form}}\label{app:underrep-opt-value-closed-form}
We break into the three cases on the right-hand-side of equation~\eqref{eq:underrep-reward-closed-form}.
\begin{enumerate}
    \item Case 1: $\mathop{\min}_{c\in[C]}  N_t^{c, \testbelow} \geq K_t/C$. Consider the selection set $\mathcal{R}$ which contains exactly $\lceil K_t/C \rceil$ candidates of each category $c \in [C]$, which is feasible since $\mathop{\min}_{c\in[C]}  N_t^{c, \testbelow} \geq K_t$. Furthermore, it is straightforward to see that it achieves the maximal underrepresentation index value of $1/C$. 
    Therefore, it remains to show that such a set satisfies self-consistency. This follows from the fact that $K_t/C = \big\lceil\frac{m}{e^{(t)}_i\alpha}\big\rceil/C\geq \frac{m}{e^{(t)}_i\alpha C}$ for any $i \in \mathcal{R}$ and hence \[|\mathcal{R}| = C\lceil K_t/C\rceil \geq \frac{m}{\alpha e_i^{(t)}}.\] Rearranging yields the self-consistency inequality.

    \item Case 2: $n - N_t^{\Above} > \rho_t$. To show that $O_t=-1/C$, it suffices to show that the only self-consistent selection set with respect to the e-values $e_1^{(t)}, \ldots, e_m^{(t)}$ is the empty set. Observe that $t - n + N_t^{\Above}$ is equal to the number of test points including and past time $t$. Also notice that $\rho_t \leq t$. We break into two further subcases: (1) $\rho_t < n - N_t^{\Above} = t$ and (2) $\rho_t < n - N_t^{\Above} < t$. In case (1), the number of test points $t - n + N_t^{\Above}$ is equal to zero. Hence, all the e-values are equal to zero and the only self-consistent selection set is indeed the empty set. In case (2), the empty set will be the only self-consistent rejection set if and only if the inequality \begin{equation}\label{app:self-consistency-violate}\frac{n+1}{1+ n-N_t^{\Above}} < \frac{m}{\alpha(t-n+N_t^{\Above})}\end{equation} holds, since $t-n+N_t^{\Above}$ is the largest possible size of a selection set containing non-zero e-values. Rearranging, this inequality is equivalent to $n - N_t^{\Above} > \rho_t$.

    \item Case 3: $\mathop{\min}_{c\in[C]} N_t^{c, \testbelow} < K_t/C$ and $n - N_t^{\Above} \leq \rho_t$. In this case, we first  show that $O_t \geq \mathop{\min}_{c\in[C]} N_t^{c, \testbelow} / K_t$, by finding a selection set that is self-consistent with respect to $e_1^{(t)}, \ldots, e_m^{(t)}$ whose underrepresentation index value is of this value. We have that \begin{equation}\label{app:num-test-below-equality}\sum_{c=1}^C N_t^{c, \testbelow} = t - n + N_t^{\Above}.\end{equation}
    Also, $n - N_t^{\Above} \leq \rho_t$ implies $t - n + N_t^{\Above} \geq K_t$. This implies, by way of equation~\eqref{app:num-test-below-equality} as well as the fact that $\mathop{\min}_{c\in[C]} N_t^{c, \testbelow} < K_t/C$, that there exist values $\{M_c\}_{c\in[C]}$ such that $\mathop{\min}_{c\in [C]}  N_t^{c, \testbelow} \leq M_c \leq N_t^{c, \testbelow}$ and $\sum_{c=1}^C M_c = K_t$. The selection set $\mathcal{R}$ which contains (any arbitrary) $M_c$ candidates belonging to category $c$ whose scores are below $W_{(t)}$ has underrepresentation index value greater than or equal to $\mathop{\min}_{c\in[C]} N_t^{c, \testbelow} / K_t$. Furthermore, such a selection set is self-consistent because $|\mathcal{R}| = K_t \geq \frac{m}{e_i^{(t)}\alpha}$. 

    Now we show that $O_t \leq \mathop{\min}_{c\in [C]} N_t^{c, \testbelow} / K_t$ for which it suffices to show that the underrepresentation index of any $(e_1^{(t)}, \ldots, e_m^{(t)})$-self-consistent selection set is at most $\mathop{\min}_{c\in [C]} N_t^{c, \testbelow} / K_t$. First, observe that any $(e_1^{(t)}, \ldots, e_m^{(t)})$-self-consistent selection set $\mathcal{R}$ must obey $|\mathcal{R}| \geq K_t$ due to the condition $|\mathcal{R}| \geq m/(\alpha e_i^{(t)})$. Second, the number of occurrences of the least-represented category in $\mathcal{R}$ is at most $\mathop{\min}_{c\in [C]} N_t^{c, \testbelow}$. Hence we must have $O_t \leq \mathop{\min}_{c\in [C]} N_t^{c, \testbelow} / K_t$ as desired.
\end{enumerate}

\subsection{Additional technical results for Section~\ref{dacs-underrep}: underrepresentation index}\label{appendix:underrep-additional}
In this section, we present additional technical results for DACS' computation for the underrepresentation index.

\subsubsection{Reward function computation}\label{appendix:underrep-additional-reward}
Proposition~\ref{underrep-opt-value-closed-form} enables us to express the rewards $R_t$ in terms of expectations under multivariate hypergeometric distributions:
\begin{corollary}\label{underrep-reward-closed-form}
    For each $t \in [\tau_{\bh}]$ and $c \in [C]$, define $N_t^c := \sum_{s=1}^t \mathds{1}\left\{Z_s = c\right\}$ to be the $\Zsort$-measurable (hence fixed, by Remark~\ref{global-conditioning}) random variable equal to the number of points past time $t$ belonging to category $c$. Furthermore, for each fixed $s_t\in \Omega_t$, let $(H^{(t,s_t)}_1, \ldots, H^{(t,s_t)}_C) \sim \text{MultiHypergeom}\left(t-n+s_t, (N_t^1, \ldots, N_t^C)\right)$ be a draw from the multivariate hypergeometric distribution.\footnote{The multivariate hypergeometric distribution $\text{MultiHypergeom}(t-n+s_t, (N_t^1, \ldots, N_t^C))$ is the joint distribution of colors of sampled marbles when $t-n+s_t$ marbles are sampled uniformly without replacement from an urn containing $C$ different colors that starts off with $N_t^c$ many marbles of color $c$.} Let \[S_{t,s_t}(\nu) := \bP\big(\mathop{\min}_{c\in[C]} H_c^{(t,s_t)} \geq \nu\big)\] denote the survival function of $\mathop{\min}_{c\in[C]} H_c^{(t,s_t)}$. Then
    \begin{equation}\label{eq:multivar-hypergeom-reward-survival}
        R_t(s_t) = \mathds{1}\{n-s_t\leq\rho_t\}\left(\frac{S_{t,s_t}\left(\lceil K_t(s_t)/C\rceil\right)}{C} + \frac{\sum_{\nu=1}^{\lceil K_t(s_t)/C\rceil-1}S_{t,s_t}(\nu)}{K_t(s_t)}\right) - \frac{\mathds{1}\{n-s_t>\rho_t\}}{C},
    \end{equation} 
    where $K_t(s_t) := \left\lceil\frac{m(1+n-s_t)}{\alpha (n+1)}\right\rceil$ is value of the random variable $K_t$ defined in Proposition~\ref{underrep-opt-value-closed-form} when $N_t^{\Above} = s_t$.
\end{corollary}
\begin{proof}
    Define the events \[\mathcal{E}_1^{(t,s_t)} := \left\{\mathop{\min}_{c\in[C]} H_c^{(t,s_t)} \geq \lceil K_t(s_t)/C\rceil\right\}, \quad \mathcal{E}_2^{(t,s_t)} := \left\{n-s_t > \rho_t\right\}, \quad \mathcal{E}_3^{(t,s_t)} := \left(\mathcal{E}_1^{(t,s_t)} \cup \mathcal{E}_2^{(t,s_t)}\right)^c\] corresponding to the three cases in equation~\eqref{eq:underrep-reward-closed-form}. Then, the exchangeability under $\bP_{\exch}$ implies that $(N_t^{c, \testbelow})_{c=1}^C$ is indeed a sample from $\text{MultiHypergeom}\left(t-n+s_t, (N_t^1, \ldots, N_t^C)\right)$ when $N_t^{\Above} = s_t$ and hence by definition, we have 
    \begin{equation}\label{eq:multivar-hypergeom-reward}
        R_t(s_t)= \mathbb{E}\left[\frac{1}{C} \cdot \left(\mathds{1}\big\{\mathcal{E}_1^{(t,s_t)}\big\} - \mathds{1}\big\{\mathcal{E}_2^{(t,s_t)}\big\}\right)+ \frac{\mathop{\min}_{c\in[C]} H_c^{(t,s_t)}}{\lceil{m(1+n-s_t)}/{\alpha (n+1)}\rceil} \cdot \mathds{1}\big\{\mathcal{E}_3^{(t,s_t)}\big\}\right].
    \end{equation}

    Equation~\eqref{eq:multivar-hypergeom-reward-survival} then follows from the fact that the expectation of a non-negative integer random variable is equal to the infinite sum of its survival function.
\end{proof}

Corollary~\ref{underrep-reward-closed-form} therefore tells us that, to compute the rewards at each $s_t$ value in $\Omega_t$, it suffices to compute the survival function $S_{t,s_t}(\cdot)$ of the the minimum value of a draw from $\text{MultiHypergeom}\Big(t-n+s_t,$$ (N_t^1, $$\ldots, $$N_t^C)\Big)$ (recall, in view of Remark~\ref{global-conditioning}, that DACS conditions on $\mathcal{F}_{\tau_{\bh}}$ and hence the values $(N_t^1, \ldots, N_t^C)$ are fixed as they are a measurable function of $\Zsort$).

\paragraph{Computing $S_{t,s_t}(\cdot)$} We use an approach from \cite{lebrun2013efficient}, who studies computation of certain probabilities for various discrete multivariate distributions. The key observation enabling \cite{lebrun2013efficient}'s computation of the survival function of $\mathop{\min}_{c\in[C]} H_c^{(t,s_t)}$ is that a multivariate hypergeometric distribution is the conditional distribution of \emph{independent} Binomial distributions given their sum. More specifically, \begin{equation}\label{mvh-binom}\left(H^{(t,s_t)}_1, \ldots, H^{(t,s_t)}_C\right) \overset{d}{=} \left(M^{(t)}_1, \ldots, M^{(t)}_C\right) ~\Big|~ \sum_{c=1}^C M^{(t)}_c = t-n+s_t,\end{equation} where $M^{(t)}_c \overset{\text{ind}}{\sim} \text{Binom}(N_t^c, 1/2), c = 1, \ldots, C,$ are independent Binomial random variables. Using the relation~\eqref{mvh-binom}, \cite{lebrun2013efficient} uses Bayes' theorem to find \begin{equation}\label{min-cdf}
S_{t,s_t}(\nu) = \frac{\overbrace{\bP\left(\sum_{c=1}^C M_c^{(t)} = t-n+s_t \mid M_c^{(t)} \geq \nu, \forall c \in [C]\right)}^{(\RN{1})}\overbrace{\prod_{c=1}^C\bP\left(M_c^{(t)} \geq \nu\right)}^{(\RN{2})}}{\bP\left(\sum_{c=1}^C M^{(t)}_c = t-n+s_t\right)}.
\end{equation} 
Term $(\RN{2})$ in the numerator of equation~\eqref{min-cdf} can be readily calculated as it only involves the survival function of a one-dimensional Binomial distribution which existing software is able to compute efficiently. Similarly, since $\sum_{c=1}^C M^{(t)}_c \sim \text{Binom}(t,1/2)$, the denominator is also readily computed as it simply involves the Binomial PMF. Term $(\RN{1})$ in the numerator is the probability that a sum of independent \emph{truncated} Binomial random variables takes the value $t-n+s_t$. \cite{lebrun2013efficient} observes that this probability can be computed by convolving the PMFs of these truncated Binomial distributions and suggests, as one approach, to compute this convolution using the fast Fourier transform (FFT). 

Algorithm~\ref{dacs-algo} requires to compute the reward functions $R_t(s_t)$ in equation~\eqref{eq:multivar-hypergeom-reward-survival} for each $s_t \in \Omega_t$ and $t \in [\tau_{\bh}]$. This, in turn requires that the survival function value $S_{t,s_t}(\nu)$ be computed for all $s_t \in \Omega_t, \nu \in [\lceil K_t(s_t)/C\rceil]$ and $t \in [\tau_{\bh}]$. Note, however, that since $\mathop{\min}_{c\in[C]} H_c^{(t,s_t)} \leq \min_{c\in [C]}N^c_t $, we have that $S_{t,s_t}(\nu) = 0$ for all $\nu > \min_{c\in [C]} N^c_t $ and hence the survival function need only be computed at each $s_t \in \Omega_t, \nu \in \left[\lceil K_t(s_t)/C\rceil \wedge \min_{c\in [C]} N^c_t\right]$ and $t \in [\tau_{\bh}]$. This means, in particular, that term (I) in equation~\eqref{min-cdf} must be computed for all these values of $(s_t,\nu,t)$. 
Fixing $\nu$ and $t$, the FFT convolution calculation to do so will involve the following steps:
\begin{enumerate}
    \item Observing that $[N_c^t]$ contains the support of the truncated Binomial variable $M_c^{(t)} \mid M_c^{(t)} \geq \nu$, compute, for each $c \in [C]$, the PMFs of the truncated Binomial random variables $\bP(M_c^{(t)}=\omega \mid M_c^{(t)} \geq \nu)$ for all $\omega \in [0:N_c^t]$ and zero-pad them to be of length $1 + N_t^1 + \cdots + N_C^t = t+1$.
    \item For each $c \in [C]$, apply the FFT to the zero-padded PMF $\big(\bP(M_c^{(t,s_t)}=\omega \mid M_c^{(t,s_t)} \geq \nu)\big)_{\omega=0}^{t}$ to obtain the (sign-reversed) characteristic function $\varphi_c^{(t)}$ of the truncated Binomial random variable $M_c^{(t)} \mid M_c^{(t)} \geq \nu$.
    \item Take the elementwise products of the (sign-reversed) characteristic functions $\varphi_c^{(t)}$ across $c \in [C]$ to obtain $\varphi^{(t)}$, the (sign-reversed) characteristic function of the sum of the independent truncated Binomial random variables.
    \item Apply the inverse FFT to $\varphi^{(t)}$ to recover the PMF values \[\left(\bP\left(\sum_{c=1}^C M_c^{(t)} = \omega \mid M_c^{(t)} \geq \nu, \forall c \in [C]\right)\right)_{\omega=0}^{t}.\]
\end{enumerate}

The above calculation does not depend on the value $s_t$. This suggests that, to compute $S_{t,s_t}(\nu)$ for all $\nu \in \left[\lceil K_t(s_t)/C\rceil \wedge \min(N^1_t,\ldots,N^C_t)\right], s_t \in \Omega_t, t \in [\tau_{\bh}]$ we should compute term (I) by iterating over $t \in [\tau_{\bh}]$ and then over $\nu$ in some set containing $\left[\lceil K_t(s_t)/C\rceil \wedge \min_{c\in [C]} N^c_t\right]$ for all $s_t \in \Omega_t$, and then extracting $\bP\left(\sum_{c=1}^C M_c^{(t)} = t-n+s_t \mid M_c^{(t)} \geq \nu, \forall c \in [C]\right)$ from the full PMF vector obtained in Step 4. By definition of the set $\Omega_t$, we have that $s_t \geq \max\big\{ N^{\Above}_{\tau_{\bh}}, n-t\big\}$ for all $s_t \in \Omega_t$ and hence we find, for all supported $s_t$, that \begin{align*}
    \lceil K_t(s_t)/C\rceil &= \left\lceil \Big\lceil \frac{m(1+n-s_t)}{\alpha(n+1)} \Big\rceil/C\right\rceil\\
    &\leq \left\lceil \bigg\lceil \frac{m\big(1+\min\{  t,n-N_{\tau_{\bh}}^{\Above}\}\big)}{\alpha(n+1)} \bigg\rceil/C\right\rceil.
\end{align*} Therefore, we will simply iterate over $\nu = 1, \ldots, \left\lceil \Big\lceil \frac{m(1+\min\{  t,n-N_{\tau_{\bh}}^{\Above}\}}{\alpha(n+1)} \Big\rceil/C\right\rceil \wedge \min_{c\in [C]} N^c_t$. To compute the remaining terms in $S_{t,s_t}(\nu)$ in equation~\eqref{eq:multivar-hypergeom-reward-survival}, we simply iterate over $t$ and $s_t$. Pseudocode for the full computation of $S_{t,s_t}(\nu)$ is given in Algorithm~\ref{alg:mvh-min}.

\begin{algorithm}[h]
  
  \KwInput{BH stopping time $\tau_{\bh}$, number of each category $(N_t^1, \ldots, N_t^C)_{t=1}^{\tau_{\bh}}$ at each time $t \in [\tau_{\bh}]$, and number of calibration points $N_{\tau_{\bh}}^{\Above}$ past BH stopping time}
  \For{$t = 2, \ldots, \tau_{\bh}$}{
      \For{$s_t \in \Omega_t$}{
            Compute and store  $\bP\big(\sum_{c=1}^C M_c^{(t)} = t-n+s_t\big)$ using the PMF for $\sum_{c=1}^C M_c^{(t)} \sim \text{Binom}(t,1/2)$
        }
    \For{$\nu = 1, \ldots, \left\lceil \left\lceil \frac{m(1+\min\{  t,n-N_{\tau_{\bh}}^{\Above}\})}{\alpha(n+1)} \right\rceil/C\right\rceil \wedge \min(N^1_t,\ldots,N^C_t)$}{
        Compute $\bP(M_c^{(t)}=\omega \mid M_c^{(t)} \geq \nu)= \frac{\bP(M_c^{(t)}=\omega)\mathds{1}\{\omega \geq \nu\}}{\bP(M_c^{(t)} \geq \nu)}$ for each $\omega \in [0:t]$ and $c \in [C]$\\
        Apply, for each $c \in [C]$, the FFT to $\big(\bP(M_c^{(t)}=\omega \mid M_c^{(t)} \geq \nu)\big)_{\omega =0 }^t$ to obtain the (sign-reversed) characteristic functions $\varphi_c^{(t)}$ of $M_c^{(t)} \mid M_c^{(t)} \geq \nu$, \\
        Compute the elementwise product $\mathbb{C}^{t+1} \ni \varphi^{(t)} := \bigodot_{c=1}^C\varphi_c^{(t)}$\\
        Applying the inverse FFT to $\varphi^{(t)}$, obtain and store the PMF values $\left(\bP\big(\sum_{c=1}^C M_c^{(t)} = \omega \mid M_c^{(t)} \geq \nu, \forall c \in [C]\big)\right)_{\omega=0}^{t}$\\
        Compute and store  $\bP (M_c^{(t)} \geq \nu )$ using the Binomial survival function for each $c \in [C]$ \\
        \For{$s_t \in \Omega_t$}{
         Using the stored values from lines 3, 8, and 9, set \[S_{t,s_t}(\nu) \gets \frac{\bP\left(\sum_{c=1}^C M_c^{(t)} = t-n+s_t \mid M_c^{(t)} \geq \nu, \forall c \in [C]\right)\prod_{c=1}^C\bP (M_c^{(t)} \geq \nu )}{\bP \big(\sum_{c=1}^C M^{(t)}_c = t-n+s_t \big)}\]
        }
    }
  }

  \KwOutput{ Survival function values $S_{t,s_t}(\nu)$ for $\nu \in [\lceil K_t(s_t)/C\rceil], s_t \in \Omega_t, t \in [\tau_{\bh}]$} 

\caption{Survival function $S_{t,s_t}(\cdot)$ computation}
\label{alg:mvh-min}
\end{algorithm}

Once the survival functions have all been computed, the reward functions are calculated, for each $s_t \in \Omega_t, t \in [\tau_{\bh}]$, as given by equation~\eqref{eq:multivar-hypergeom-reward-survival} in Corollary~\ref{underrep-reward-closed-form}. The following proposition characterizes the worst-case computational complexity of our method to compute the reward functions $R_t(\cdot)$.

\begin{proposition}\label{prop:worst-case-complexity}
    The worst-case complexity for computing the values $R_t(s_t)$ for all $s_t \in \Omega_t, t \in [\tau_{\bh}]$ in Algorithm~\ref{alg:mvh-min} is $\widetilde{O} ((n+m)^2m )$,
    where the tilde suppresses factors logarithmic in $n,m,$ or $C$.
\end{proposition}
\begin{proof}
The main computational bottleneck in computing the survival function values is lines 6--8 of Algorithm~\ref{alg:mvh-min} and hence it suffices to just study the time complexity contribution of these lines (summed across all the values of $t$ and $\nu$ at which they are computed) to obtain a time complexity bound for the survival function computation. First, observe that \[\frac{m(1+n-N_{\tau_{\bh}}^{\Above})}{\alpha(n+1)} \leq t-n + N_t^{\Above} \leq m,\] where the first inequality is by rearranging the definition of the BH stopping time in equation~\eqref{def:bh-st}. Therefore $K_t(s_t) \leq m$ for all $s_t \in \Omega_t, t \in [\tau_{\bh}]$ and the worst-case complexity is at most $O((n+m)m/C)$ times the worst-case time complexity of lines 6--8 among all $(t,\nu)$ pairs. Applying the FFT in line 6 takes $\widetilde{O}(Ct)$ time across all $c \in [C]$, the elementwise multiplication in line 7 requires $O(Ct)$, and the inverse FFT in line 8 takes again $\widetilde{O}(Ct)$ time total across all $c \in [C]$. Thus, the total complexity of lines 6--8 for fixed $(t,\nu)$ is at most $\widetilde{O}(Ct) \leq \widetilde{O}(C(n+m))$; multiplying by the factor $O((n+m)m/C)$ found earlier, the worst-case complexity of computing the survival function values is $\widetilde{O}\left((n+m)^2m\right)$.

Given the survival function values, the computation of the reward functions in equation~\eqref{eq:multivar-hypergeom-reward-survival}, for each $s_t \in \Omega_t, t \in [\tau_{\bh}]$, requires only $O\left(\lceil K_t(s_t)/C\rceil\right) \leq O(m/C)$ time, and hence the total computation time of this step is at most $O\left((n+m)nm\right)$, which is less than the computational complexity of lines 6--8 of Algorithm~\ref{alg:mvh-min}. Hence, the worst-case complexity is $\widetilde{O}\left((n+m)^2m\right)$, as was to be shown.
\end{proof}

Proposition~\ref{prop:worst-case-complexity} tells us that, up to logarithmic factors, DACS' computation for the underrepresentation index is at worst cubic in $n+m$. Perhaps interestingly, apart from a log-factor, there is no dependence on $C$, the number of categories.

\subsubsection{E-value optimization for the underrepresentation index}\label{app:underrep-evalue-opt}
\begin{algorithm}[!ht]
  
  \KwInput{Sets $\mathcal{S}_{t}^{(1), \testbelow},$ $ \ldots,$ $ \mathcal{S}_{t}^{(C), \testbelow}$}
  Compute $\rho_t=\frac{\alpha t(n+1) - m}{\alpha(n+1)+m}$  and initialize $\mathcal{R}^*_{t} \gets \emptyset$\\
 \If{$n - N_t^{\Above}>\rho_t$}{
    Set $\mathcal{R}^*_{t} \gets \emptyset$
 }
 \ElseIf{$C|\mathcal{S}_{t}^{(1), \testbelow}| \geq K_{t}$}{
    Choose (any) subsets $\mathcal{T}_c \subseteq \mathcal{S}_{t}^{(c), \testbelow}$, $c\in[C]$, such that $|\mathcal{T}_1| = \cdots = |\mathcal{T}_C| = |\mathcal{S}_{t}^{(1), \testbelow}|$\\
    Set $\mathcal{R}^*_{t} \gets  \bigcup_{c'=1}^C \mathcal{T}_{c'} $
 }

 \Else{
 Initialize $c=1$\\
 \While{$(C-c+1)|\mathcal{S}_{t}^{(c), \testbelow}| < K_{t}-|\mathcal{R}^*_{t}|$}{
 Update $\mathcal{R}^*_{t} \gets \mathcal{R}^*_{t} \cup \mathcal{S}_{t}^{(c), \testbelow}$, and $c\gets c+1$}
 {Choose (any) subsets $\mathcal{T}_c \subseteq \mathcal{S}_{t}^{(c), \testbelow}, \ldots, \mathcal{T}_C \subseteq \mathcal{S}_{t}^{(C), \testbelow}$ such that (1) $|\mathcal{T}_c| + \cdots +|\mathcal{T}_C| = K_t-|\mathcal{R}^*_{t}|$ and (2) $|\mathcal{T}_c|, \ldots, |\mathcal{T}_C| \geq |\mathcal{S}_{t}^{(c-1), \testbelow}|$\\
  Update $\mathcal{R}^*_{t} \gets \mathcal{R}^*_{t}  \cup \big(\bigcup_{c'=c}^C \mathcal{T}_{c'}\big)$  
  }
 }

  \KwOutput{ Selected set $\mathcal{R}^*_{t}$} 

\caption{E-value optimization for the underrepresentation index}
\label{underrep-algo}
\end{algorithm} Algorithm~\ref{underrep-algo} gives pseudocode for the greedy algorithm that solves the e-value optimization program for the underrepresentation index for the time-$t$ e-values $e_1^{(t)}, \ldots, e_m^{(t)}$. Proposition~\ref{appendix:e-value-underrep-optimality} formally establishes the correctness of this algorithm.

\begin{proposition}\label{appendix:e-value-underrep-optimality}
    The selection set $\mathcal{R}^*_{\tau^*}$ returned by Algorithm~\ref{underrep-algo} at time $\tau^*$ is optimal. In other words, it attains the maximal underrepresentation index value among all $(e_1^{(\tau^*)}, \ldots, e_m^{(\tau^*)})$-self-consistent selection sets.
\end{proposition}
\begin{proof}
    By Proposition~\ref{underrep-opt-value-closed-form}, it suffices to show that \begin{equation}\label{eq:underrep-opt-value-proof}\varphi^{\underrep}(\mathcal{R}^*_{\tau^*}; \bZ) = \begin{cases}
            1/C, & \text{ if } \mathop{\min}_{c\in[C]} N_{\tau^*}^{c, \testbelow} \geq K_{\tau^*}/C,\\
             -1/C, & \text{ if } n-N_{\tau^*}^{\Above} > \rho_{\tau^*},\\
            \mathop{\min}_{c\in[C]} N_{\tau^*}^{c, \testbelow} \mathbin{/} K_{\tau^*}, & \text{ otherwise},
        \end{cases}
    \end{equation} where, as reminder $K_t := \big\lceil \frac{m(1+n-N_t^{\Above})}{\alpha (n+1)} \big\rceil$ and $\rho_t := \frac{\alpha t(n+1) - m}{\alpha(n+1)+m}$. Lines 2--3 and 4--6 of Algorithm~\ref{underrep-algo} occur, respectively, if and only if the second and first cases above occur, and they indeed yield selection sets with underrepresentation index equal to $-1/C$ and $1/C$ respectively. Thus we need only consider the third case, in which Algorithm~\ref{underrep-algo} finds itself executing lines 8--13. 
    For simplicity, we prove the desired result for any $t$ given in the input (i.e., $t=\tau^*$).

    First we prove that the ``while'' statement in line 10 is not true at $c = C$ and hence the while loop will be exited in line 13 at some value $\hat{c} \leq C$ for which $(C-\hat{c}+1)|\mathcal{S}_{t}^{(\hat{c}), \testbelow}| \geq K_{t}-|\mathcal{R}^*_{t}|$. Observe that, at iteration $C$ in the ``while'' loop, the running selection set $\mathcal{R}^*_t$ has cardinality \[\Big|\bigcup_{c=1}^{C-1}\mathcal{S}_t^{(c),\testbelow}\Big| = \sum_{c=1}^{C-1}\Big|\mathcal{S}_t^{(c),\testbelow}\Big|,\] and hence the inequality in the ``while'' statement holds if and only if $\sum_{c=1}^{C} |\mathcal{S}_t^{(c),\testbelow} | < K_t.$ The left hand side is equal to $t-n + N_t^{\Above}$ and hence the above inequality holds if and only if $t-n+N_t^{\Above} = 0$ or (upon rearranging) \[\frac{n+1}{1+ n-N_t^{\Above}} < \frac{m}{\alpha(t-n+N_t^{\Above})}.\] The latter inequality is precisely inequality~\eqref{app:self-consistency-violate} and the proof of case 2 in Section~\ref{app:underrep-opt-value-closed-form} shows that either $t-n+N_t^{\Above} = 0$ or inequality~\eqref{app:self-consistency-violate} occurs if and only if $n-N_t^{\Above} > \rho_t$, a case which we have already eliminated.

    Therefore, the value of $c$ after the ``while'' argument stops holding will take on some value $\hat{c} \in [2:C]$ for which $(C-\hat{c}+1)|\mathcal{S}_{t}^{(\hat{c}), \testbelow}| \geq K_{t}-|\mathcal{R}^*_{t}|$. By construction, the final selection set constructed in lines 12--13 will have cardinality exactly equal to $K_t$ and will have number of elements from the least-represented category equal to \[\Big|\mathcal{S}_t^{(1),\testbelow}\Big| = \mathop{\min}_{c\in[C]} N_{t}^{c, \testbelow}.\] The underrepresentation index value of this selection set is then exactly equal to that given in the third case of the right-hand-side of equation~\eqref{eq:underrep-opt-value-proof}. (It is always possible to find sets $\mathcal{T}_c, \ldots, \mathcal{T}_C$ satisfying conditions (1) and (2) in line 12 because we know that, at iterate $\hat{c}$, $(C-\hat{c}+1)|\mathcal{S}_{t}^{(\hat{c}), \testbelow}| \geq K_{t}-|\mathcal{R}^*_{t}|$ and also that $(C-\hat{c}+1)|\mathcal{S}_{t}^{(\hat{c-1}), \testbelow}| < K_{t}-|\mathcal{R}^*_{t}|$ because the ``while'' loop condition was false at $\hat{c}-1$.)
\end{proof}

\subsection{Additional technical results for Section~\ref{dacs-heuristics}: Monte Carlo and heuristics}

\subsubsection{Approximate FDR control of relaxed self-consistency}\label{appendix:rsc}
In this section, we show that the relaxed self-consistency condition ensures approximate FDR control:

\begin{proposition}\label{appendix:rsc-prop}
    Let $e_1, \ldots, e_m$ be e-values in the sense that $\bE[e_i \mathds{1}\{Y_{n+i} \leq 0\}] \leq 1$ for all $1, \ldots, m$. Suppose that $\bchi \in [0,1]^m$ satisfies relaxed self-consistency with respect to these e-values:
    \begin{align*}
       & \chi_i \leq \frac{\alpha e_i}{m}\sum_{j=1}^m\chi_i, \quad i = 1, \dots, m \\
    &0 \leq \chi_i \leq 1, \quad i = 1, \dots, m
    \end{align*}

    Then the FDR of the selection set $\mathcal{R} := \{i \in [m]: \xi_i = 1\}$ where $\xi_i \overset{\text{ind}}{\sim} \text{Bern}(\chi_i), i = 1,\ldots, m$ is at most $1.3\alpha$.
\end{proposition}

The key technical tool we use to prove Proposition~\ref{appendix:rsc-prop} is the following identity, which comes from a simple modification of an argument from \cite{stackexchange}.

\begin{lemma}\label{lemma:stack-exchange}
    Let $b_i \overset{\text{ind}}{\sim} \text{Bern}(x_i)$ for $i = 1, \ldots, m$ for some fixed sequence of $[0,1]$-valued variables $x_1, \ldots, x_m$. Then \[\mathbb{E}\left[\frac{b_1}{1 \vee \sum_{j=1}^m b_j}\right] = \int_0^1 x_1 \prod_{j=2}^m (x_js + (1-x_j))ds\]
\end{lemma}
\begin{proof}
    We use the same argument as in \cite{stackexchange}, first observing that \[\frac{b_1}{1 \vee \sum_{j=1}^m b_j} = \frac{b_1}{\sum_{j=1}^m b_j}\cdot I[b_1 \neq 0] = \int_0^1 b_1s^{\sum_{j=1}^mb_j-1}I[b_1 \neq 0]ds,\] where $\frac{0}{0}\cdot 0$, if it occurs in the middle expression, is defined to be zero.
    Therefore, using Fubini's theorem along with independence, we find that 
    \begin{align*}
        \mathbb{E}\left[\frac{b_1}{1 \vee \sum_{j=1}^m b_j}\right] &= \int_0^1 \mathbb{E}\left[b_1s^{b_1-1}I[b_1 \neq 0]\right]\mathbb{E}\left[s^{\sum_{j=2}^mb_j}\right]ds\\
        &= \int_0^1 \mathbb{E}\left[\frac{d}{ds}s^{b_1}I[b_1 \neq 0]\right]\prod_{j=2}^m (x_js + (1-x_j))ds
    \end{align*}
    The usual domination conditions for exchanging the derivative and expectation (see, e.g., \citet[][Theorem 16.8]{billingsley2017probability}) are met and so we obtain the desired result.
\end{proof}

We are now ready to prove Proposition~\ref{relaxed-fdr}. 
\begin{proof}[Proof of Proposition~\ref{relaxed-fdr}]
Write the FDR as \begin{align*}
        \fdr &= \mathbb{E}\left[\frac{\sum_{i =1}^m\xi_i\mathds{1}\{Y_{n+i} \leq 0\}}{1 \vee \sum_{i=1}^m\xi_i}\right]
    \end{align*}

    Fixing any $\epsilon > 0$, the above is equal to \begin{equation}\label{appendix:intermediate-bound}\mathbb{E}\left[\frac{\sum_{i =1}^m\xi_i\mathds{1}\{Y_{n+i} \leq 0\}}{1 \vee \sum_{i=1}^m\xi_i}\mathds{1}\left\{\sum_{i=1}^m \chi_i \leq \epsilon\right\}\right] + \mathbb{E}\left[\frac{\sum_{i =1}^m\xi_i\mathds{1}\{Y_{n+i} \leq 0\}}{1 \vee \sum_{i=1}^m\xi_i}\mathds{1}\left\{\sum_{i=1}^m \chi_i > \epsilon\right\}\right]\end{equation}
    Using Lemma~\ref{lemma:stack-exchange}, apply the law of iterated expectations as well as linearity to see that the second term above is equal to
    \begin{align*}
    &\mathbb{E}\left[\mathds{1}\left\{\sum_{i =1}^m\chi_i > \epsilon\right\} \cdot \sum_{i=1}^m \chi_i\mathds{1}\{Y_{n+i} \leq 0\}\int_0^1 \prod_{j \in [m]\backslash i}(\chi_js + (1-\chi_j))ds\right]\\
    &\leq \mathbb{E}\left[\mathds{1}\left\{\sum_{i =1}^m\chi_i > \epsilon\right\} \cdot\sum_{i=1}^m \chi_i\mathds{1}\{Y_{n+i} \leq 0\}\int_0^1 \exp\left((s-1)\sum_{j \in [m]\backslash i}\chi_j\right)ds\right],
    \end{align*} where the inequality is by the fact that $1+x \leq e^x, \forall x \in \mathbb{R}$.

   We will now further decompose, within the inner summation, on whether or not $\sum_{j \in [m]\backslash i} \chi_j> 0$. We find that 
    \begin{align*}
        &\mathbb{E}\left[\mathds{1}\left\{\sum_{i =1}^m\chi_i > \epsilon\right\}\cdot\sum_{i=1}^m \chi_i\mathds{1}\{Y_{n+i} \leq 0\}\int_0^1 \exp\left((s-1)\sum_{j \in [m]\backslash i}\chi_j\right)ds\right]\\
        &= \underbrace{\mathbb{E}\left[\mathds{1}\left\{\sum_{i =1}^m\chi_i > \epsilon\right\} \cdot\sum_{i=1}^m \mathds{1}\left\{\sum_{j \in [m]\backslash i} \chi_j > 0\right\} \cdot \chi_i\mathds{1}\{Y_{n+i} \leq 0\}\int_0^1 \exp\left((s-1)\sum_{j \in [m]\backslash i}\chi_j\right)ds\right]}_{A}\\
        &+ \underbrace{\mathbb{E}\left[\mathds{1}\left\{\sum_{i =1}^m\chi_i > \epsilon\right\} \cdot\sum_{i=1}^m \mathds{1}\left\{\sum_{j \in [m]\backslash i} \chi_j = 0\right\} \cdot \chi_i\mathds{1}\{Y_{n+i} \leq 0\}\int_0^1 \exp\left((s-1)\sum_{j \in [m]\backslash i}\chi_j\right)ds\right]}_{B}
    \end{align*}

    Let us first focus on term B. We have that
    \begin{align*}
        &\mathbb{E}\left[\mathds{1}\left\{\sum_{i =1}^m\chi_i > \epsilon\right\} \cdot\sum_{i=1}^m \mathds{1}\left\{\sum_{j \in [m]\backslash i} \chi_j = 0\right\} \cdot \chi_i\mathds{1}\{Y_{n+i} \leq 0\}\int_0^1 \exp\left((s-1)\sum_{j \in [m]\backslash i}\chi_j\right)ds\right]\\
        &= \mathbb{E}\left[\mathds{1}\left\{\sum_{i =1}^m\chi_i > \epsilon\right\} \cdot\sum_{i=1}^m \mathds{1}\left\{\sum_{j \in [m]\backslash i} \chi_j = 0\right\} \cdot \chi_i\mathds{1}\{Y_{n+i} \leq 0\}\right]\\
        &\leq \mathbb{E}\left[\mathds{1}\left\{\sum_{i =1}^m\chi_i > \epsilon\right\} \cdot\sum_{i=1}^m \mathds{1}\left\{\sum_{j \in [m]\backslash i} \chi_j = 0\right\} \cdot \frac{\chi_i}{\sum_{j \in [m]} \chi_j}\mathds{1}\{Y_{n+i} \leq 0\}\right]\\
        &\leq \mathbb{E}\left[\mathds{1}\left\{\sum_{i =1}^m\chi_i > \epsilon\right\} \cdot\sum_{i=1}^m \mathds{1}\left\{\sum_{j \in [m]\backslash i} \chi_j = 0\right\} \cdot \frac{\alpha e_i}{m}\mathds{1}\{Y_{n+i} \leq 0\}\right]
    \end{align*} where the penultimate inequality is because, on the event $ \{\sum_{i \in [m]} \chi_i >\epsilon, \sum_{j \in [m]\backslash i} \chi_j = 0 \}$ we have that $\chi_i \leq 1= \frac{\chi_i}{\sum_{j \in [m]}\chi_j}.$ The last inequality is because $\bchi$ satisfies relaxed self-consistency.

    Now, we turn to analyzing term A. It is equal to
    \begin{align*}
    &\mathbb{E}\left[\mathds{1}\left\{\sum_{i =1}^m\chi_i > \epsilon\right\} \cdot\sum_{i=1}^m \mathds{1}\left\{\sum_{j \in [m]\backslash i} \chi_j > 0\right\} \cdot \chi_i \mathds{1}\{Y_{n+i} \leq 0\} \cdot \frac{1- e^{-\sum_{j \in [m]\backslash i}\chi_j}}{\sum_{j \in [m]\backslash i}\chi_j}\right]\\
    &= \mathbb{E}\left[\mathds{1}\left\{\sum_{i =1}^m\chi_i > \epsilon\right\} \cdot\sum_{i=1}^m \mathds{1}\left\{\sum_{j \in [m]\backslash i} \chi_j > 0\right\} \cdot\frac{\chi_i}{\sum_{j \in [m]}\chi_j} \mathds{1}\{Y_{n+i} \leq 0\} \cdot \left(1+\frac{\chi_i}{\sum_{j \in [m]\backslash i}\chi_j}\right)\left(1- e^{-\sum_{j \in [m]\backslash i}\chi_j}\right)\right]\\
    &\leq \mathbb{E}\left[\mathds{1}\left\{\sum_{i =1}^m\chi_i > \epsilon\right\} \cdot\sum_{i=1}^m \mathds{1}\left\{\sum_{j \in [m]\backslash i} \chi_j > 0\right\} \cdot\frac{\alpha e_i \mathds{1}\{Y_{n+i} \leq 0\}}{m} \cdot \left(1+\frac{\chi_i}{\sum_{j \in [m]\backslash i}\chi_j}\right)\left(1- e^{-\sum_{j \in [m]\backslash i}\chi_j}\right)\right]
    \end{align*} where the last inequality is because $\bchi$ satisfies relaxed self-consistency.

    Combining terms A and B we have shown that the second term in expression~\eqref{appendix:intermediate-bound} is upper-bounded by \begin{equation}\label{appendix:intermediate2}\mathbb{E}\left[\mathds{1}\left\{\sum_{i =1}^m\chi_i > \epsilon\right\} \cdot\sum_{i=1}^m  \frac{\alpha e_i \mathds{1}\{Y_{n+i} \leq 0\}}{m} \cdot \zeta_i\right],\end{equation} where we define \[\zeta_i := \begin{cases}
        \left(1+\frac{\chi_i}{\sum_{j \in [m]\backslash i}\chi_j}\right)\left(1- e^{-\sum_{j \in [m]\backslash i}\chi_j}\right) &\text{ if } \sum_{j \in [m]\backslash i}\chi_j > 0\\
        1 &\text{ otherwise}
    \end{cases}.\] 

    Substituting the upper bound~\eqref{appendix:intermediate2} into expression~\eqref{appendix:intermediate-bound}, we have shown for all $\epsilon>0$ that \begin{equation}\label{appendix:intermediate3}
                \fdr \leq \mathbb{E}\left[\frac{\sum_{i =1}^m\xi_i\mathds{1}\{Y_{n+i} \leq 0\}}{1 \vee \sum_{i=1}^m\xi_i}\mathds{1}\left\{\sum_{i=1}^m \chi_i \leq \epsilon\right\}\right] + \mathbb{E}\left[\mathds{1}\left\{\sum_{i =1}^m\chi_i > \epsilon\right\} \cdot\sum_{i=1}^m  \frac{\alpha e_i \mathds{1}\{Y_{n+i} \leq 0\}}{m} \cdot \zeta_i\right].
        \end{equation} 

    Observe that the second term in the right-hand-side of inequality~\eqref{appendix:intermediate3} is upper-bounded upon removing the indicator $\mathds{1}\left\{\sum_{i =1}^m\chi_i > \epsilon\right\}$. Additionally, notice that $\mathds{1}\left\{\sum_{i=1}^m \chi_i \leq \epsilon\right\} \overset{a.s.}{\rightarrow} \mathds{1}\left\{\sum_{i=1}^m \chi_i= 0\right\}$ as $\epsilon \downarrow 0$. Hence, the dominated convergence theorem yields, taking $\epsilon \downarrow 0$, that \begin{align}
        \fdr &\leq \mathbb{E}\left[\frac{\sum_{i =1}^m\xi_i\mathds{1}\{Y_{n+i} \leq 0\}}{1 \vee \sum_{i=1}^m\xi_i}\mathds{1}\left\{\sum_{i=1}^m \chi_i = 0\right\}\right] + \mathbb{E}\left[\sum_{i=1}^m  \frac{\alpha e_i \mathds{1}\{Y_{n+i} \leq 0\}}{m} \cdot \zeta_i\right]\\
        &= \mathbb{E}\left[\sum_{i=1}^m  \frac{\alpha e_i \mathds{1}\{Y_{n+i} \leq 0\}}{m} \cdot \zeta_i\right]\label{eq:zeta-inflated}
    \end{align} where the equality is because  we must have $\xi_1 = \cdots = \xi_m = 0$ on the event that $\left\{\sum_{i=1}^m\chi_i = 0\right\}$. Finally, observe that \[\zeta_i \overset{a.s.}{\leq} \max\left(1, \sup_{x>0}\left(1 + \frac{1}{x}\right)\left(1 - e^{-x}\right)\right), \text{ for all } i = 1, \ldots, m.\] We show that the right-hand-side is at most $1.3$ in Lemma~\ref{functional-bound} below, and hence the result follows from the fact that $e_1, \ldots, e_m$ are e-values. 
\end{proof}

\begin{lemma}\label{functional-bound}
    Define the function \[f(x) = \left(1 + \frac{1}{x}\right)\left(1 - e^{-x}\right).\] Then $\sup_{x > 0}f(x) \leq 1.3$.
\end{lemma}

While the conclusion of Lemma~\ref{functional-bound} can be readily checked by any numerical solver, here we provide an analytical proof. 

\begin{proof}
    The derivative of $f(x)$ is \[f'(x) = \frac{e^{-x}(x^2+x-e^x+1)}{x^2}.\] We will study the sign of $f'(x)$ for which it suffices to study the sign of $g(x) := x^2+x-e^x+1$. Defining $x_L := 1.7932$ and $x_R := 1.7933$, a direct calculation shows that $g(x_L) > 0$ and $g(x_R) < 0$. We show that (1) $g(x) > 0$ for all $x \in (0,x_L]$ and (2) $g(x) < 0$ for all $x \in [x_R, \infty)$.

    To show (1), we will decompose the interval $(0, x_L] = (0, 1.3) \cup [1.3, x_L]$ and show the statement on each subinterval. Starting with $x \in (0,1.3)$, observe that a third-order Taylor expansion yields \[e^x = 1 + x + x^2/2 + x^3/6 + e^{\tilde{x}}x^4/24,\] for some $\tilde{x} \in [0,x]$. Thus, we obtain the bound, for $x \in (0,1.3)$ \[g(x) \geq x^2/2 - x^3/6 - \frac{e^{1.3}x^4}{24}.\] The right-hand-side is larger than zero for all $x \in (0,1.3)$ if and only if $12-4x-e^{1.3}x^2 > 0$ for all $x \in (0,1.3)$. A direct calculation shows that the left root of this quadratic is less than $0$ and the right root is above $1.3$, whereupon this inequality follows. For $x \in [1.3, x_L]$, we compute the derivative of $g(x)$: $g'(x) = 2x + 1-e^x$. This is a concave function, which at $x = 1.3$ is seen to be negative by a direct calculation. Since a direct calculation also shows that $g''(1.3) < 0$, this implies that also $g(x) > 0$ for $x \in [1.3,x_L]$ as $g(x_L) > 0$.

    Now, we show (2). Notice that $g'(x_R) < 0$. The inequality then follows from the concavity of $g'$ as well as the fact that $g(x_R) < 0$.
    
    The mean-value theorem then implies that $\sup_{x \in (0, x_L)}f(x) \leq f(x_L)$ and $\sup_{x \in (x_R, \infty)}f(x) \leq f(x_R)$ and hence $\sup_{x > 0}f(x) = \sup_{x \in [x_L, x_R]} f(x)$. A direct calculation yields the following upper-bound on the derivative of $f$ in $[x_L, x_R]$: \[\sup_{x \in [x_L, x_R]}f'(x) \leq \frac{e^{-x_L}(x_R^2+x_R-e^{x_L}+1)}{x_L^2} =: \lambda \leq 2.98 \cdot 10^{-5}.\] Hence, $f$ is $\lambda$-Lipschitz on this interval and hence we obtain the upper-bound \[\sup_{x \in [x_L, x_R]} f(x) \leq \max(f(x_L), f(x_R)) + \lambda (x_R-x_L) \leq 1.3,\] as was to be shown.
\end{proof}

\subsubsection{Approximate FDR control of relaxed self-consistency for standard deterministic multiple hypothesis testing}\label{appendix:rsc-det}
In this section, we show that the approximate FDR control of relaxed self-consistent multiple testing procedures continues to hold in the standard deterministic multiple hypothesis testing setting. In particular, for this section alone, we assume that we are in the standard deterministic multiple hypothesis testing setting and have $m$ hypotheses $H_1, \ldots, H_m$ as well as e-values $e_1, \ldots, e_m$ for these hypotheses. Because our hypotheses are now deterministic, $e_1, \ldots, e_m$ are e-values in the (standard) sense that $\bE[e_i] \leq 1$ for all $i \in \mathcal{H}$ where $\mathcal{H} \subseteq [m]$ is the set of true null hypotheses \citep{wang2022false}. The FDR of a rejection set $\mathcal{R}$, in this setting, is given by the standard definition: \[\fdr := \bE\left[\frac{|\mathcal{R} \cap \mathcal{H}|}{1 \vee |\mathcal{R}|}\right].\]

\begin{proposition}\label{appendix:rsc-prop-det}
    Let $e_1, \ldots, e_m$ be e-values in the sense that $\bE[e_i] \leq 1$ for all $i \in \mathcal{H}$. Suppose that $\bchi \in [0,1]^m$ satisfies relaxed self-consistency with respect to these e-values:
    \begin{align*}
       & \chi_i \leq \frac{\alpha e_i}{m}\sum_{j=1}^m\chi_i, \quad i = 1, \dots, m \\
    &0 \leq \chi_i \leq 1, \quad i = 1, \dots, m
    \end{align*}

    Then the FDR of the selection set $\mathcal{R} := \{i \in [m]: \xi_i = 1\}$ where $\xi_i \overset{\text{ind}}{\sim} \text{Bern}(\chi_i), i = 1,\ldots, m$ is at most $1.3\alpha$. 
\end{proposition}
The proof of Proposition~\ref{appendix:rsc-prop-det} follows by the exact same proof as Proposition~\ref{appendix:rsc-prop} upon replacing the indicator variables $\mathds{1}\{Y_{n+i} \leq 0\}$ with the indicator variables $\mathds{1}\{i \in \mathcal{H}\}$.

\subsubsection{Warm-starting and coupling}\label{appendix:warm-start}

We will now show how the optimal values in Section~\ref{warm-start} are related. To do so, we will make a few benign assumptions on the relaxed diversity metric $\varphi$.

\begin{assumption}\label{perm-invar}
    Let $\pi$ be any permutation of $[t]$ and let $\bchi \in [0,1]^{t}$ for any $t \in [n+m]$. Then $\varphi$ is permutation invariant in the sense that \[\varphi(\bchi; \Zsort) = \varphi(\bchi^{\pi}; \bZ^{(),\pi}),\] where $\bchi^{\pi} := (\chi_{\pi(1)}, \ldots, \chi_{\pi(t)})$ and $\bZ^{(),\pi} := (Z_{\pi\circ\pi_{\sort}(1)}, \ldots, Z_{\pi\circ\pi_{\sort}(t)})$.
\end{assumption}
Assumption~\ref{perm-invar} says that the extended diversity metric is invariant to permutations. This assumption is natural because it is true, by definition, for the non-extended diversity metric which is defined in terms of (unordered) sets of diversification variables. Assumption~\ref{perm-invar} applies for extensions of the Sharpe ratio and Markowitz objective. We make one more mild assumption:

\begin{assumption}\label{zero-doesnt-affect}
    Fix any $t \in [n+m]$. Let $\bchi \in [0,1]^{t}$ be such that $\chi_{i} = 0$ for all $i$ in some list $S$ containing distinct elements of $[t]$. Then the extended diversity metric satisfies \[ \varphi(\bchi; \Zsort) = \varphi(\bchi_{\neg S}; \Zsort_{\neg S}),\] where $\bchi_{\neg S}$ and $\Zsort_{\neg S}$ are the subvectors of $\bchi$ and $\Zsort$ with indices not in $S$. In other words, the diversity does not depend on diversification variables $i$ for which $\chi_i = 0$ (i.e., those which we do not select). 
\end{assumption}
Assumption~\ref{zero-doesnt-affect} also holds, by definition, for the original (i.e., non-extended) diversity metric $\varphi$ and simply says that the diversification variables corresponding to candidates which we are sure not to select do not contribute to the extended diversity measure. Again, Assumption~\ref{zero-doesnt-affect} is satisfied by both the extended Sharpe ratio and Markowitz objective.

The following proposition relates the optimization problems required to compute the relaxed optimal values at two consecutive time steps, namely, $\widehat{O}_{t}(\bb)$ and $\widehat{O}_{t+1}(\bb')$ for $\bb \in \mathcal{B}(t,s_{t+1}+1)$ and $\bb' \in \mathcal{B}(t+1, s_{t+1})$ for values $s_{t+1} \in \Omega_{t+1}, s_{t+1}+1 \in \Omega_t$. Recall that these two relaxed optimal values are defined as \[\widehat{O}_t(\bb) := \bE_{\xi_i \overset{\text{ind}}{\sim} \text{Bern}(\chi^{*,t,\bb}_i),\\ i \in [t]}\Big[\varphi (\boldsymbol{\xi}; \Zsort_{1:t})\Big]\text{ where } \bchi^{*,t,\bb} = \argmax_{\substack{\bchi \in [0,1]^t: \\ \bchi \text{ is RSC w.r.t.~} (\esort_i(\bb))_{i=1}^t}} \varphi (\bchi; \Zsort_{1:t} )\] and \[\widehat{O}_{t+1}(\bb) := \bE_{\xi_i \overset{\text{ind}}{\sim} \text{Bern}(\chi^{*,t+1,\bb'}_i),\\ i \in [t+1]}\Big[\varphi (\boldsymbol{\xi}; \Zsort_{1:t+1} )\Big]\text{ where } \bchi^{*,t+1,\bb'} = \argmax_{\substack{\bchi \in [0,1]^{t+1}: \\ \bchi \text{ is RSC w.r.t.~} (\varepsilon_i^{(t+1)}(\bb'))_{i=1}^{t+1}}} \varphi(\bchi; \Zsort_{1:t+1}).\] 
Here, we recall that $\varepsilon_i^{(t)}(\bb)$ is the $i$-th e-value at time $t$ when the membership in the calibration fold is given by the binary vector $\bb$, i.e., $\varepsilon_i^{(t)}(\bb) = \frac{(1-b_i)(n+1)}{1+b_1+\dots+b_t}$.

In the following, we demonstrate that the optimization programs defining $\bchi^{*,t,\bb}$ and $\bchi^{*,t+1,\bb'}$ are related.

\begin{proposition}\label{warm-start-prop}
Fix any $\bb \in \mathcal{B}(t,s_{t+1}+1), \bb' \in \mathcal{B}(t+1, s_{t+1})$. Define \[\bd := \left(i \in [t]: b'_i=1,b_i=0\right) \text{ and } \bd' := \left(i \in [t+1]: \begin{cases}
    b'_i=0, b_i=1 \text{ if } i \leq t\\
    b'_i = 0 \text{ if } i = t+1
\end{cases} \right) \] to be, respectively, the sorted lists of indices $i \in [t+1]$ on which $(b_i,b'_i) = (0,1)$ and $(b_i,b'_i) = (1,0)$ if we were to right-append a $1$ to $\bb$. The lists $\bd$ and $\bd'$ have the same length, which we denote as $l$. 

For any vector $\bv \in \mathbb{R}^{t+1}$, we define the vector $\bv^{\bd\leftrightarrow\bd'}$ by swapping the $d_j^{\text{th}}$ and $d'_j{}^{\text{th}}$ elements of $\bv$ for each $j \in [l]$; we will write $\bZ^{(),\bd\leftrightarrow\bd'}$ for the same transformation applied to the vector $\Zsort$. Then,
\begin{gather}
    \label{reordered-relaxed-go-def1}\bchi^{*,t,\bb} = \argmax_{\bchi \in [0,1]^t: \bchi \text{ is RSC w.r.t.~}\left( \varepsilon^{(t)}_i(\bb)\right)_{i \in [t]}}\varphi\left(\bchi; \Zsort_{1:t}\right),\\
    \label{reordered-relaxed-go-def2}\big(\bchi^{*,t+1,\bb'}\big)^{\bd\leftrightarrow\bd'}  = \argmax_{\bchi \in [0,1]^t: \bchi \text{ is RSC w.r.t.~}\left(\frac{n-s_{t+1}}{1+n-s_{t+1}}\cdot \varepsilon^{(t)}_i(\bb)\right)_{i \in [t]}}\varphi\Big(\bchi; \big(\bZ^{(),\bd\leftrightarrow\bd'}\big)_{1:t}\Big).
\end{gather}
\end{proposition}

The implication of Proposition~\ref{warm-start-prop} is that, if the two membership vectors $\bb$ and $\bb'$ agree in many of their coordinates, then the lists of disagreements $\bd, \bd'$ will be small, and the objectives in~\eqref{reordered-relaxed-go-def1} and \eqref{reordered-relaxed-go-def2} will be (nearly) the same. Additionally, the constraint set in~\eqref{reordered-relaxed-go-def2} is contained in that of~\eqref{reordered-relaxed-go-def1}.

\begin{proof}
Recall that, in Proposition~\ref{global-null-prop}, we defined $\mathcal{B}(t,s)=\{\bb\in \{0,1\}^t\colon b_1+\dots+b_t=n-s\}$, and for any $\bb\in \{0,1\}^t$, we define $\varepsilon_i^{(t)}(\bb) = \frac{(1-b_i)(n+1)}{1+b_1+\dots+b_t}$.  
Thus, $\bb\in \{0,1\}^{t}$ contains exactly $n-s_{t+1}-1$ many $1$'s, while $\bb'\in \{0,1\}^{t+1}$ contains exactly $n-s_{t+1}$ many $1$'s. 
    Let $\tilde{\bb}$ denote the length $(t+1)$ vector obtained by appending $1$ to $\bb$. Then $\bd$ counts the number of entries for which $\tilde{\bb}$ equals $0$ while $\bb'$ equals $1$ and $\bd'$ counts the number of entries for which $\tilde{\bb}$ equals $1$ and $\bb'$ equals $0$. Because, by construction, $\tilde{\bb}$ and $\bb'$ contain the same number of $0$'s and $1$'s, it then follows that $\bd$ and $\bd'$ have the same length.

    Now we establish the validity of equations~\eqref{reordered-relaxed-go-def1} and \eqref{reordered-relaxed-go-def2}. Equation~\eqref{reordered-relaxed-go-def1} is simply the definition of the optimal solution used to construct the relaxed optimal value $\widehat{O}_t(\bb)$ and hence all that must be shown is \eqref{reordered-relaxed-go-def2}. By the permutation-invariance of $\varphi$ guaranteed by Assumption~\ref{perm-invar}, we have that \[\varphi(\bchi; \bZ_{1:t+1}^{()}) = \varphi\left(\bchi^{\bd\leftrightarrow\bd'}; \big(\bZ^{(),\bd\leftrightarrow\bd'}\big)_{1:t+1}\right)\] for any $\bchi \in [0,1]^{t+1}$. Consequently, \begin{align*}\bchi^{*,t+1,\bb'} &= \argmax_{\bchi \in [0,1]^{t+1}: \, \bchi \text{ is RSC w.r.t.~} (\varepsilon_i^{(t+1)}(\bb'))_{i=1}^{t+1}} \varphi\left(\bchi^{\bd\leftrightarrow\bd'}; \big(\bZ^{(),\bd\leftrightarrow\bd'}\big)_{1:t+1}\right)
    \end{align*} which implies that 
    \begin{align}\label{swapped-argmax}
    \left(\bchi^{*,t+1,\bb'}\right)^{\bd\leftrightarrow\bd'} &= \argmax_{\bchi \in [0,1]^{t+1}: \, \bchi \text{ is RSC w.r.t.~} \left((\varepsilon_i^{(t+1)}(\bb'))_{i=1}^{t+1}\right)^{{\bd\leftrightarrow\bd'}}} \varphi\left(\bchi; \big(\bZ^{(),\bd\leftrightarrow\bd'}\big)_{1:t+1}\right).\end{align} 
    Now, since $\bd\leftrightarrow \bd'$ swaps pairs of coordinates with $(0,1)$ and $(1,0)$ values in $(\bb,\bb')$ to make the two vectors align in their first $t$ coordinates, we have that
    \begin{align*}
b_j = 1 ~\Leftrightarrow ~ (\bb')^{\bd\leftrightarrow\bd'}_j = 1 , \quad \text{for } j \in [t].
    \end{align*}
    
    This, in combination with the definition of the sorted e-values yields:
    \begin{align}\left((\varepsilon_i^{(t+1)}(\bb'))_{i=1}^{t+1}\right)^{{\bd\leftrightarrow\bd'}}_{j} &= \frac{n+1}{1+n-s_{t+1}} \cdot \left(1 - (\bb')^{\bd\leftrightarrow\bd'}_j\right) \\
    &= \frac{n+1}{1+n-s_{t+1}} \cdot \left(1 - b_j\right)\\
    &=\label{eq:equiv-evalues} \frac{n-s_{t+1}}{1+n-s_{t+1}} \cdot \varepsilon_j^{(t)}(\bb),\end{align} for any $j \in [t]$.
    
    In addition, we note that $(\bb')^{\bd\leftrightarrow\bd'}_{t+1}= 1$, and hence \begin{equation}\label{eval-conseq2}\left((\varepsilon_i^{(t+1)}(\bb'))_{i=1}^{t+1}\right)^{{\bd\leftrightarrow\bd'}}_{t+1} = 0.\end{equation} Equation~\eqref{eval-conseq2}, in particular, implies that $\chi_{t+1} = 0$ for any vector $\bchi$ that is RSC w.r.t.~$\left((\varepsilon_i^{(t+1)}(\bb'))_{i=1}^{t+1}\right)^{{\bd\leftrightarrow\bd'}}$, due to the self-consistency constraint. 
    Therefore, substituting~\eqref{eq:equiv-evalues} as well as~\eqref{eval-conseq2} into \eqref{swapped-argmax} then yields equation~\eqref{reordered-relaxed-go-def2} under Assumption~\ref{zero-doesnt-affect}, as was to be shown.
\end{proof} 

The permutation-invariance of the diversity metric guaranteed by Assumption~\ref{perm-invar} implies that \[\widehat{O}_{t+1}(\bb) := \bE_{\xi_i \overset{\text{ind}}{\sim} \text{Bern}\left(\left(\chi^{*,t+1,\bb'}\right)^{\bd \leftrightarrow \bd'}_i\right),\\ i \in [t+1]}\left[\varphi\left(\boldsymbol{\xi}; \left(\bZ^{(),\bd \leftrightarrow \bd'}\right)_{1:t+1}\right)\right]\] and hence to calculate the optimal value $\widehat{O}_{t+1}(\bb)$ it suffices to just obtain the solution $\left(\bchi^{*,t+1,\bb'}\right)^{\bd \leftrightarrow \bd'}$ in equation~\eqref{reordered-relaxed-go-def2}. In Proposition~\ref{warm-start-prop}, the objectives in the optimization programs \eqref{reordered-relaxed-go-def1} and \eqref{reordered-relaxed-go-def2} are similar if $\bd$ and $\bd'$ have short lengths, which occurs if $\bb$ and $\bb'$ agree on most of their first $t$ indices. Furthermore, the feasible set of $\eqref{reordered-relaxed-go-def2}$ is contained in that of \eqref{reordered-relaxed-go-def1}. This is because the e-values in the former are simply rescaled versions of the latter by the factor $\frac{n-s_t}{1+n-s_t} \leq 1$; it is harder to satisfy relaxed self-consistency with smaller e-values. 

\section{Implementation details for PGD solvers}\label{app:warm-start-implement}
In this section, we discuss implementation details for the PGD solvers for the Sharpe ratio and Markowitz objective. Section~\ref{app:feasibility-check-sec} describes a simple feasibility check we perform before running either solver and Section~\ref{app:pgd-implement} discusses implementation details of the solvers.

\subsection{Feasibility check}\label{app:feasibility-check-sec}
Before discussing specific details, we first mention that before we actually run either solver, we perform a feasibility check, which we now describe.

\begin{proposition}[Feasibility check for relaxed self-consistency]\label{feasibility-check-prop}
    Suppose we have e-values $\be := (e_1, \ldots, e_m)$ each of which takes on the value $0$ or some positive value $\beta$ (indeed this is the case for the family of e-values and sorted e-values considered in this paper). Then there exists a non-zero relaxed self-consistent vector $\bchi$ with respect to the e-values $e_1, \ldots e_m$ if and only if $\beta \geq \frac{m}{\alpha \|\be\|_0}$.
\end{proposition}
\begin{proof}
    The ``if'' part is immediate because the condition $\beta \geq \frac{m}{\alpha \|\be\|_0}$ implies that $\bchi = \left(\mathds{1}\{e_i > 0\}\right)_{i=1}^m$ is self-consistent and hence relaxed self-consistent. 
    
    For the ``only if'' direction, suppose, that there exists a non-zero relaxed-self-consistent vector $\widetilde{\bchi}$. Now, by symmetry, any permutation of the non-zero elements of $\widetilde{\bchi}$ is relaxed self-consistent. Because the relaxed self-consistency constraints are linear it follows that $\widetilde{\bchi}^{\text{avg}}$, defined to be the average of all such permutations, is also relaxed self-consistent. Because $\widetilde{\bchi}^{\text{avg}}$'s entries are all either equal to zero or equal to the same positive number, this implies that $\widehat{\bchi}^{\text{avg}} := \left(\mathds{1}\{\widetilde{\chi}^{\text{avg}}_i > 0\}\right)_{i=1}^m$ is a rescaling of $\widetilde{\bchi}^{\text{avg}}$ and hence it is also (relaxed) self-consistent. This implies that $\beta \geq \frac{m}{\alpha \|\widetilde{\bchi}^{\text{avg}}\|_0}$. The result then follows from the inequality $\frac{m}{\alpha \|\widetilde{\bchi}^{\text{avg}}\|_0} \geq \frac{m}{\alpha \|\be\|_0}$, which is implied by the fact that if $e_i = 0$ then $\bchi'_i = 0$ for any (relaxed) self-consistent $\bchi'$.
\end{proof}

Proposition~\ref{feasibility-check-prop} shows that checking (non-relaxed) self-consistency (which is simple and amounts to asking if the value of any non-zero e-value---which all take on the same value in our problem setting---is at least $m$ divided by the product of $\alpha$ and the number of non-zero e-values) suffices to check relaxed self-consistency feasibility. This is precisely what we do for the Markowitz objective. The same statement continues to hold true even if we add the condition that $\mathbf{1}^\top \bchi = 1$ (as is required in the quadratic program for the Sharpe ratio) and so we use the same simple feasibility check for the Sharpe ratio as well.

\subsection{Accelerated PGD with adaptive restarts}\label{app:pgd-implement}
In this section, we discuss how we implement the PGD solvers for the Sharpe ratio and Markowitz objective; we also discuss how to incorporate warm-starts. 

The first observation behind the procedure is that the indices whose corresponding e-values are zero can be dropped from the relaxed e-value optimization problems for the Sharpe ratio and Markowtiz objective (this is a consequence of the fact that both objectives satisfy Assumption~\ref{zero-doesnt-affect}). Consequently, the quadratic programs required to solve these problems are, for some $\kappa>0$ and positive definite $\bS \in \mathbb{R}^{d \times d}$, of the form

\begin{align}
     \label{imp-sharpe0}\text{minimize} \quad & \frac{1}{2}\bx^\top \bS\bx\\
     \label{imp-sharpe1}\text{s.t.} \quad & \bx \leq \kappa \mathbf{1}\\
      \label{imp-sharpe2}  & \bx^\top \mathbf{1} = 1 \\
     \label{imp-sharpe3}   & \mathbf{0} \leq \bx \leq \mathbf{1}
\end{align}
and 
\begin{align}
     \label{imp-mark0}\text{minimize} \quad & \frac{\gamma}{2}\bx^\top \bS\bx - \bx^\top \mathbf{1}\\
     \label{imp-mark1}\text{s.t.} \quad & \bx \leq \kappa \mathbf{1}\mathbf{1}^\top \bx \\
     \label{imp-mark2}   & \mathbf{0} \leq\bx \leq \mathbf{1}
\end{align} respectively, where inequalities between vectors are elementwise.

\subsubsection{Sharpe ratio}

The algorithm we use to solve the program~\eqref{imp-sharpe0}--\eqref{imp-sharpe3} is projected gradient descent with momentum using the adaptive warm-start heuristic of \cite{o2015adaptive}; we include a description here for the reader's convenience. Letting $\mathcal{C}_{\sharpe}$ denote the feasible set defined by~\eqref{imp-sharpe1}--\eqref{imp-sharpe3}, the procedure maintains three state variables $\bx^k, \by^k, \theta_k$ updated as follows: 
\begin{align}
    \label{x-assign}\bx^{k+1} &\gets \proj_{\mathcal{C}_{\sharpe}}\left(\by^k-t_k \bS\by^k\right),\\
    \label{theta-assign}\theta_{k+1} &\gets \frac{1 + \sqrt{1+4\theta_k^2}}{2},\\
    \label{beta-assign}\beta_{k+1} &\gets \frac{\theta_k-1}{\theta_{k+1}},\\
    \label{y-assign}\by^{k+1} &\gets \bx^{k+1} + \beta_{k+1}(\bx^{k+1}-\bx^{k}),
\end{align} where $\proj_{\mathcal{C}}(\bz)$ denotes the projection of $\bz$ onto the set $\mathcal{C}$.

The initializations are given by $\theta_0 = 1$ and $\by^0 \gets \bx^0$ where the initial value $\bx^0$ is an input to the program; our warm-starting procedure simply sets $\bx^0$ to be the solution to the previous optimization program. The step-size $t_k$ in equation~\eqref{x-assign} is chosen by a standard backtracking line search. The final ingredient is that we use the adaptive restart mechanism of \cite{o2015adaptive} which sets $\by^k \gets \bx^k$ and $\theta_k=1$ if \[(\by^k-\bx^{k+1})^\top(\bx^k-\bx^{k-1})>0.\] This procedure is essentially 
the exact same as Algorithm 6 in \cite{o2015adaptive} with adaptive restarts. It is well known that the projection $\proj_{\mathcal{C}}$ onto the feasible set defined by Sharpe ratio constraints can be computed efficiently; see, for example, \citet{kiwiel2008breakpoint} and \citet{helgason1980polynomially}. We use an approach similar to that given in \cite{held1974validation}.

\subsubsection{Markowitz objective and the general relaxed e-value consistency constraint set}
We use the same PGD algorithm described above for our Markowitz solver (though, of course, the gradient step in~\eqref{x-assign} is now given by $\by^k-t_k (\gamma\bS\by^k-\mathbf{1})$ and we now project onto the Markowitz feasible set given by~\eqref{imp-mark1}-\eqref{imp-mark2}). In principle this PGD algorithm can be used for any relaxed e-value optimization program so long as it is of the same form as the program~\eqref{relaxedsc0}--\eqref{relaxedsc2} in the main text (i.e., the original diversity metric $\varphi$ is concave and we do not need to introduce constraints and/or variables to make the problem convex) and gradients of the extended metric $\varphi$ can be computed efficiently. This is because the Markowitz constraint set given by~\eqref{imp-mark1}--\eqref{imp-mark2} is precisely the same as the relaxed self-consistency constraints~\eqref{relaxedsc1}--\eqref{relaxedsc2}. We shall refer to this constraint set as $\mathcal{C}_{\rsc}$, the RSC constraint set.

Projection onto the constraint set $\mathcal{C}_{\rsc}$ is more complicated than that for the Sharpe ratio (apart from some special cases of $\kappa$ that we describe later). For this reason, we introduce a variable $s = \mathbf{1}^\top \bx \in \RR$ and consider the constraint set $\mathcal{C}_s$ indexed by $s$ defined as
\begin{align}
     \label{timp-mark1}   & \mathbf{0} \leq\bx \leq \min(\kappa s,1)\mathbf{1}\\
     \label{timp-mark2}   & \mathbf{1}^\top \bx = s,
\end{align} so that $\mathcal{C}_{\rsc} = \bigcup_{s \in [0,d]} \mathcal{C}_s$. The constraint set $\mathcal{C}_s$ is of essentially the same form as that for the Sharpe ratio and \cite{kiwiel2008breakpoint} show that the projection of $\by$ onto $\mathcal{C}_s$ is given by \begin{equation}\label{opt-soln-s-mu}x^*_i(s) := \clip_0^{\min(\kappa s,1)}\left(y_i-\mu(s)\right)\; \forall i \in [d], \text{ where } \mu(s) \in \RR \text{ satisfies}  \sum_{i=1}^d \clip_0^{\min(\kappa s,1)}\left(y_i-\mu(s)\right) = s,\end{equation} and $\clip_a^b(c) := \min(\max(c,a),b)$. Projecting $\by$ onto the RSC constraint set then amounts to finding $\bx^*(s^*)$ where \begin{equation}\label{s-equation}s^* \in \argmin_{s \in [0,d]}\frac{1}{2}\|\bx^*(s)-\by\|^2.\end{equation} In the discussion that follows, we assume that $1/d \leq \kappa \leq 1$ since if $\kappa < 1/d$ the only feasible solution is $\boldsymbol{0}$ and if $\kappa \geq 1$, a direct calculation shows that the projection is simply given by $\clip_0^1(\by)$.

 Our approach will be to split the feasible region $[0,d]$ of $s$ into two sub-intervals $\mathcal{I}_\low := [0,1/\kappa], \mathcal{I}_{\high} := [1/\kappa, d]$, find optimal solutions $\bx^*(s^*_{\low}), \bx^*(s^*_{\high})$ in each (here $s^*_{\low}, s^*_{\high}$ are the solutions to equation~\eqref{s-equation} when the full interval $[0,d]$ over which the minimum is taken is replaced by $\mathcal{I}_\low, \mathcal{I}_{\high}$, respectively), and then return whichever is closer to $\by$. 
 Algorithm~\ref{projection-algo} gives pseudocode for our algorithm to find $\bx^*(s^*_\low)$ (a similar procedure can be used to find $\bx^*(s^*_\high)$). In it and the discussion that follows, weak inequalities involving variables equal to infinity should be interpreted as being strict.

\begin{algorithm}[!ht]
  
  \KwInput{$\kappa>0, \by \in \mathbb{R}^d$}
  Sort $y_{(1)} \leq \cdots \leq y_{(d)}$, breaking ties arbitrarily, and define $y_{(0)} := -\infty, y_{(d+1)} := \infty$\\
  Compute and store the values \begin{equation}\label{precomputed}\sum_{j=1}^i y_{(j)}, \sum_{j=1}^i y_{(j)}^2, \text{ for all } i \in [d]\end{equation}\\
  Initialize running minimum $\mathcal{V}^*_\running \gets \infty$ and corresponding variables $s^*_\running \gets \varnothing, \mu^*_\running \gets \varnothing$\\
  Initialize $r \gets 1$\\
 \For{$\ell=1,\ldots,d$}{
    \While{$y_{(r)} < y_{(\ell)}+1$ and $r < d$}{
    $r \gets r + 1$
    }
    \For{$k=\ell,\ldots,r$}{
        Using the precomputed quantities from~\eqref{precomputed}, compute the values $L_{\ell,k}, U_{\ell,k}$  where $[L_{\ell,k}, U_{\ell,k}]$ equals the interval of $s$ values specified by the constraints
         \begin{gather*}
                y_{(\ell-1)} \leq \frac{\sum_{j=\ell}^{k} y_{(j)} + \kappa s(d-k)-s}{k-\ell+1} \leq y_{(\ell)} \\
                y_{(k)} \leq \frac{\sum_{j=\ell}^{k} y_{(j)} + \kappa s(d-k)-s}{k-\ell+1} + \kappa s\leq y_{(k+1)}\\
                0 \leq s \leq 1/\kappa
    \end{gather*}\\
    \If{$L_{\ell,k} \leq U_{\ell,k}$ (i.e., the above constraints are feasible)}{
        Again using the precomputed values from~\eqref{precomputed}, calculate the (unconstrained) minimizer $\tilde{s}_{\ell,k}$ of the following quadratic function in $s$: \begin{equation*}
            Q_{\ell,k}(s) := \sum_{j=1}^{\ell-1}y_{(j)}^2 + \sum_{j=k+1}^{d}(\kappa s-y_{(j)})^2 + (k-\ell+1)\left(\frac{\sum_{j=\ell}^{k} y_{(j)} + \kappa s(d-k)-s}{k-\ell+1}\right)^2
        \end{equation*}\\
        Clip to obtain the constrained optimal solution: $s^*_{\ell,k} \gets \clip_{L_{\ell,k}}^{U_{\ell,k}}(\tilde{s}_{\ell, k})$\\
        Compute the value of the quadratic at the clipped solution: $\mathcal{V}_{\ell,k} \gets Q_{\ell,k}(s^*_{\ell,k})$\\
        \If{$\mathcal{V}_{\ell,k} < \mathcal{V}^*_\running$}{
        Compute $\mu(s^*_{\ell,k})$ using the precomputed quantities from~\eqref{precomputed}: \[\mu_{\ell,k} \gets \frac{\sum_{j=\ell}^{k} y_{(j)} + \kappa s^*_{\ell,k}(d-k)-s^*_{\ell,k}}{k-\ell+1}\]\\
        Set $\mathcal{V}^*_\running \gets \mathcal{V}_{\ell,k}$, $s^*_\running \gets s^*_{\ell,k}$, $\mu^*_\running \gets \mu_{\ell,k}$
        }
        }
    }
 }
 Set $x^*(s^*_\running)_i \gets \clip_0^{\kappa s^*_\running}(y_i-\mu^*_\running)$ for all $i=1, \ldots, d'$\\
  \KwOutput{ $\bx^*(s^*_\running)$} 

\caption{Projection onto $\mathcal{I}_\low$}
\label{projection-algo}
\end{algorithm}

The intuition behind Algorithm~\ref{projection-algo} is that it searches over all possible intervals---created by the sorted $\by$ vector---in which $\mu(s^*_\low)$ and $\mu(s^*_\low) + \kappa s^*_\low$ can lie. For any given pair of such intervals, finding the optimal $s$ for which $\mu(s)$ and $\mu(s) + \kappa s$ lies in these intervals can be obtained by minimizing a quadratic formula and then clipping the solution to an appropriate interval. A running minimum is kept as the algorithm iterates over the intervals. We now give a proof of correctness of Algorithm~\ref{projection-algo}.
\begin{proof}[Proof of correctness of Algorithm~\ref{projection-algo}]
    Any solution $s^*_\low$ satisfies, by equation~\eqref{opt-soln-s-mu}, \begin{equation}\label{restated}\sum_{i=1}^d \clip_0^{\kappa s^*_\low}\left(y_i-\mu(s^*_\low)\right) = s^*_\low.\end{equation} Now, for each $\ell \in [d]$, define $r_\ell$ to be the least index $r \leq d$ such that $y_{(r)} \geq y_{(\ell)}+1$ if such an index exists or $d$ if no such index exists, where $y_{(r)}$ is the $r$-th order statistic of $\{y_i\}$; $r_\ell$ is then precisely the value of $r$ just before line $8$ begins in Algorithm~\ref{projection-algo}. Then, for any $\ell \in [d]$, we claim that there exists $k \in [\ell:r_\ell]$ such that \begin{equation}\label{winequalities}y_{(\ell-1)} \leq \mu(s^*_\low) \leq y_{(\ell)} \text{ and } y_{(k)} \leq \mu(s^*_\low) + \kappa s^*_\low \leq y_{(k+1)}\end{equation} for some choice $\mu(s^*_\low)\in \RR$ satisfying~\eqref{restated} and some optimal $s^*_\low$. The only way in which~\eqref{winequalities} could be potentially violated is if, for some $\ell \in [d]$, we have \begin{equation}\label{weird-case}y_{(\ell-1)} \leq \mu(s^*_\low) \leq y_{(\ell)} \text{ and } y_{(\ell-1)} \leq \mu(s^*_\low) + \kappa s^*_\low \leq y_{(\ell)},\end{equation} (for $r_\ell < d$, it is impossible for $\mu(s^*_\low) + \kappa s^*_\low > y_{(r_\ell+1)}$, since $\mu(s^*_\low) \leq y_{(\ell)}$ and $\kappa s^*_\low \leq 1$ imply that $\mu(s^*_\low) + \kappa s^*_\low \leq y_{(\ell)} + 1$ which $y_{(r_\ell+1)}$ must be larger than). 
    
    If~\eqref{weird-case} occurs, then~\eqref{restated} implies that one of the following happens: (1)~$y_{(\ell-1)} = y_{(\ell)}$, (2)~$s^*_\low = 0$, or (3)~$\kappa\cdot (d-\ell+1) = 1, y_{(\ell-1)} \neq y_{(\ell)}, s^*_\low \neq 0$. This is because if neither of the first two cases happen, then we must have $\mu(s^*_\low) < y_{(\ell)}$ and $y_{(\ell-1)} \leq \mu(s^*_\low) + \kappa s^*_\low \leq y_{(\ell)}$ in which case the left hand side of~\eqref{restated} becomes $\kappa s^*_\low(d-\ell+1)$. We handle these three cases now:
    \begin{enumerate}
        \item[(1)] If $y_{(\ell-1)} = y_{(\ell)}$, then~\eqref{weird-case} implies \[y_{(\ell-1)} \leq \mu(s^*_\low) \leq y_{(\ell)} \text{ and } y_{(\ell)} \leq \mu(s^*_\low) + \kappa s^*_\low \leq y_{(\ell+1)}\] which is covered by the search among the inequalities~\eqref{winequalities}.
        \item[(2)] If $s^*_\low = 0$, then one valid solution that satisfies~\eqref{restated} is to have $\mu(s^*_\low) = y_{(\ell)}$ in which case \[y_{(\ell-1)} \leq \mu(s^*_\low) \leq y_{(\ell)} \text{ and } y_{(\ell)} \leq \mu(s^*_\low) + \kappa s^*_\low \leq y_{(\ell+1)}\] which is covered in~\eqref{winequalities}.
        \item[(3)] If $\kappa(d-\ell+1) = 1$ and neither of the previous two cases occur, then define $\tilde{\mu}(s^*_\low) := y_{(\ell)}-\kappa s^*_\low$. Then \begin{equation}\label{winequalities-covered}y_{(\ell-1)} \leq \tilde{\mu}(s^*_\low) < y_{(\ell)} \text{ and } y_{(\ell)} \leq \tilde{\mu}(s^*_\low) + \kappa s^*_\low \leq y_{(\ell+1)}\end{equation} and also \[\sum_{i=1}^d \clip_0^{\kappa s^*_\low}\left(y_i-\tilde{\mu}(s^*_\low)\right) =  \kappa(d-\ell+1)s^*_\low = s^*_\low\] where the last equality is because we are in the case that $\kappa(d-\ell+1) = 1$. The solution vector $\bx^*$ implied by $(\tilde{\mu}(s^*_\low),s^*_\low)$ is equal to $\bx^*(s^*_\low)$ (which is the solution implied by $(\mu(s^*_\low),s^*_\low)$) and therefore it is also optimal. The inequalities~\eqref{winequalities-covered} are again covered by~\eqref{winequalities}. Consequently, the inequalities~\eqref{winequalities} must be satisfied for \emph{some} solution pair $(s^*_\low,\tilde{\mu}(s^*_\low))$.
    \end{enumerate}

    The remainder of the proof is straightforward. If the pair $(s^*_\low, \mu(s^*_\low))$ satisfies~\eqref{winequalities} for indices $(\ell^*, k^*)$, then~\eqref{restated} becomes 
    \begin{equation}\label{restated'}
        \sum_{j=\ell^*}^{k^*}(y_{(j)}-\mu(s^*_\low)) + \kappa(d-k^*)s^*_\low = s^*_\low
    \end{equation} which implies that $\mu(s^*_\low)$ must be given by the linear equation $\mu_{\ell^*,k^*}$ in line 15. The inequalities~\eqref{winequalities} on $\mu^*_{\low}(s^*_\low)$ and $s^*_\low$ (as well as the fact that $0 \leq s^*_\low \leq 1/\kappa$) are equivalent to $s^*_\low \in [L_{\ell^*,k^*}, U_{\ell^*,k^*}]$. Because of the fact that $\mu(s^*_\low)$ is linear in $s^*_\low$, the objective $\frac{1}{2}\|\bx^*(s^*_\low)-\by\|^2$ is quadratic in $s^*_\low$, given by the quadratic function $Q_{\ell^*,k^*}$, and the clipped solution $s^*_{\ell^*,k^*}$ found in line 12 is precisely the minimizing value in $[L_{\ell^*,k^*}, U_{\ell^*,k^*}]$ and therefore must coincide with (or at least have the same objective value as) $s^*_\low$. Because $s^*_{\low}$ is an optimal solution of $s$, the value of $\frac{1}{2}\|\bx^*(s^*_{\ell,k})-\by\|^2$ corresponding to $s^*_{\ell,k}$ for any other $\ell,k$ cannot be less than that corresponding to $s^*_{\ell^*,k^*}$ and hence the algorithm indeed returns $\bx^*(s^*_{\low})$ for some optimal $s^*_\low$, as desired.
\end{proof}

\paragraph{Computation} The running time of Algorithm~\ref{projection-algo} is $O(d^2)$. This is because all of the following can be accomplished, for each $(\ell,k)$ pair, in constant time:
\begin{itemize}
    \item Determining the bounds $L_{\ell,k}, U_{\ell,k}$ in line 9 as the computation involves only simple algebraic manipulations using the precomputed quantities in~\eqref{precomputed}.
    \item Finding the unconstrained optimal solution $\tilde{s}_{\ell,k}$ to $Q_{\ell,k}$ in line 11 as it is easily expressed in terms of the coefficients of $Q_{\ell,k}$ which again we have constant-time access to due to the precomputation in line 2.
    \item Computing $\mu_{\ell,k}$, since the calculation again only involves the precomputed quantities from~\eqref{precomputed}.
\end{itemize}

\section{Additional experimental results}\label{app:app-exper-results}
\subsection{Job hiring dataset}\label{app:exper-job-hire}

\subsubsection{Diversity objective}

\begin{figure}
    \centering
    \includegraphics[scale=0.4]{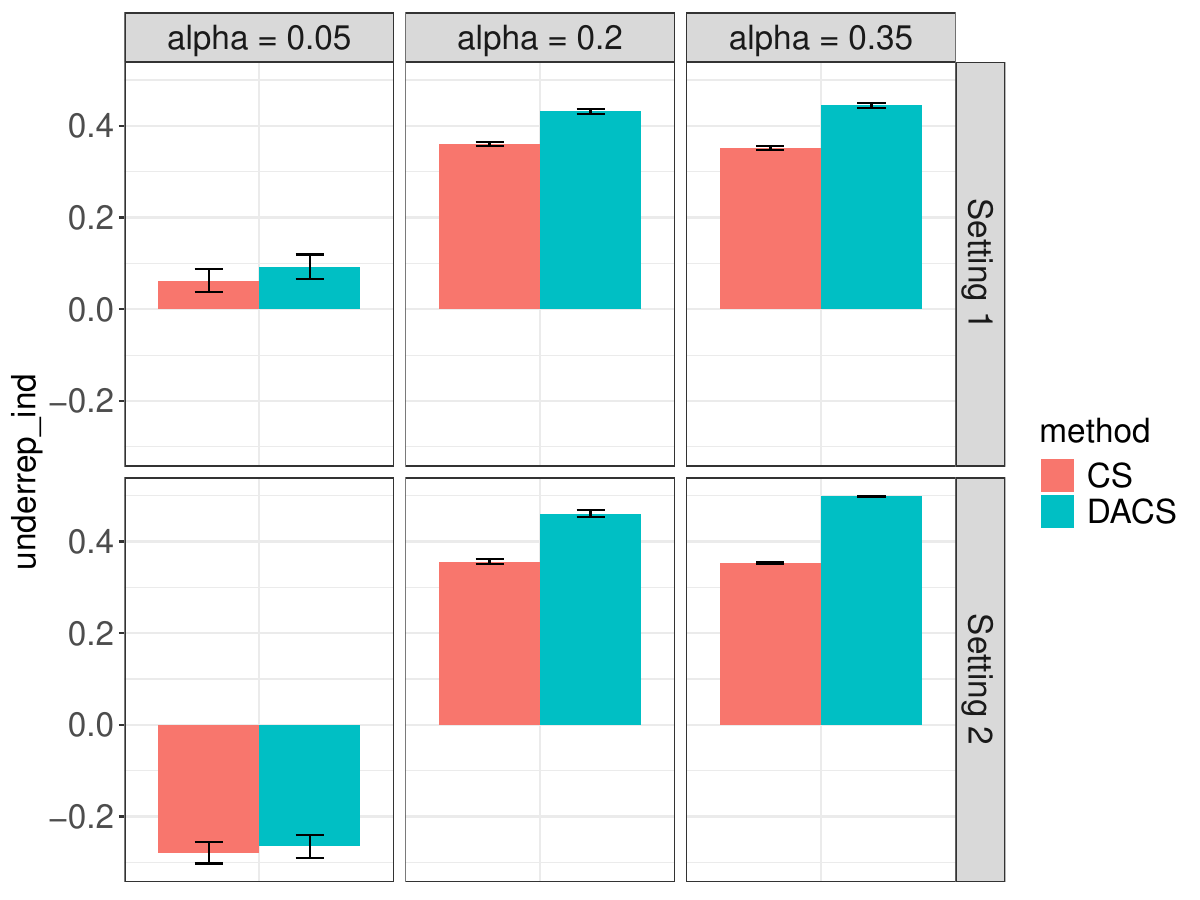}
    \caption{Underrepresentation index of DACS (blue) and CS (red) for various values of $\alpha$ and settings of $n$ and $m$ on job hiring dataset.}
    \label{fig:hiring-div-obj}
\end{figure}

Figure~\ref{fig:hiring-div-obj} shows the average diversity objective for DACS compared to CS as measured by the underrepresentation index. Our method produces selection sets with greater diversity than CS, as judged by the underrepresentation index.

\subsubsection{FDR, power, and number of selections}
\begin{figure}
    \centering
    \includegraphics[scale=0.4]{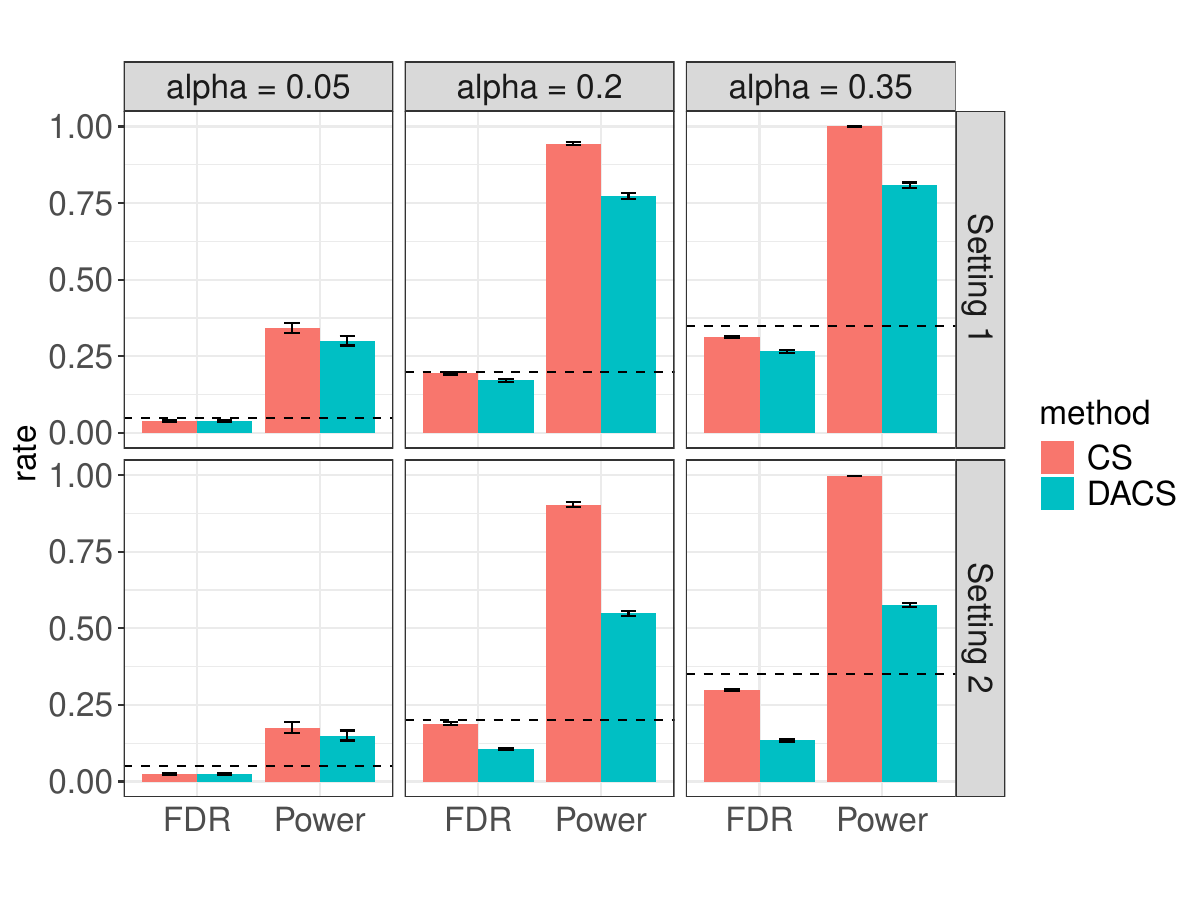}
    \caption{FDR and power of DACS (blue) and CS (red) for various values of $\alpha$ and settings of $n$ and $m$ on job hiring dataset.}
    \label{fig:hiring-fdr-power}
\end{figure}

Figure~\ref{fig:hiring-fdr-power} shows the FDR and power of both CS and DACS using the underrepresentation index. Our method controls FDR below the nominal level. Although it has lower power (since its selection set is a subset of CS'), it produces more diverse selection sets as seen in Figure~\ref{hiring:main-result}. The lower power of DACS is also reflected by the fact that---as we expect---it makes fewer selections as shown in Figure~\ref{fig:hiring-numr}

\begin{figure}
    \centering
    \includegraphics[scale=0.4]{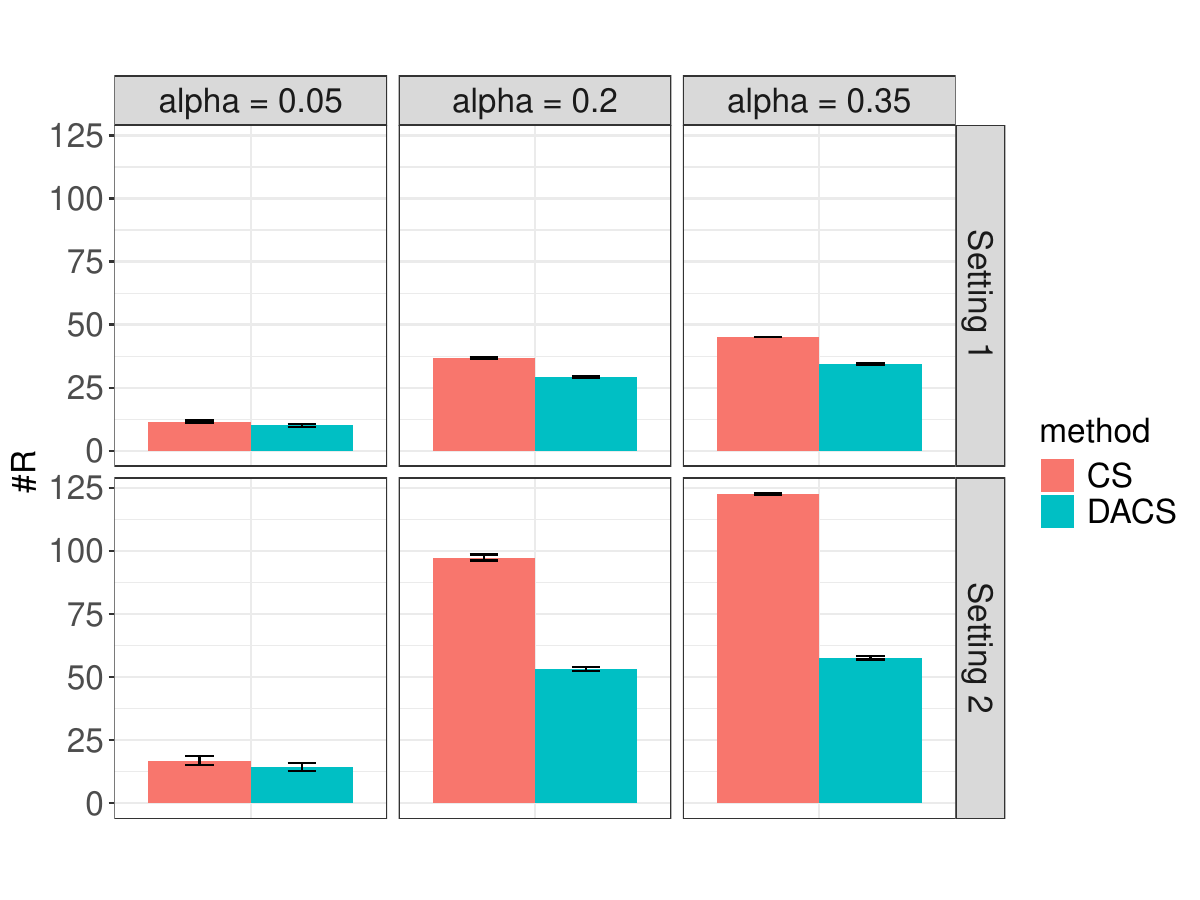}
    \caption{Number of selections made by DACS (blue) and CS (red) for various values of $\alpha$ and settings of $n$ and $m$ on job hiring dataset.}
    \label{fig:hiring-numr}
\end{figure}

\subsubsection{Average time}
\begin{figure}
    \centering
    \includegraphics[scale=0.4]{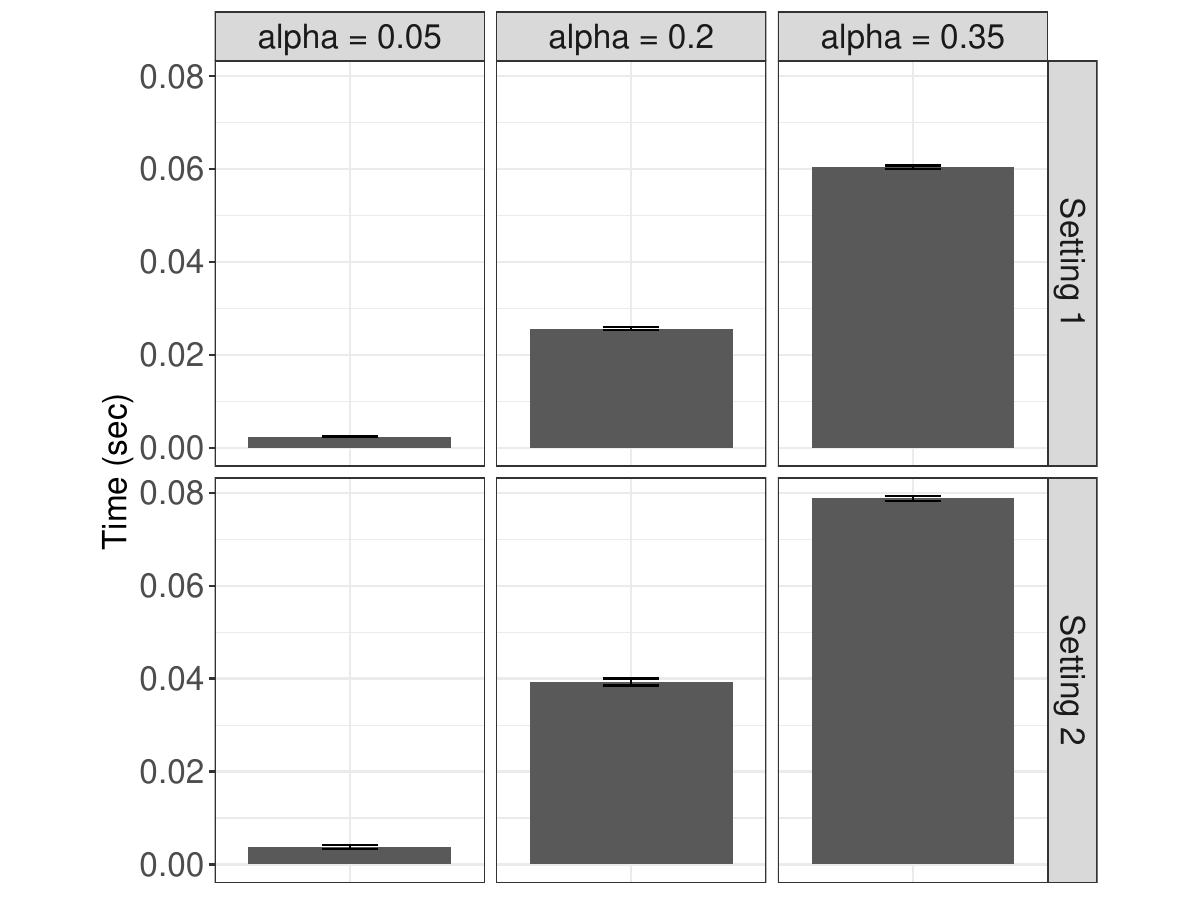}
    \caption{Average time taken by DACS for various values of $\alpha$ and settings of $n$ and $m$ on job hiring dataset.}
    \label{fig:hiring-time}
\end{figure}

Figure~\ref{fig:hiring-time} shows the average time taken by DACS using the underrepresentation index for various values of $\alpha$ and the two simulation settings considered in Section~\ref{expers:hiring}. In all cases, our method takes no longer than $0.08$ seconds, on average. These timing calculations were performed on a 2022 MacBook Air (M2 chip) with 8GB of RAM.

\subsubsection{Further details about the data}
In this section we simply report the null proportions and conditional null proportions given the diversification variables. In particular, the average proportion of candidates in the dataset that were not hired is $0.31$. The average proportion of men not hired is $0.28$ and the average proportion of women not hired is $0.37$.

\subsection{Drug discovery dataset}\label{app:exper-drug-discover}

\subsubsection{FDR, power, and number of selections}
\begin{figure}
    \centering
    \includegraphics[scale=0.4]{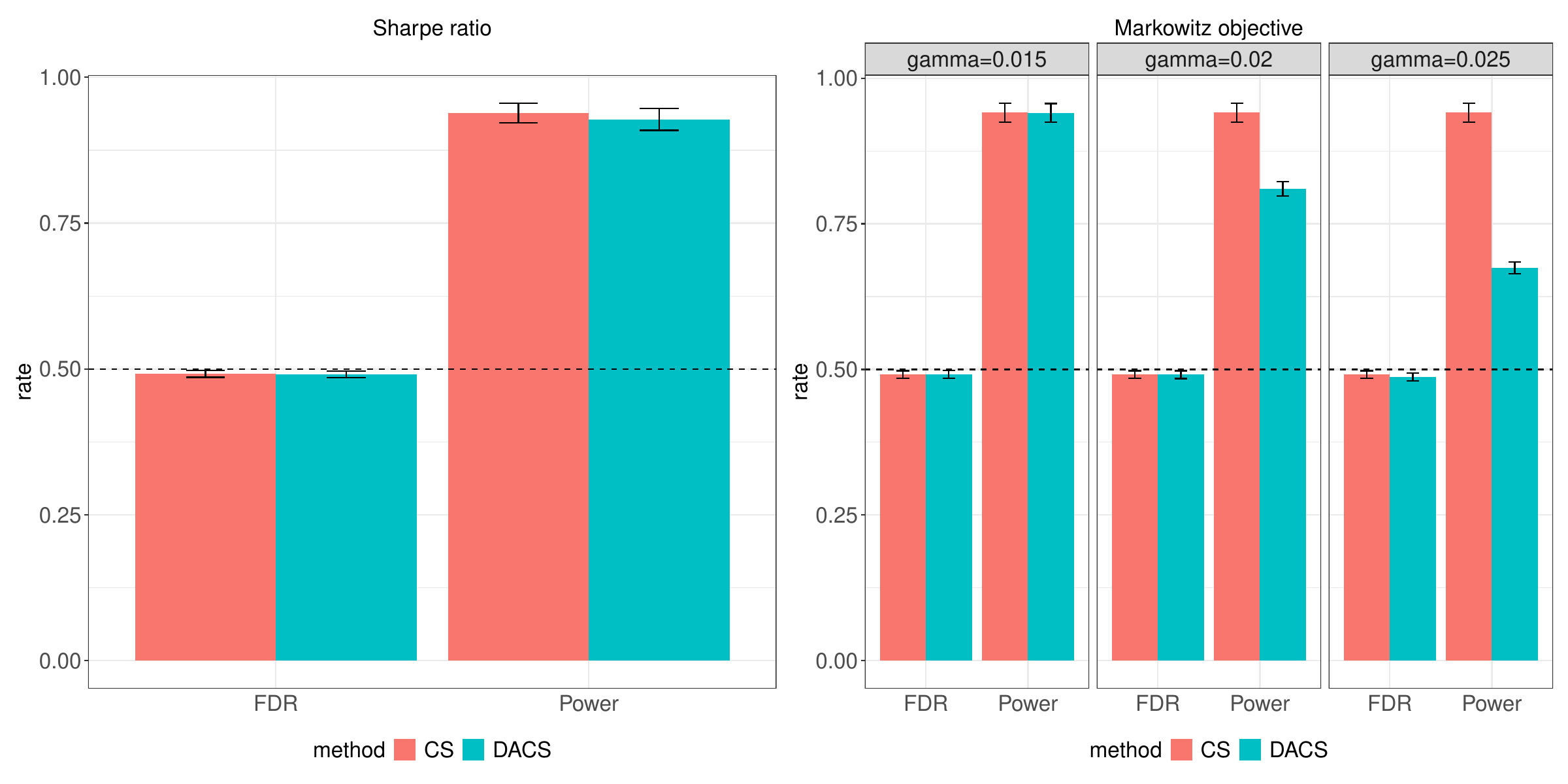}
    \caption{FDR and power under Sharpe ratio and Markowitz objective for various values of $\gamma$ (note that CS does not depend on $\gamma$, hence why its results for different $\gamma$ values are the same) for DACS (blue) and CS (red). Dashed horizontal line denotes nominal level $\alpha$.}
    \label{fig:drugexper-fdr-power}
\end{figure}

\begin{figure}
    \centering
    \includegraphics[scale=0.4]{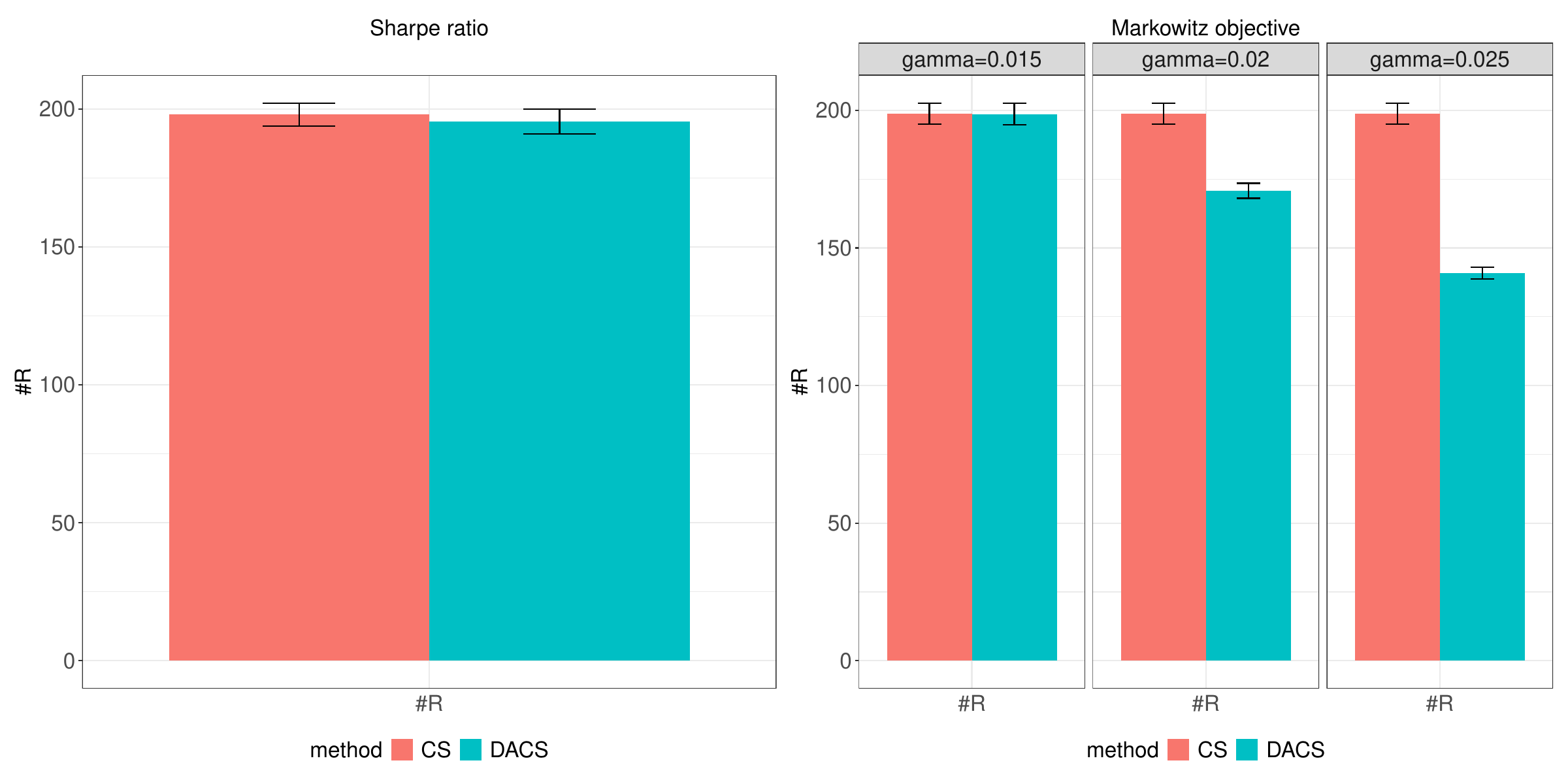}
    \caption{Number of selections under Sharpe ratio and Markowitz objective for various values of $\gamma$ (note that CS does not depend on $\gamma$, hence why its results for different $\gamma$ values are the same) for DACS (blue) and CS (red).}
    \label{fig:drugexper-numr}
\end{figure}

Figures~\ref{fig:drugexper-fdr-power} and \ref{fig:drugexper-numr} show the FDR, power, and number of selections made by DACS and CS for the Sharpe ratio and Markowitz objective for various values of $\gamma$. Under the Sharpe ratio, DACS and CS perform extremely similarly, obtaining nearly the same FDR, power, and number of selections. For the Markowitz objective, the results mirror those in Appendix~\ref{app:add-sim-results-sharpe-markowitz-fdr-power}: as $\gamma$ decreases, the power and number of selections also increase because the objective favors making more selections. Finally, just as in Appendix~\ref{app:add-sim-results-sharpe-markowitz-fdr-power}, the FDR is controlled below the nominal level $\alpha$ despite the fact that Proposition~\ref{relaxed-fdr} guarantees its control only below $1.3\alpha$.

\subsubsection{Average time}

\begin{figure}
    \centering
    \includegraphics[scale=0.4]{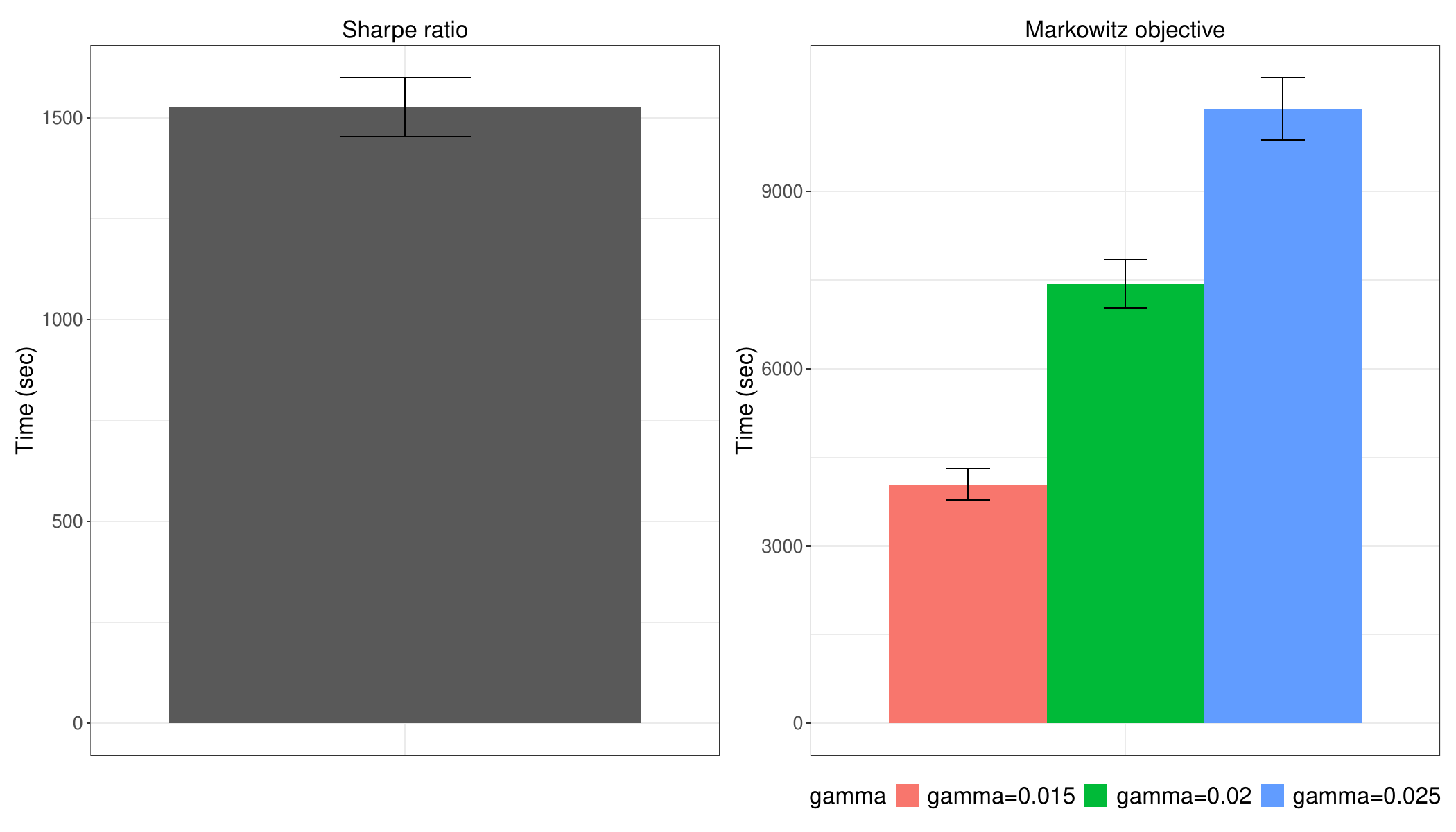}
    \caption{Average time under Sharpe ratio and Markowitz objective for various values of $\gamma$ for DACS.}
    \label{fig:drugexper-time}
\end{figure}

Figure~\ref{fig:drugexper-time} shows the average time taken by our method for the Sharpe ratio and the Markowitz objective for various values of $\gamma$ when run on Stanford's Sherlock cluster using a single core and $10$GB of memory. For the Sharpe ratio it takes, on average, just over 25 minutes. For the Markowitz objective, the computation time takes longer as $\gamma$ increases, taking just under $2.9$ hours for $\gamma=0.025$. We note that reward computation---which constitutes the bulk of the overall runtime---is embarrassingly parallel over the $500$ Monte Carlo samples used for approximation. Consequently, substantial speedups can be achieved by distributing this computation across multiple compute nodes.

\section{Further simulation details}\label{app:add-sim-detail}

\subsection{Underrepresentation index}\label{app:add-sim-detail-underrep}
In this section, we provide detailed descriptions of the $6$ simulation environments considered in Section~\ref{sims:underrep}. In all settings, $X_i$ are generated via a hierarchical Gaussian mixture model with $Z_i$ equaling the cluster index of the higher level. More specifically, in each of these settings, there are $C$ total components and $C_{\sub}$ sub-clusters per component, where both of these numbers are setting-dependent. The (observed) categorical diversification variable $Z_i$ is sampled according to a $\text{Multinomial}\left(1, \left(\pi_1, \ldots, \pi_C\right)\right)$ distribution and the (unobserved) sub-cluster index $Z^{\sub}_{i}$ is sampled, conditional on $Z_i$, according to a $\text{Multinomial}\left(1, \left(\pi^{(Z_i)}_{1}, \ldots, \pi^{(Z_i)}_{C_{\sub}}\right)\right)$ distribution. Finally, $X_i$ is sampled, conditional on $Z^{\sub}_{i}$, from the Normal distribution $\mathcal{N}(\mu_{Z^{\sub}_i}, 1)$ and $Y_i = \mu(X_i) + \epsilon_i$ where $\epsilon_i \overset{iid}{\sim} \mathcal{N}(0, 1)$. The probabilities $\left(\pi_1, \ldots, \pi_C\right),$ $\left(\pi^{(c)}_{1}, \ldots, \pi^{(c)}_{C_{\sub}}\right), c = 1, \ldots, C$ as well as the means $\mu_j$ and true regression function $\mu(\cdot)$ are setting-dependent and we provide their values (along with the values of $C$ and $C_{\sub}$) for each simulation setting in Table~\ref{fig:underrep-table}.

\begin{table}[h]
\centering
\begin{tabularx}{\textwidth}{lX}
\toprule
\textbf{Setting} & \textbf{Parameters} \\
\midrule
Setting 1 & 
$C=3$, $C_{\sub}=3$, 
$\boldsymbol{\pi} = (1/3, 1/3, 1/3)$\newline
$\boldsymbol{\pi}^{(1)} = (0.8, 0.05, 0.15)$, 
$\boldsymbol{\pi}^{(2)} = (0.2, 0.6, 0.2)$, 
$\boldsymbol{\pi}^{(3)} = (0.2, 0.2, 0.6)$\newline
$\boldsymbol{\mu} = (-0.5, 1.5, 2)$, 
$\mu(x) = x$ \\
\midrule
Setting 2 & 
$C=2$, $C_{\sub}=3$, 
$\boldsymbol{\pi} = (0.5, 0.5)$\newline
$\boldsymbol{\pi}^{(1)} = (0.8, 0.05, 0.15)$, 
$\boldsymbol{\pi}^{(2)} = (0.15, 0.75, 0.1)$\newline
$\boldsymbol{\mu} = (0, -2, 1.5)$, 
$\mu(x) = x^2 - 1$ \\
\midrule
Setting 3 & 
$C=2$, $C_{\sub}=3$, 
$\boldsymbol{\pi} = (0.5, 0.5)$\newline
$\boldsymbol{\pi}^{(1)} = (0.05, 0.85, 0.1)$, 
$\boldsymbol{\pi}^{(2)} = (0.4, 0.2, 0.4)$\newline
$\boldsymbol{\mu} = (0, -\pi, 0.7)$, 
$\mu(x) = 2\cos(x)$ \\
\midrule
Setting 4 & 
$C=4$, $C_{\sub}=5$, 
$\boldsymbol{\pi} = (0.25, 0.25, 0.25, 0.25)$\newline
$\boldsymbol{\pi}^{(1)} = (0.2, 0, 0, 0, 0.8)$\newline
$\boldsymbol{\pi}^{(2)} = (0, 0.22, 0.35, 0.43, 0)$\newline
$\boldsymbol{\pi}^{(3)} = (0.15, 0.35, 0.15, 0.1, 0.25)$\newline
$\boldsymbol{\pi}^{(4)} = (0.2, 0.05, 0.05, 0.05, 0.65)$\newline
$\boldsymbol{\mu} = (-2, -1, 0, 1.5, 3)$, 
$\mu(x) = 3\mathds{1}\{x > 0\} - 1.5$ \\
\midrule
Setting 5 & 
$C=2$, $C_{\sub}=3$, 
$\boldsymbol{\pi} = (0.7, 0.3)$\newline
$\boldsymbol{\pi}^{(1)} = (0.8, 0.05, 0.15)$, 
$\boldsymbol{\pi}^{(2)} = (0.15, 0.75, 0.1)$\newline
$\boldsymbol{\mu} = (0, -2, 1.5)$, 
$\mu(x) = x^2 - 1$ \\
\midrule
Setting 6 & 
$C=3$, $C_{\sub}=3$, 
$\boldsymbol{\pi} = (1/3, 1/3, 1/3)$\newline
$\boldsymbol{\pi}^{(1)} = (1, 0, 0)$, 
$\boldsymbol{\pi}^{(2)} = (0.2, 0.2, 0.6)$, 
$\boldsymbol{\pi}^{(3)} = (0.2, 0.6, 0.2)$\newline
$\boldsymbol{\mu} = (-0.75, 0.5, 1.2)$, 
$\mu(x) = x^3 + x$ \\
\bottomrule
\end{tabularx}
\caption{Setting-dependent parameters used in underrepresentation index simulations.}
\label{fig:underrep-table}
\end{table}

\subsection{Sharpe ratio and Markowitz objective}\label{app:add-sim-detail-sharpe-markowitz}
In this section, we describe further simulation details for Section~\ref{sims:sharpe-markowitz} regarding the Sharpe ratio and Markowitz objective. We consider four settings, which we describe below.

\paragraph{Setting 1} In the first setting, the data are generated from a mixture over Gaussian distributions with means \[\begin{bmatrix}
    1\\
    -1\\
    1
\end{bmatrix}, \begin{bmatrix}
    3/4\\
    4\\
    2
\end{bmatrix}, \begin{bmatrix}
    -2\\
    -3/2\\
    1
\end{bmatrix}, \begin{bmatrix}
    3/2\\
    2\\
    3/2
\end{bmatrix}, \begin{bmatrix}
    -5\\
    3\\
    2
\end{bmatrix}\] and common covariance $I_{3 \times 3}/4$. The response is distributed as $Y \mid X \sim \mathcal{N}(X_1X_2 + X_3, 1)$.

\paragraph{Setting 2} In this setting, we sample $X \sim \mathcal{N}(0, I_{5 \times 5})$ and $Y \mid X \sim \mathcal{N}\left(\|X\|_2^2 - \frac{7}{2}, 1\right)$.

\paragraph{Setting 3} Here, we sample $X \sim \mathcal{N}(0, I_{5 \times 5})$ and take $Y \mid X \sim \mathcal{N}\left(2\cos(X_1), 1\right)$.

\paragraph{Setting 4} In this setting, we take $X \sim \mathcal{N}(0, I_{5 \times 5})$ and $Y \mid X \sim \mathcal{N}\left(\frac{7}{2}-\|X\|_2^2, 1\right)$.

As mentioned in Section~\ref{lp-relax}, the relaxed optimal values in equation~\eqref{relaxed-go-def} involve expectations of the diversity over the Bernoulli random vector $\bchi$. For the Markowitz objective, these are easily computed because, for any positive definite matrix $\Lambda \in \mathbb{R}^{p \times p}$, one simply has
\begin{align*}
    \bE_{\xi_i \overset{\text{ind}}{\sim} \text{Bern}(\chi_i), i \in [p]}\left[\bxi^\top \mathbf{1}- \frac{\gamma}{2}\bxi^\top \Lambda \bxi\right] = \bchi^\top \mathbf{1} - \frac{\gamma}{2}\left(\bchi^\top \Lambda \bchi + \text{diag}\left(\Lambda\right)^\top \bchi - \text{diag}\left(\Lambda\right)^\top [\bchi \odot \bchi]\right)
\end{align*}
for any $\bchi \in [0,1]^p$, where $\text{diag}\left(\Lambda\right)$ denotes the vector comprising the diagonal entries of $\Lambda$ and $\bchi \odot \bchi$ denotes the elementwise product of $\bchi$ with itself. On the other hand, as discussed in Section~\ref{lp-relax}, we approximate this expectation using MC for the Sharpe ratio; in particular, in our simulations we use $50$ MC samples.

\section{Additional simulation results}\label{app:app-sim-results}
This section provides additional simulation results for Section~\ref{sims}. All simulations in this section were performed on Stanford's Sherlock cluster using a single core and $10$GB of memory.
\subsection{Underrepresentation index}\label{app:add-sim-results-underrep}
In this section, we provide further simulation results for the underrepresentation index simulations presented in Section~\ref{sims:underrep} of the main text.

\subsubsection{Average proportions and number of selections for OLS and SVM}\label{app:avg-cluster-underrep}
\begin{figure}
\centering
    \includegraphics[scale=0.6]{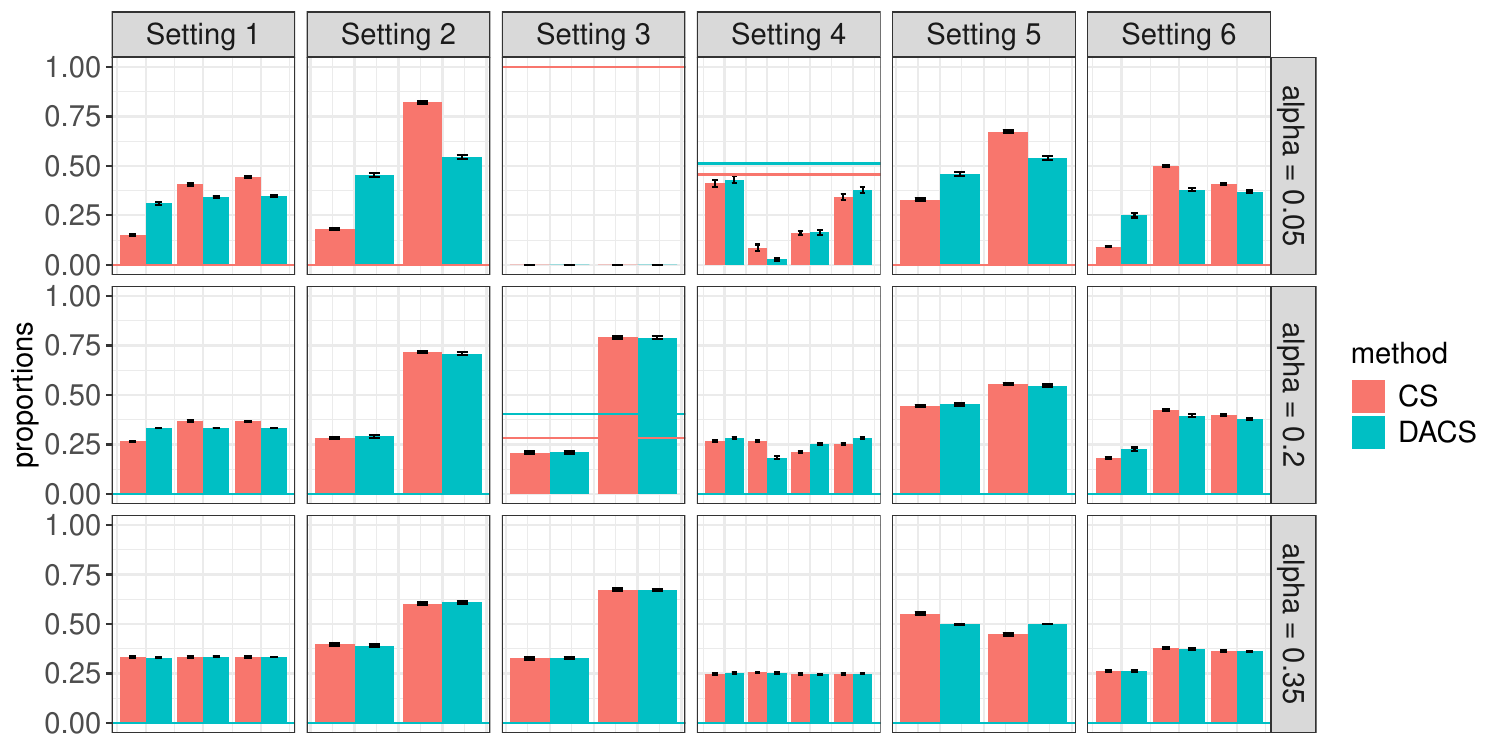}
    \caption{Average cluster proportions conditional on being non-empty and probability of being empty for DACS selection set compared to CS selection set for various simulation settings and nominal levels $\alpha$; $\hat{\mu}$ is a fitted OLS. Within any given cell, bars which are closer to uniform denote a more diverse selection set. Horizontal lines, when visible, denote the average proportion of sets which are empty. Both bars are at $1$ for Setting 3, $\alpha = 0.05$ because the nominal level is so small that both methods fail to make any selections in all $250$ simulation replications.}
    \label{fig:underrep-result-ols-cluster}
\end{figure}

\begin{figure}
\centering
    \includegraphics[scale=0.6]{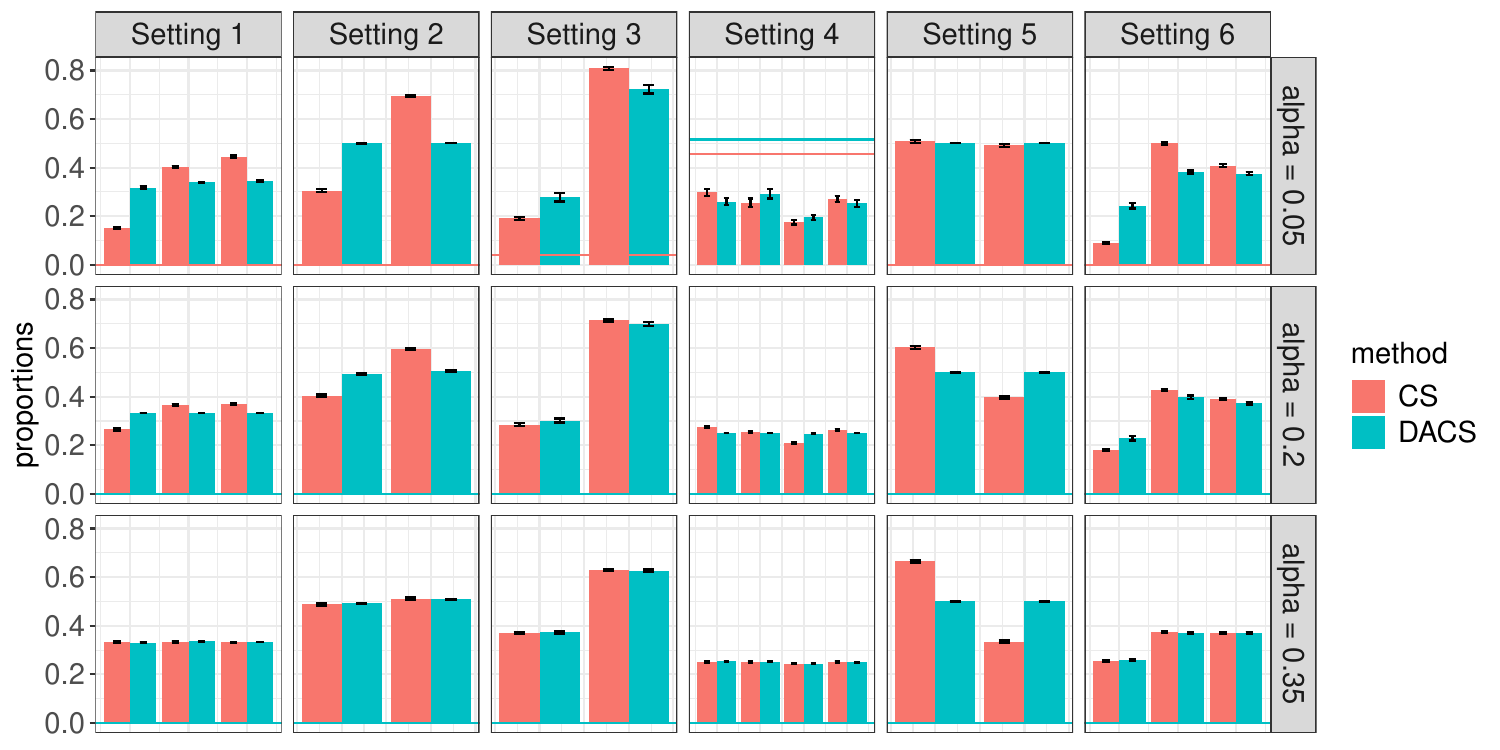}
    \caption{Average cluster proportions conditional on being non-empty and probability of being empty for DACS selection set compared to CS selection set for various simulation settings and nominal levels $\alpha$; $\hat{\mu}$ is a fitted SVM. Within any given cell, bars which are closer to uniform denote a more diverse selection set. Horizontal lines, when visible, denote the average proportion of sets which are empty.}
    \label{fig:underrep-result-svm-cluster}
\end{figure}

Figures~\ref{fig:underrep-result-ols-cluster} and \ref{fig:underrep-result-svm-cluster} give the average cluster proportions and average number of times that the selection set is empty for OLS- and SVM-fitted score functions for various simulation settings and nominal levels $\alpha$. In Setting 4, using the OLS-fitted $\hat{\mu}$, DACS does sometimes produce less diverse selection sets than CS (also, for Setting 3, $\alpha=0.2$, DACS does have lower underrepresentation index value than CS as shown in Figure~\ref{fig:diversity-result-ols-cluster}). This is most likely because, in this setting, the linear model is (very) misspecified (as per Table~\ref{fig:underrep-table}, the true regression function is a scaled and shifted sign function) which, in combination with the small FDR level constraint, forces the BH stopping time $\tau_{\bh}$ to be quite small, disallowing DACS much room for error and, as evidenced by the plot, causing it to err. Apart from this cell, these results essentially mirror those for the MLP-fitted score function given in Figure~\ref{fig:underrep-result1} in the main text. Comparing results for different nominal levels $\alpha$---in combination with Figures~\ref{fig:numr-result-ols-cluster} and \ref{fig:numr-result-svm-cluster}---we obtain the same conclusion as in Section~\ref{sims:underrep} of the main text: if one is willing to tolerate a larger nominal level for the sake of diversity, then DACS at a more liberal nominal level is able to construct more diverse selection sets of the same or greater size than CS at a more conservative nominal level, especially in settings 3, 4, and 6.

\begin{figure}
\centering
    \includegraphics[scale=0.6]{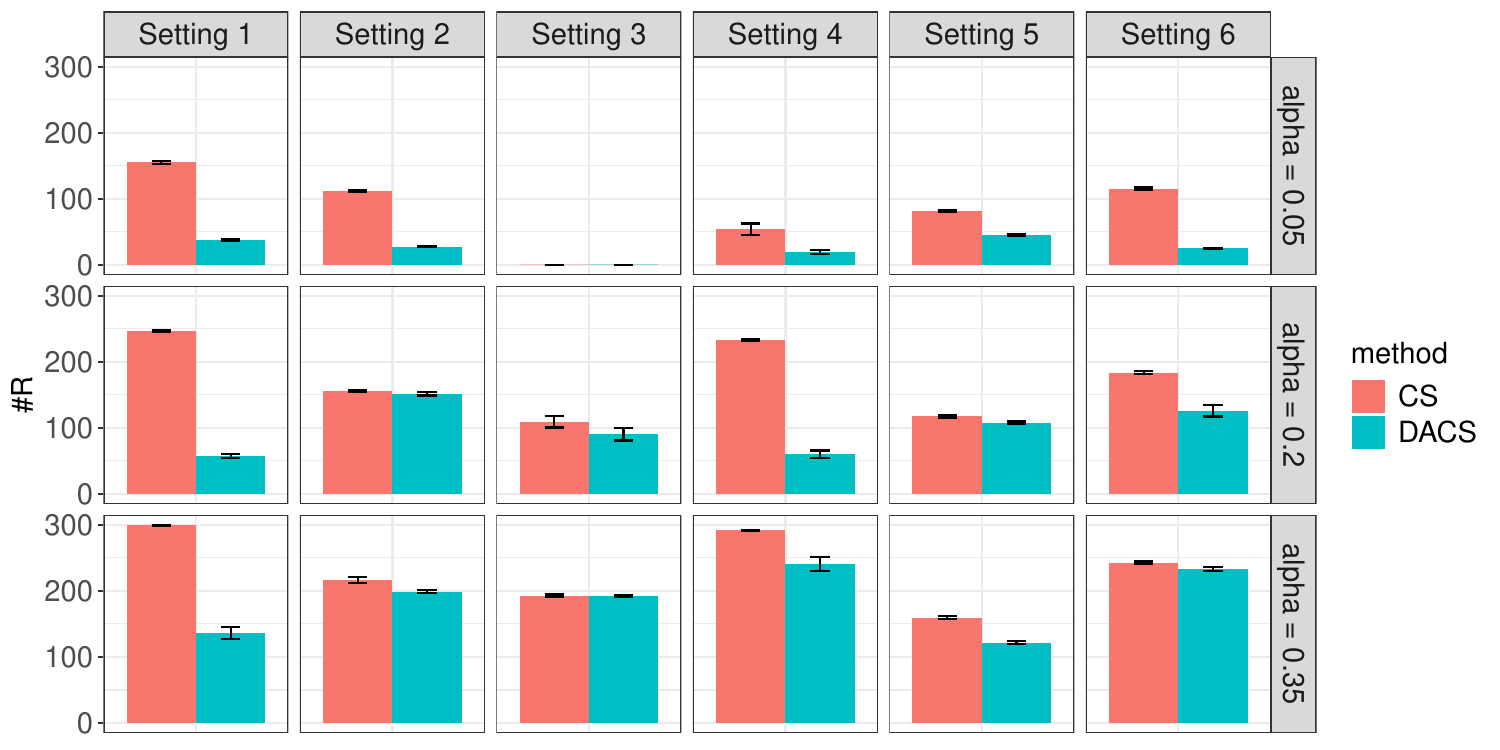}
    \caption{Average numbers of selections made by DACS selection set compared to CS selection set for various simulation settings and nominal levels $\alpha$; $\hat{\mu}$ is a fitted OLS.}
    \label{fig:numr-result-ols-cluster}
\end{figure}

\begin{figure}
\centering
    \includegraphics[scale=0.6]{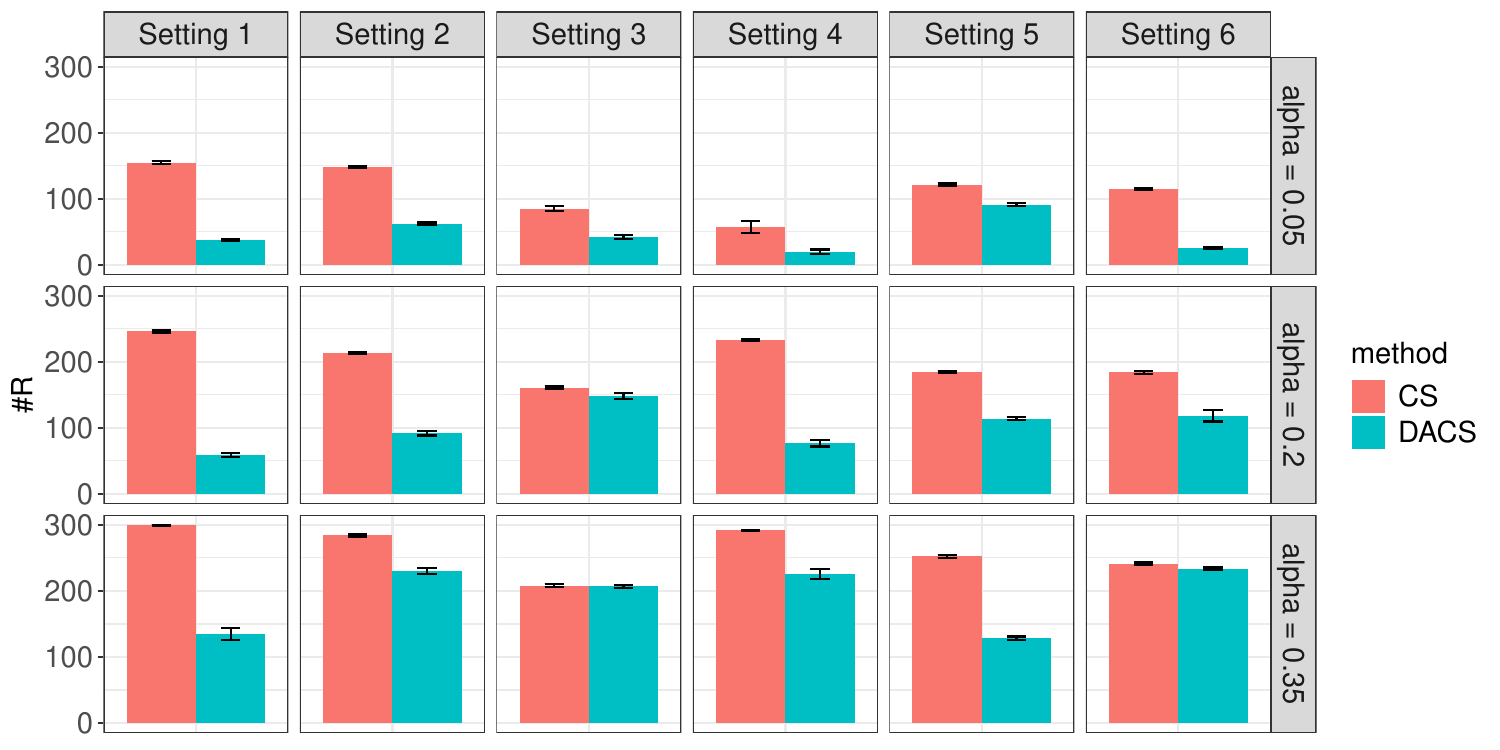}
    \caption{Average numbers of selections made by DACS selection set compared to CS selection set for various simulation settings and nominal levels $\alpha$; $\hat{\mu}$ is a fitted SVM.}
    \label{fig:numr-result-svm-cluster}
\end{figure}

\subsubsection{Average diversities for OLS, MLP, and SVM}\label{app:avg-div-underrep}
\begin{figure}
\centering
    \includegraphics[scale=0.6]{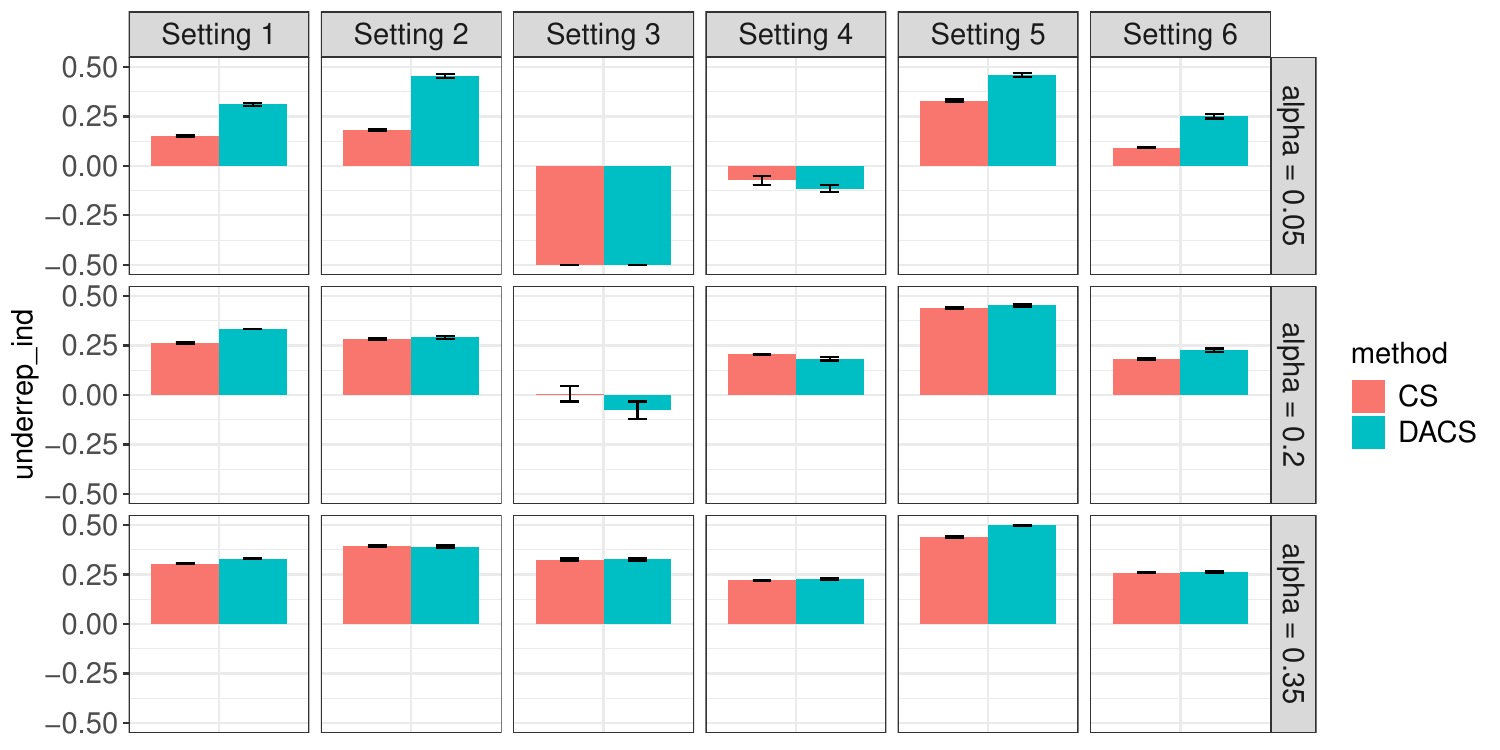}
    \caption{Average underrepresentation index value of DACS selection set compared to CS selection set for various simulation settings and nominal levels $\alpha$; $\hat{\mu}$ is a fitted OLS. }
    \label{fig:diversity-result-ols-cluster}
\end{figure}

\begin{figure}
\centering
    \includegraphics[scale=0.6]{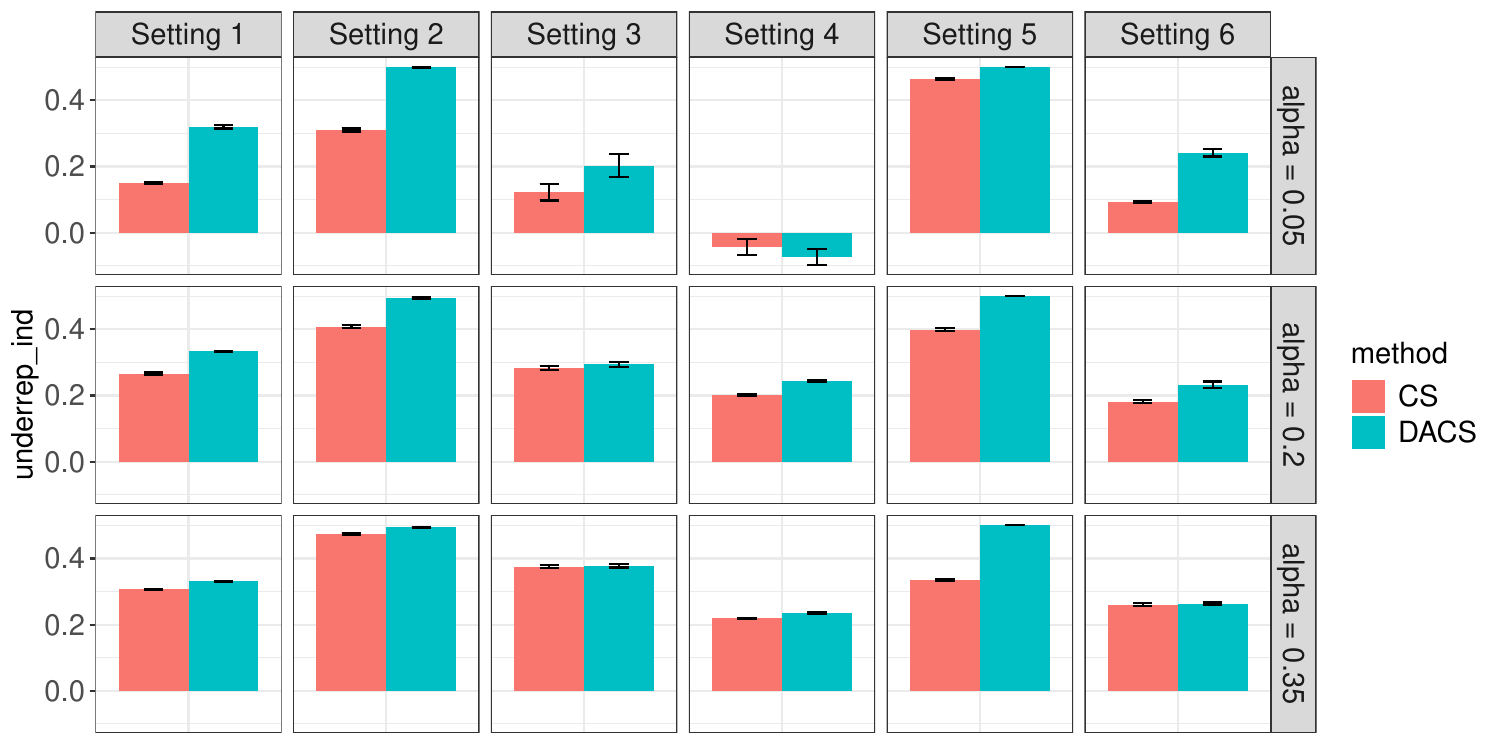}
    \caption{Average underrepresentation index value of DACS selection set compared to CS selection set for various simulation settings and nominal levels $\alpha$; $\hat{\mu}$ is a fitted MLP.}
    \label{fig:diversity-result-mlp-cluster}
\end{figure}

\begin{figure}
\centering
    \includegraphics[scale=0.6]{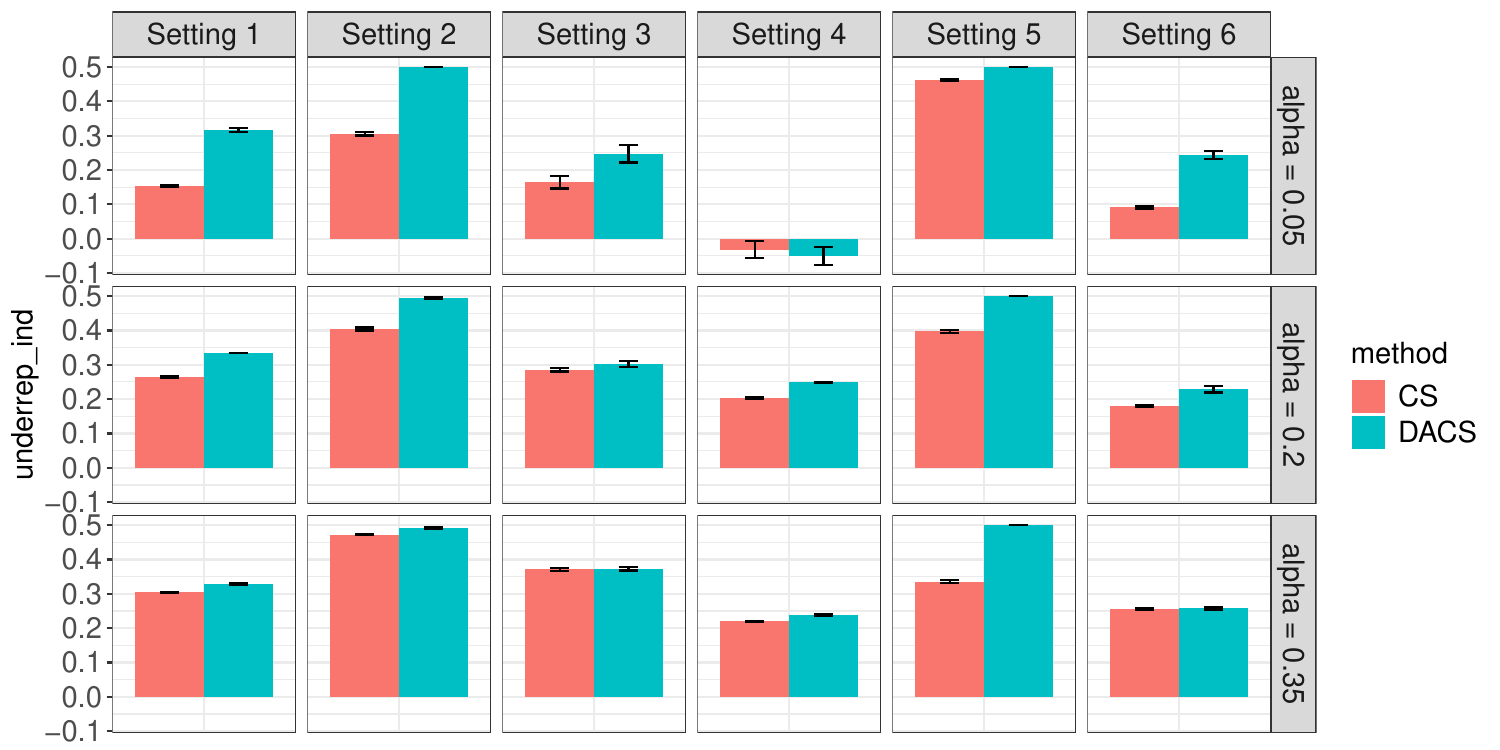}
    \caption{Average underrepresentation index value of DACS selection set compared to CS selection set for various simulation settings and nominal levels $\alpha$; $\hat{\mu}$ is a fitted SVM.}
    \label{fig:diversity-result-svm-cluster}
\end{figure}

Figures~\ref{fig:diversity-result-ols-cluster}, \ref{fig:diversity-result-mlp-cluster}, and \ref{fig:diversity-result-svm-cluster} plot the average underrepresentation index for the OLS-, MLP-, and SVM-fitted $\hat{\mu}$ models in the various simulation settings and nominal levels $\alpha$. In correspondence with the results of Figure~\ref{fig:underrep-result-ols-cluster}, the only cells in which DACS produces, on average, less diverse selection sets than CS at the same nominal level is when $\alpha =  0.05$ in Setting 4 for all three regressors and $\alpha=0.2$ in Settings 3 and 4 for OLS, and the same explanation as given in Section~\ref{sims:underrep} applies here as well: the small size of the CS selection set in these settings makes the diversification task more challenging, causing DACS to err.

\subsubsection{FDR and power for OLS, MLP, and SVM}\label{app:avg-mt-metric-underrep}
\begin{figure}
\centering
    \includegraphics[scale=0.6]{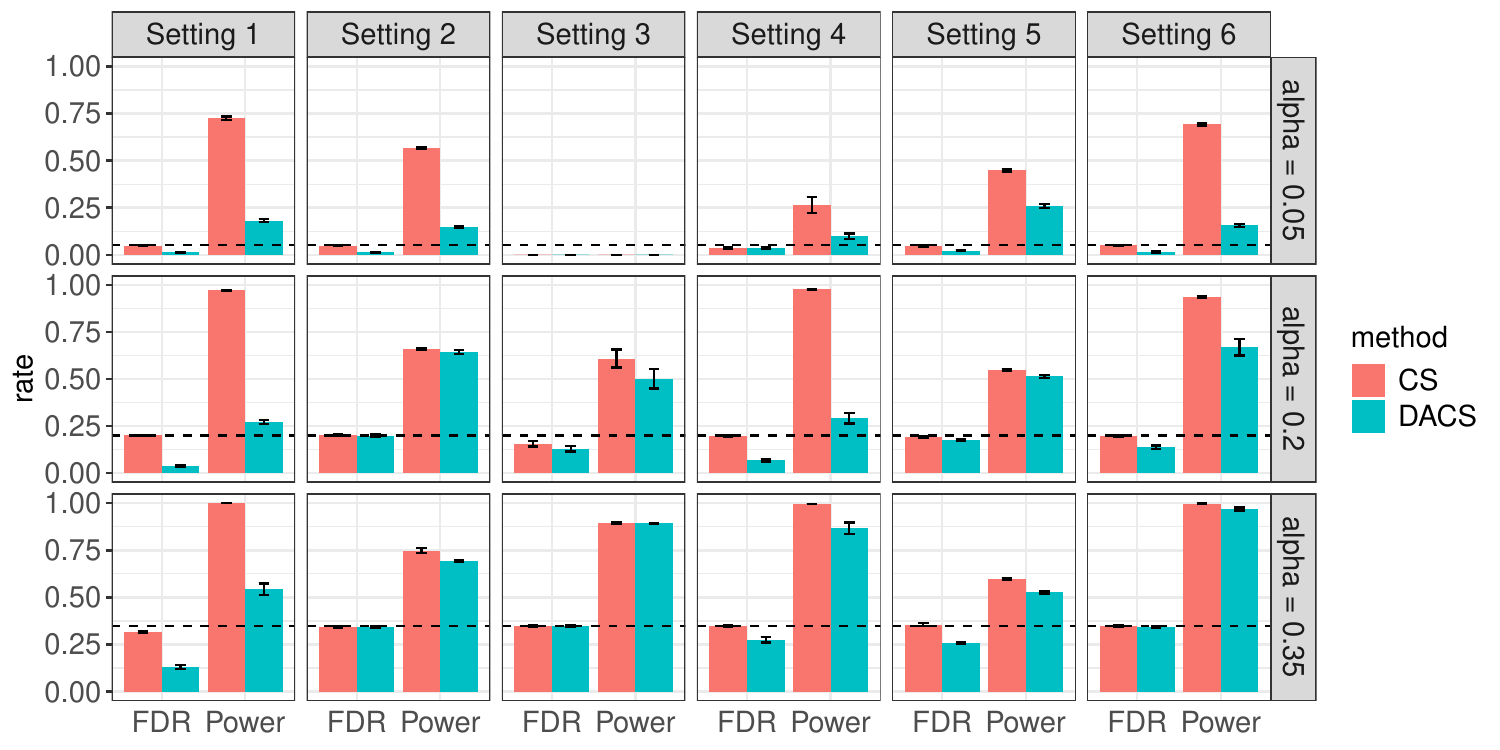}
    \caption{Average FDR and power of DACS and CS selection sets for various simulation settings and nominal levels $\alpha$ (dashed line denotes nominal level); $\hat{\mu}$ is a fitted OLS.}
    \label{fig:fdrpower-result-ols-cluster}
\end{figure}

\begin{figure}
\centering
    \includegraphics[scale=0.6]{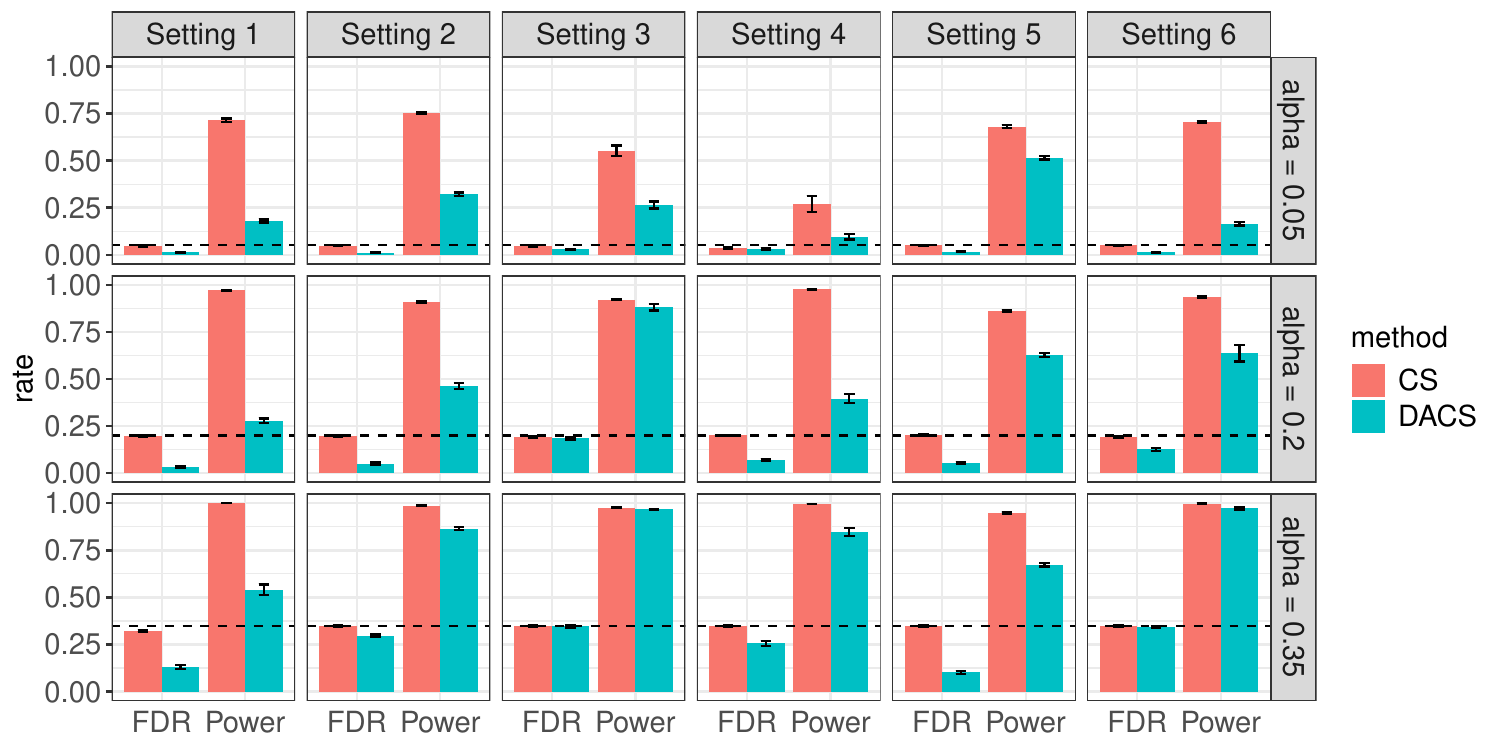}
    \caption{Average FDR and power of DACS and CS selection sets for various simulation settings and nominal levels $\alpha$ (dashed line denotes nominal level); $\hat{\mu}$ is a fitted OLS.}
    \label{fig:fdrpower-result-mlp-cluster}
\end{figure}

\begin{figure}
\centering
    \includegraphics[scale=0.6]{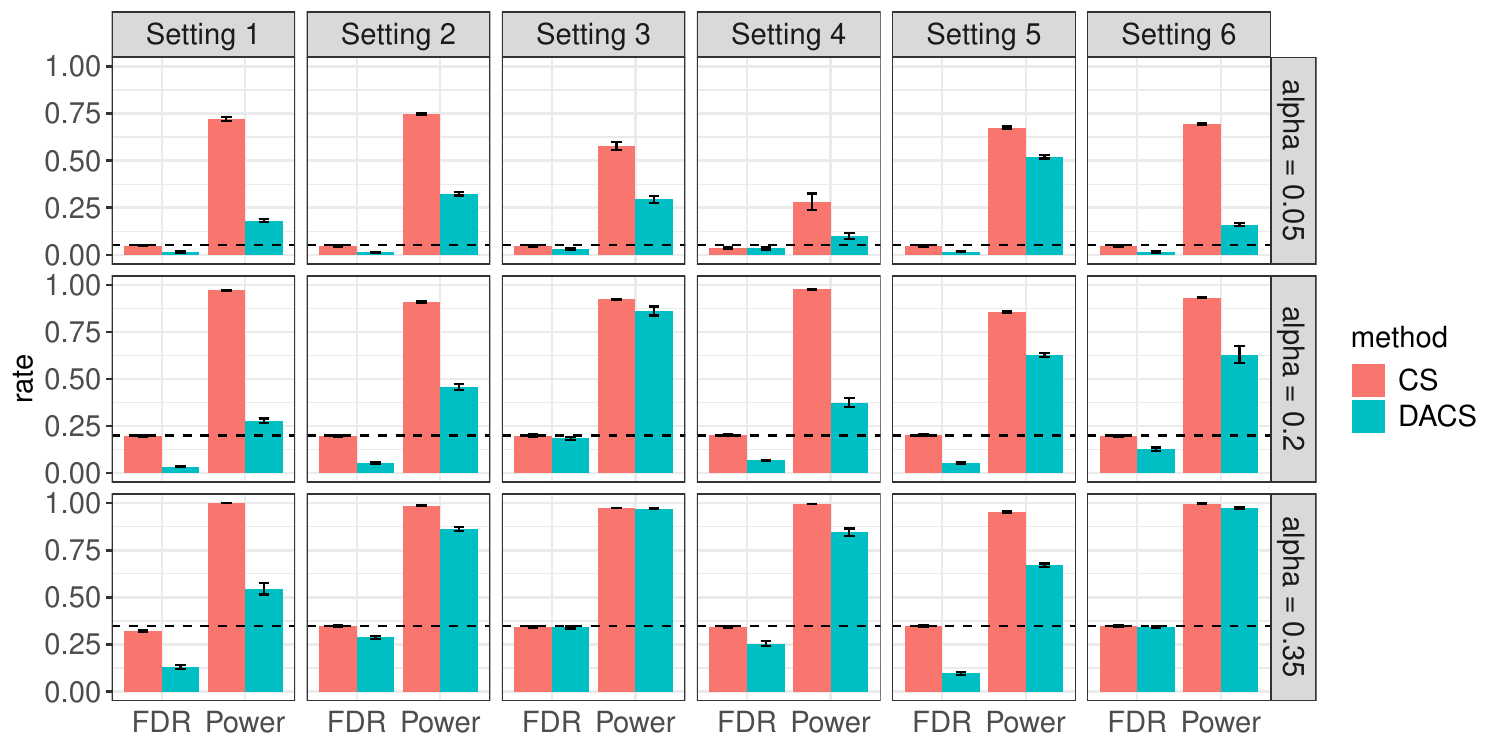}
    \caption{Average FDR and power of DACS and CS selection sets for various simulation settings and nominal levels $\alpha$ (dashed line denotes nominal level); $\hat{\mu}$ is a fitted OLS.}
    \label{fig:fdrpower-result-svm-cluster}
\end{figure}

Figures~\ref{fig:fdrpower-result-ols-cluster}, \ref{fig:fdrpower-result-mlp-cluster}, and \ref{fig:fdrpower-result-svm-cluster} plot the FDR and power of DACS and CS for various simulation settings and nominal levels. These figures demonstrate that DACS controls FDR at level $\alpha$---thereby validating the theoretical guarantee of Corollary~\ref{corollary:dacs-validity}. In view of the fact that DACS makes fewer selections than CS at the same nominal level, it has lower power. We remind the reader that power, however, is not our focus, and instead we are interested in optimizing for diversity.

\subsubsection{Average computation time for OLS, MLP, and SVM}
\begin{figure}
\centering
    \includegraphics[scale=0.6]{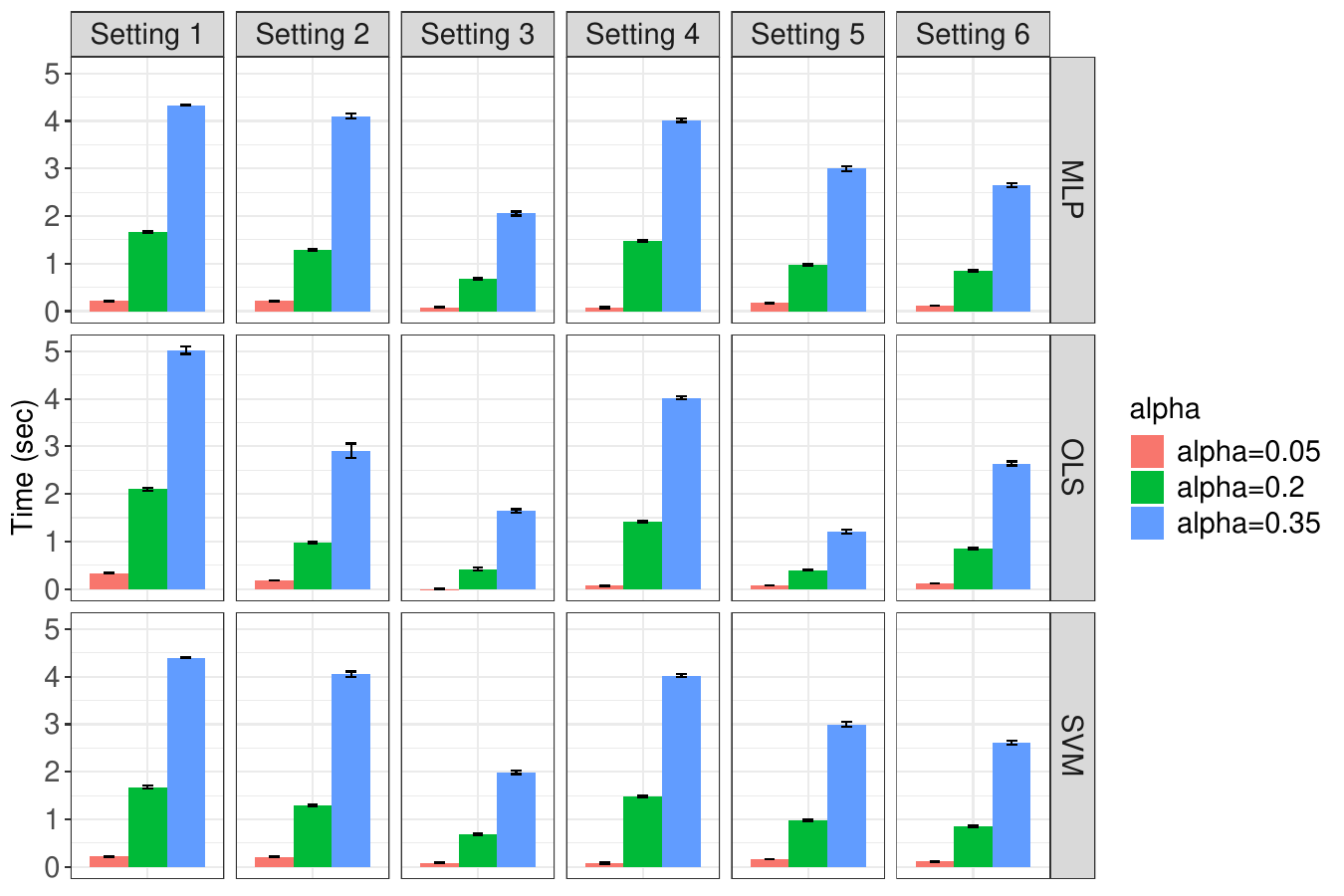}
    \caption{Average computation time of DACS for various simulation settings, $\hat{\mu}$'s, and nominal levels $\alpha$.}
    \label{fig:underrep-result5}
\end{figure}

Figure~\ref{fig:underrep-result5} shows the average computation required by DACS for the underrepresentation index. In all settings---even for $\alpha = 0.35$---DACS never takes longer than approximately $5$ seconds, on average.

\subsubsection{Simulation environment details}
\begin{figure}
\centering
    \includegraphics[scale=0.6]{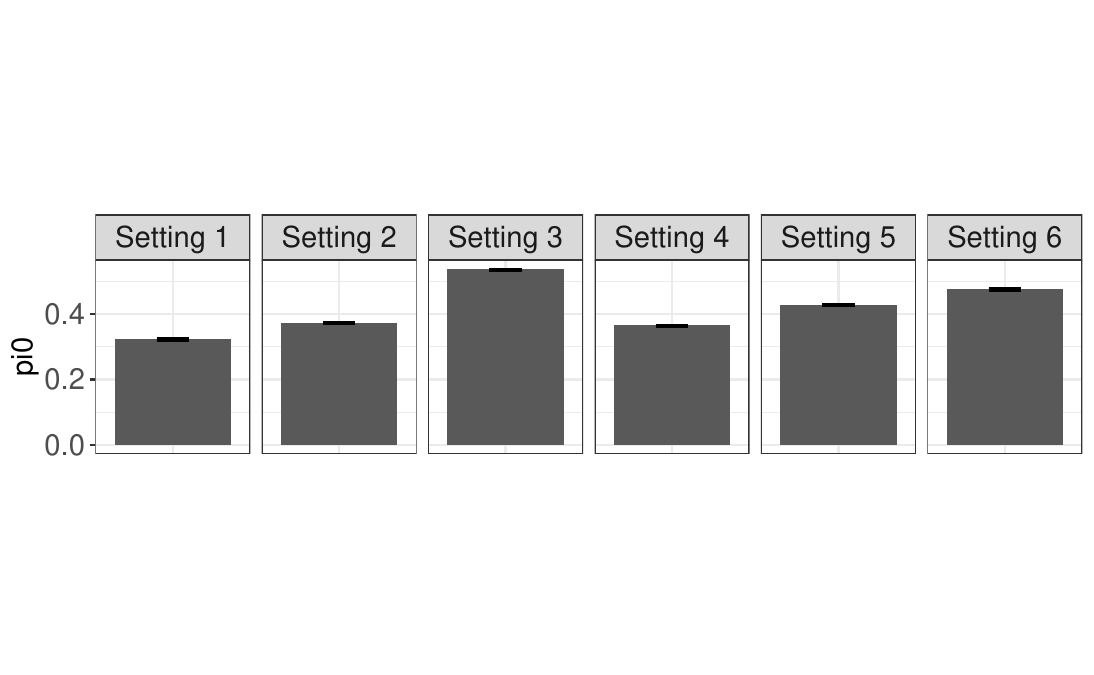}
    \caption{Values of $\pi_0$ for various simulation settings.}
    \label{fig:underrep-result6}
\end{figure}

\begin{figure}
\centering
    \includegraphics[scale=0.6]{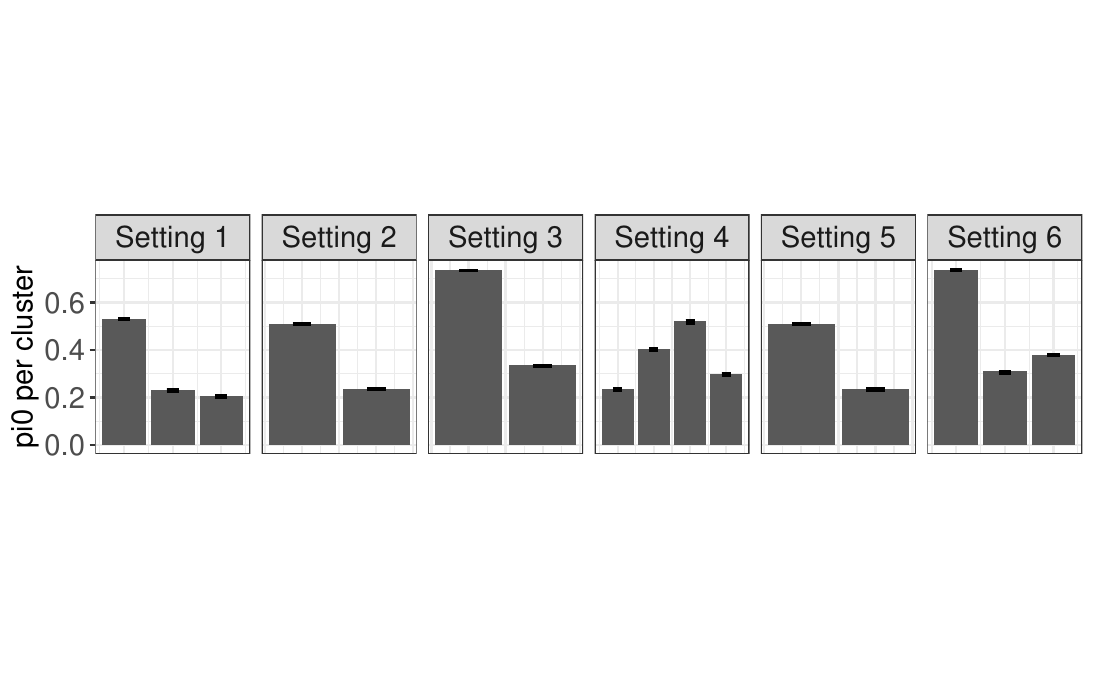}
    \caption{Values of cluster-conditional $\pi_0$'s for various simulation settings.}
    \label{fig:underrep-result7}
\end{figure}

Figure~\ref{fig:underrep-result6} plots the values of $\pi_0 = \bP(Y \leq 0)$ for the various simulation settings. Figure~\ref{fig:underrep-result7} plots cluster-conditional $\pi_0$'s. More specifically it reports the values $\bP(Y \leq 0 \mid Z = c)$ for each $c \in [C]$.

\subsubsection{Further simulation results for $m=100$}\label{app:further-underrep-100}
In this section we report additional simulation results for the same simulation settings discussed in Section~\ref{sims:underrep} except with $m=100$ test samples. Figures~\ref{fig:underrep-result100_1}--\ref{fig:underrep-result100_3} are the analogous plots to Figures~\ref{fig:underrep-result-ols-cluster}, \ref{fig:underrep-result1}, and \ref{fig:underrep-result-svm-cluster} respectively, showing the average cluster proportions made by DACS and CS. Figures~\ref{fig:underrep-result100_4}--\ref{fig:underrep-result100_6} are the analogous plots to Figures~\ref{fig:diversity-result-ols-cluster}, \ref{fig:diversity-result-mlp-cluster}, and \ref{fig:diversity-result-svm-cluster} which show the average underrepresentation index value for each method. Overall, DACS produces slightly more diverse selection sets than CS, but not significantly so. As discussed in Section~\ref{sims:underrep}, this is likely because with fewer test samples, the CS selection set's size is reduced, making it harder to prune the selection set without removing underrepresented clusters. Indeed comparing the number of selections made by each method with $m=100$ (Figures~\ref{fig:underrep-result100_7}--\ref{fig:underrep-result100_9}) to when $m=300$ (Figures~\ref{fig:numr-result-ols-cluster}, \ref{fig:underrep-result2}, \ref{fig:numr-result-svm-cluster}, respectively), we indeed find that CS makes significantly more selections when $m$ is larger, making the diversification problem easier.
\begin{figure}
\centering
    \includegraphics[scale=0.6]{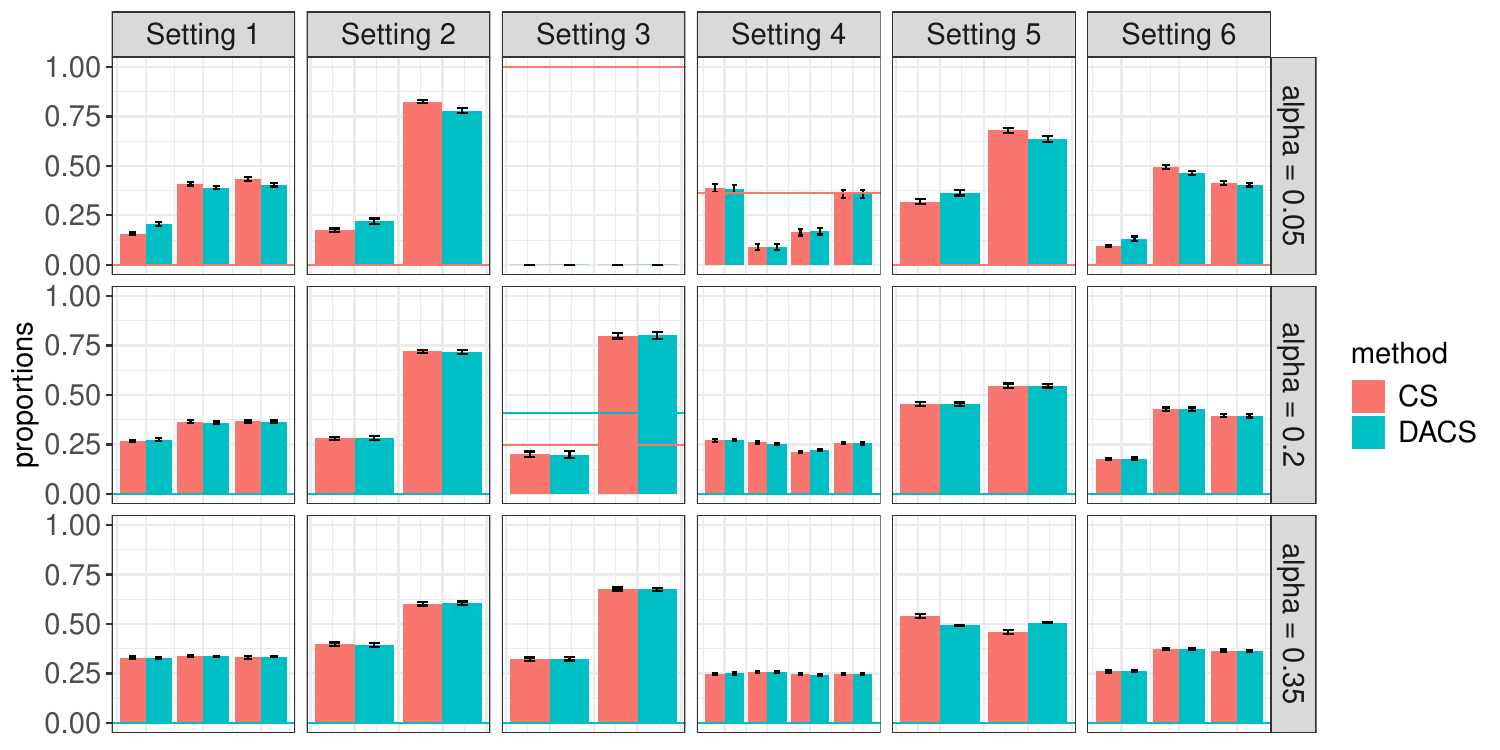}
    \caption{Analogous cluster proportion plot to Figure~\ref{fig:underrep-result-ols-cluster} with $m=100$ test samples.}
    \label{fig:underrep-result100_1}
\end{figure}

\begin{figure}
\centering
    \includegraphics[scale=0.6]{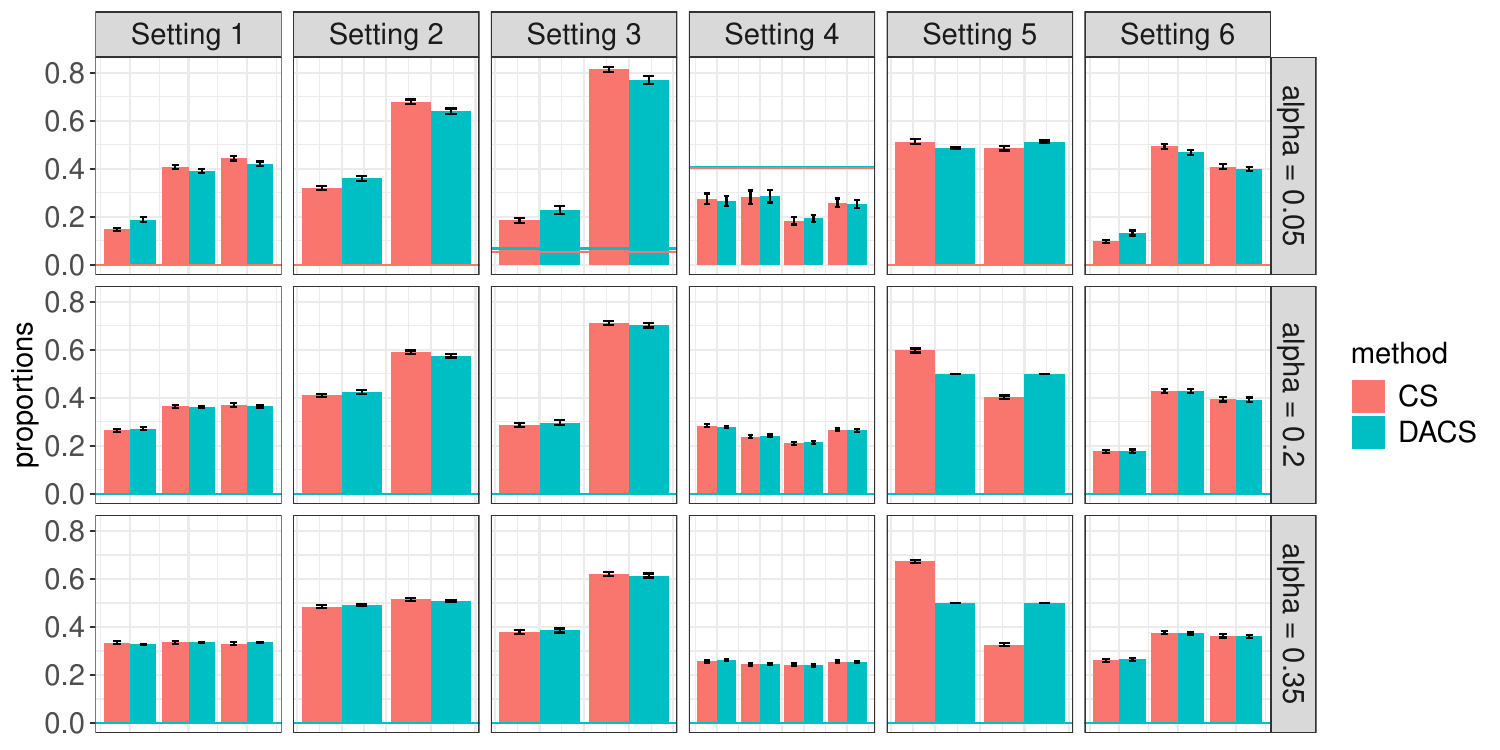}
    \caption{Analogous cluster proportion plot to Figure~\ref{fig:underrep-result1} with $m=100$ test samples.}
    \label{fig:underrep-result100_2}
\end{figure}

\begin{figure}
\centering
    \includegraphics[scale=0.6]{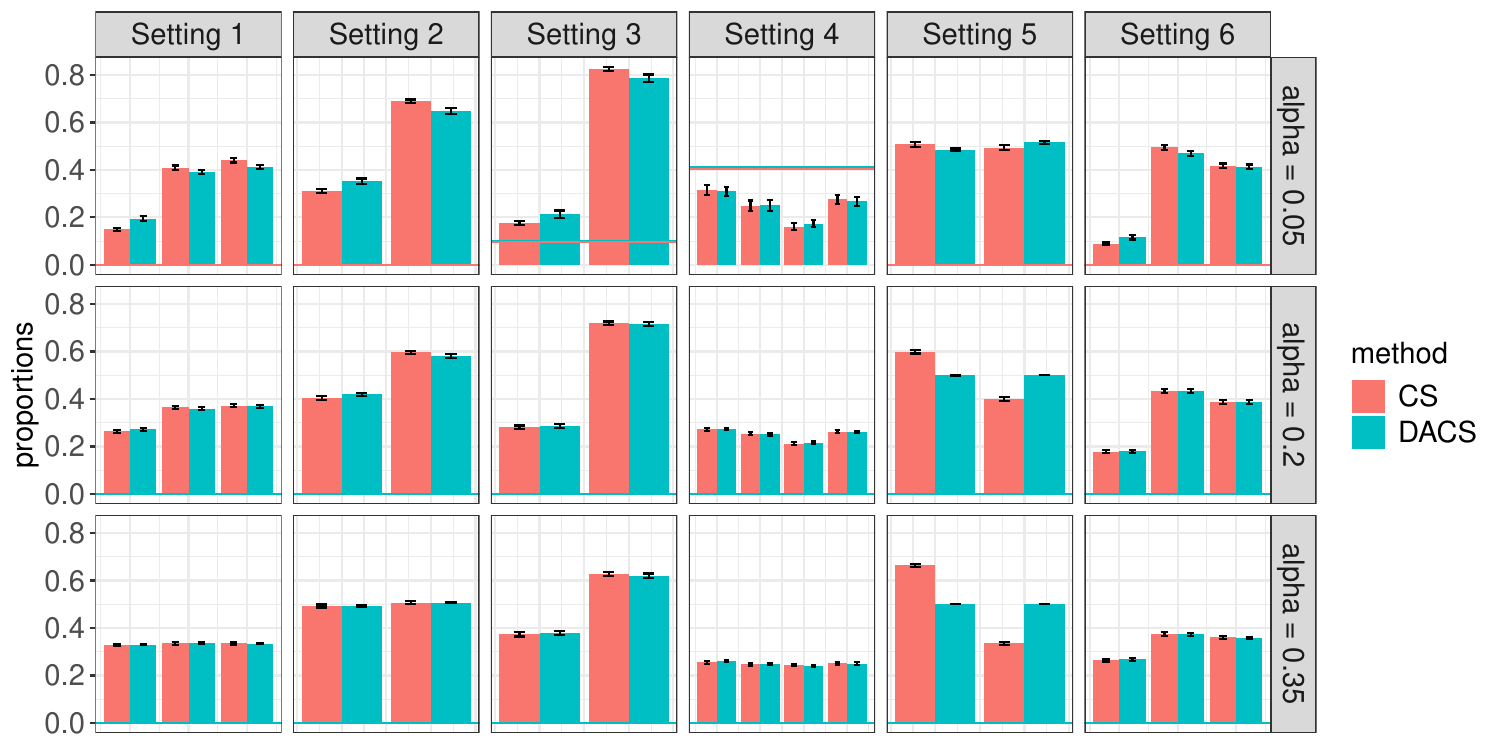}
    \caption{Analogous cluster proportion plot to Figure~\ref{fig:underrep-result-svm-cluster} with $m=100$ test samples.}
    \label{fig:underrep-result100_3}
\end{figure}

\begin{figure}
\centering
    \includegraphics[scale=0.6]{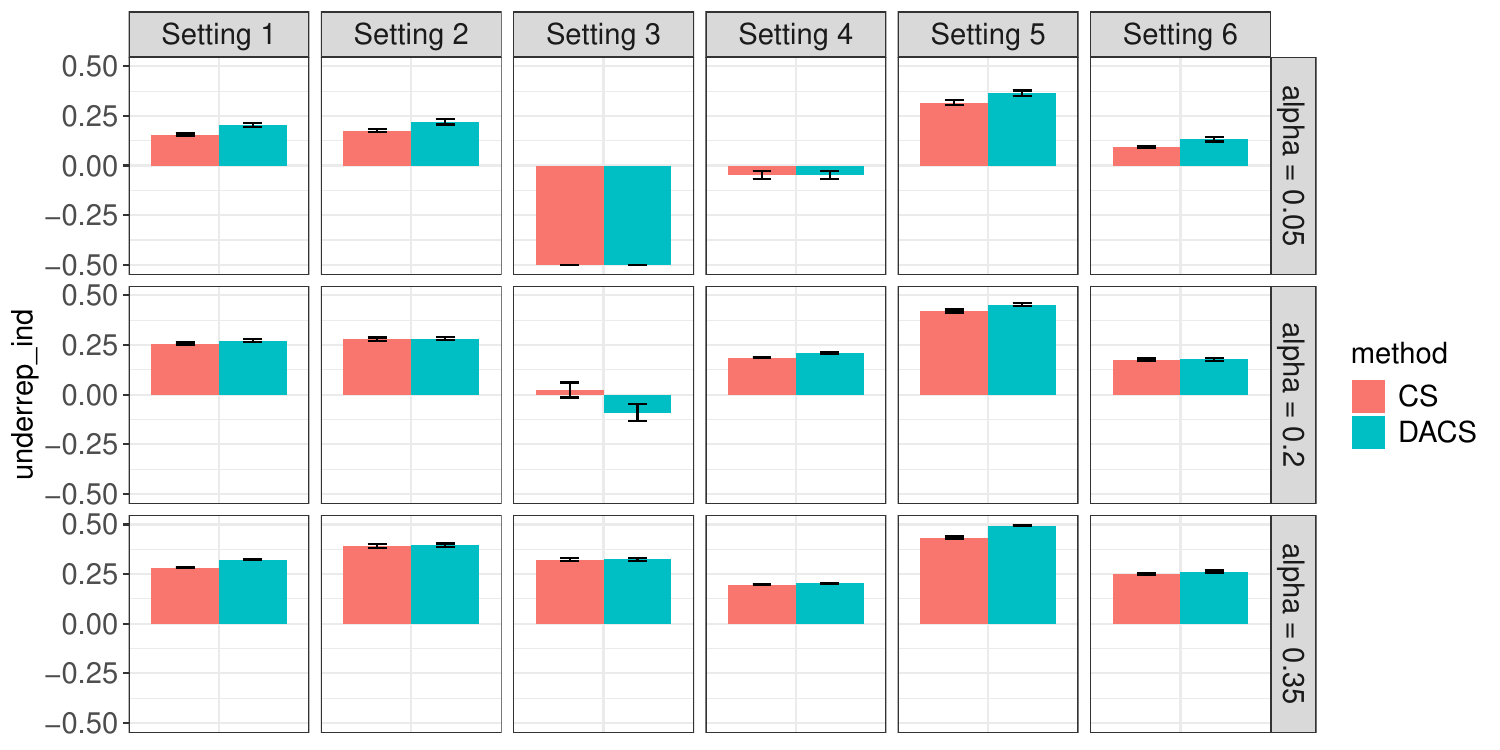}
    \caption{Analogous underrepresentation index diversity plot to Figure~\ref{fig:diversity-result-ols-cluster} with $m=100$ test samples.}
    \label{fig:underrep-result100_4}
\end{figure}

\begin{figure}
\centering
    \includegraphics[scale=0.6]{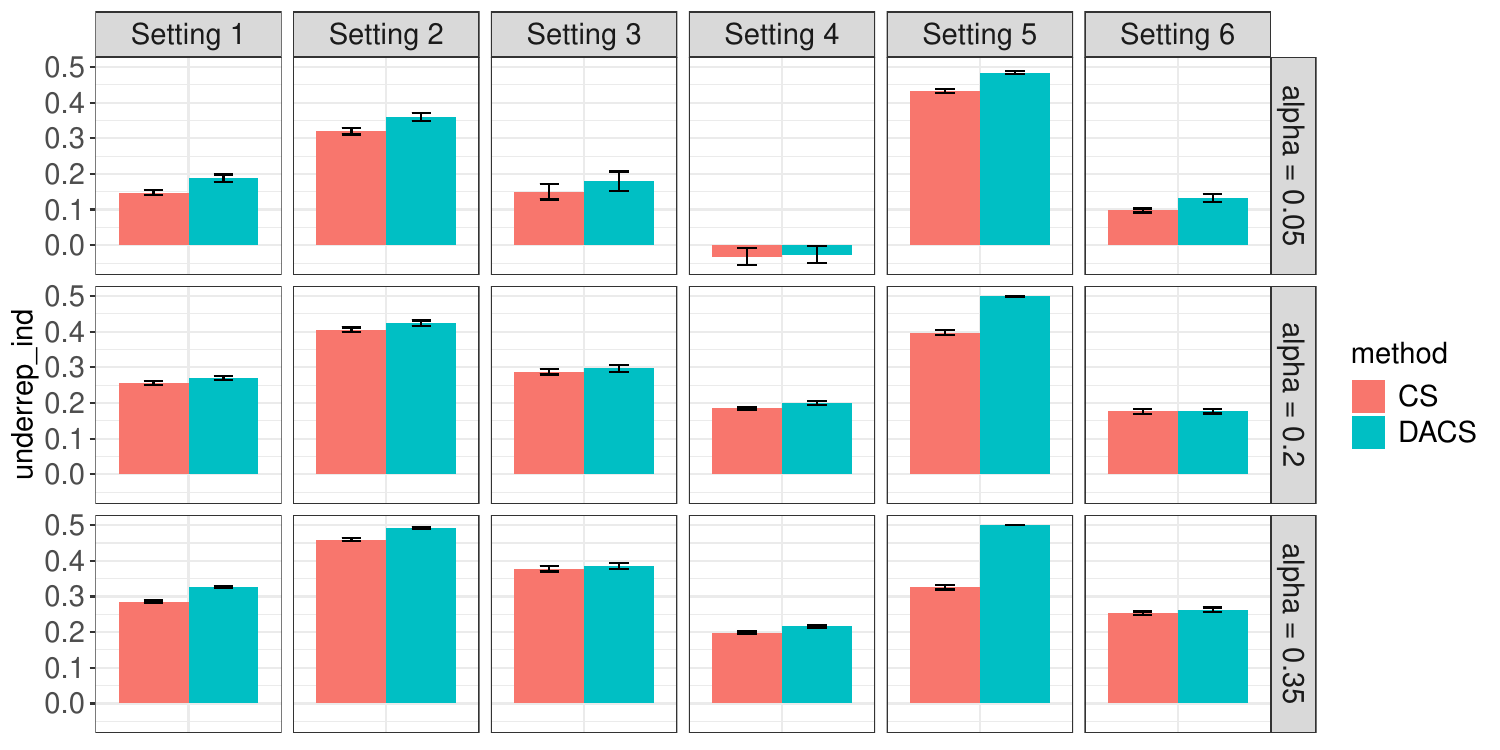}
    \caption{Analogous underrepresentation index diversity plot to Figure~\ref{fig:diversity-result-mlp-cluster} with $m=100$ test samples.}
    \label{fig:underrep-result100_5}
\end{figure}

\begin{figure}
\centering
    \includegraphics[scale=0.6]{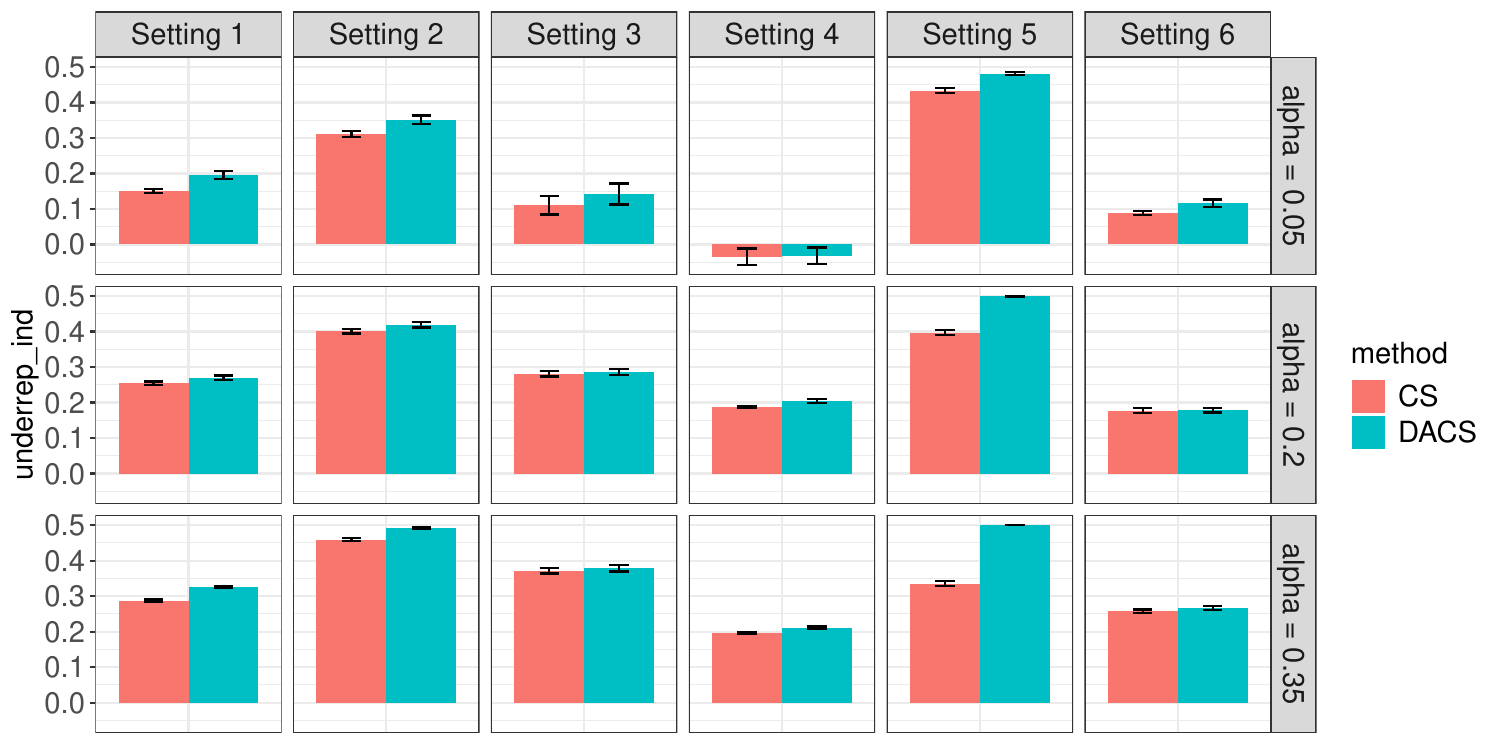}
    \caption{Analogous underrepresentation index diversity plot to Figure~\ref{fig:diversity-result-svm-cluster} with $m=100$ test samples.}
    \label{fig:underrep-result100_6}
\end{figure}

\begin{figure}
\centering
    \includegraphics[scale=0.6]{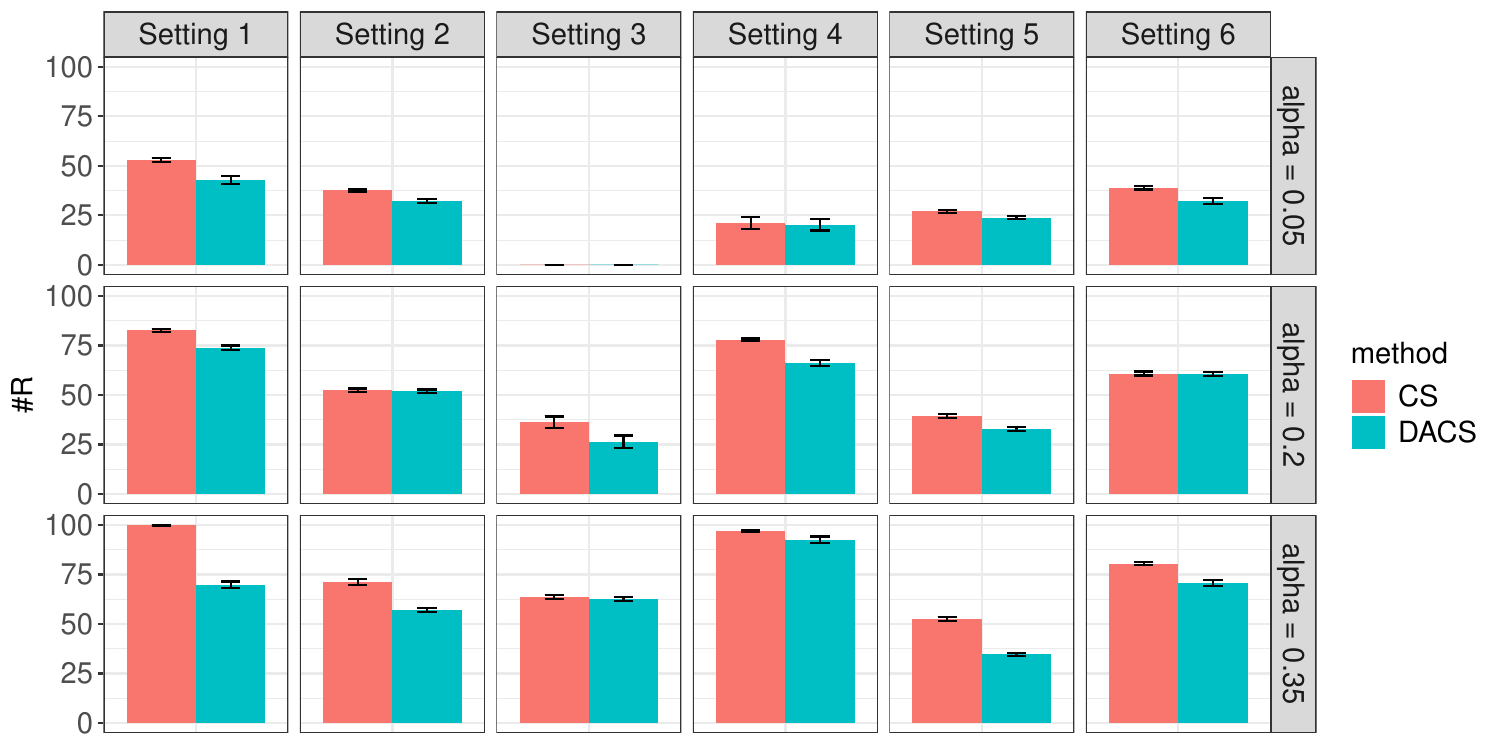}
    \caption{Analogous number of selections plot to Figure~\ref{fig:numr-result-ols-cluster} with $m=100$ test samples.}
    \label{fig:underrep-result100_7}
\end{figure}

\begin{figure}
\centering
    \includegraphics[scale=0.6]{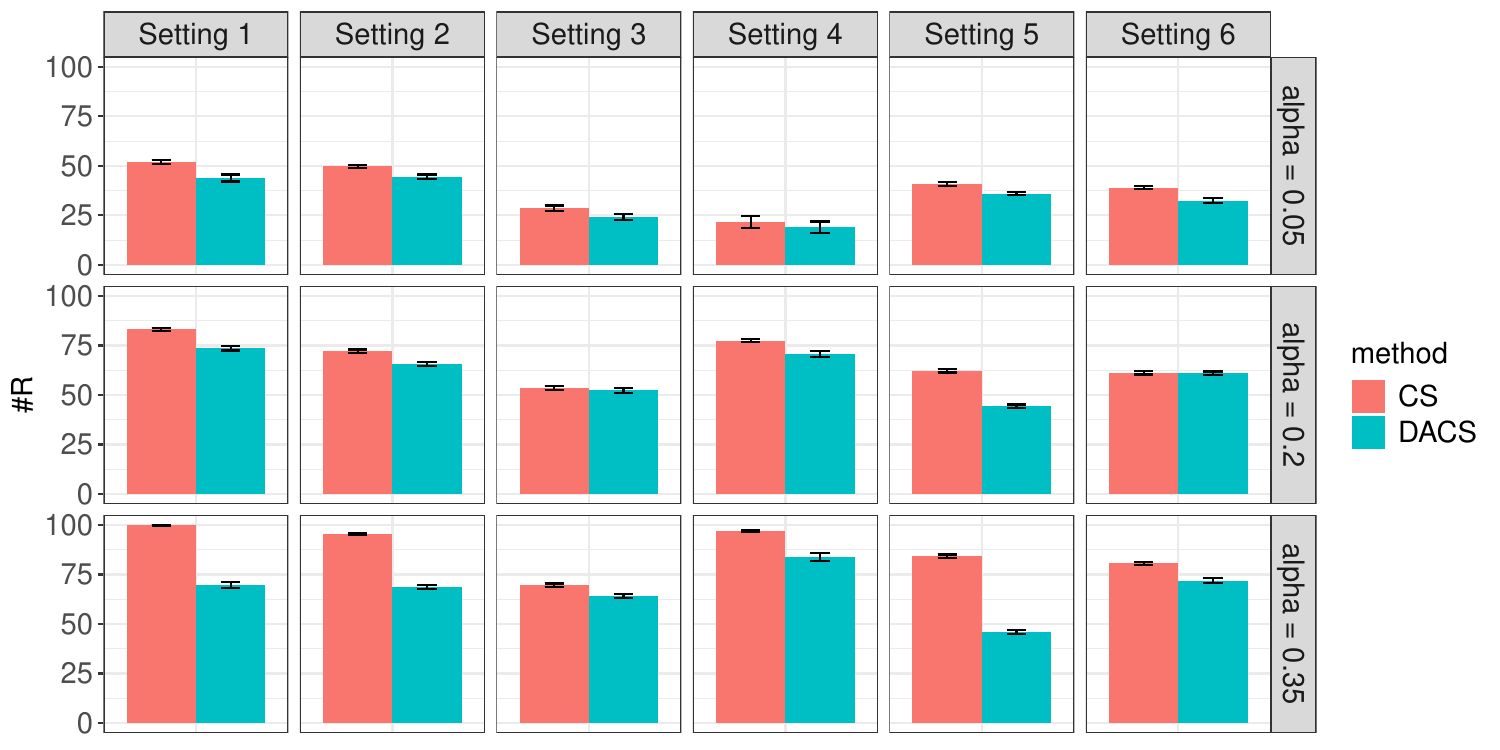}
    \caption{Analogous number of selections plot to Figure~\ref{fig:underrep-result2} with $m=100$ test samples.}
    \label{fig:underrep-result100_8}
\end{figure}

\begin{figure}
\centering
    \includegraphics[scale=0.6]{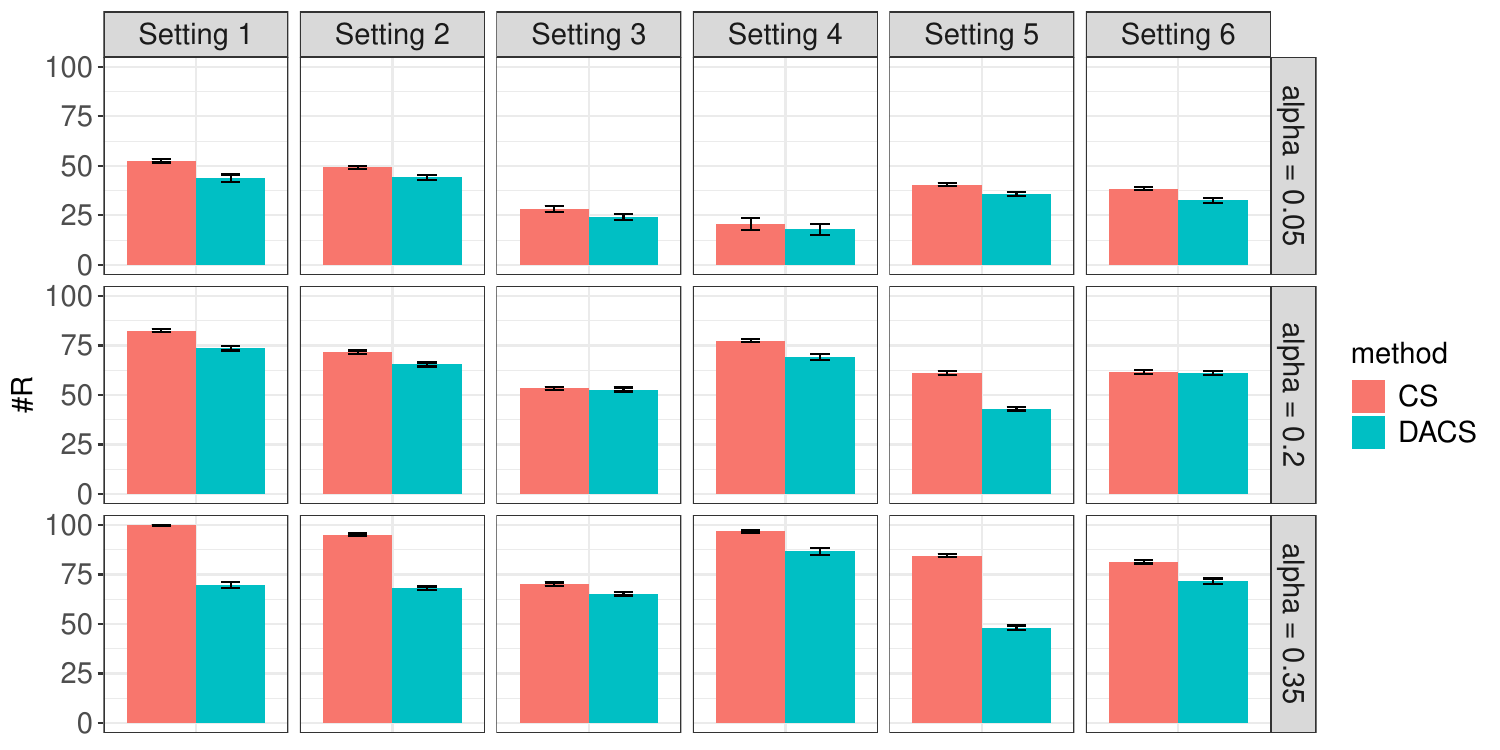}
    \caption{Analogous number of selections plot to Figure~\ref{fig:numr-result-svm-cluster} with $m=100$ test samples.}
    \label{fig:underrep-result100_9}
\end{figure}

\subsubsection{Average count of the most represented category}\label{app:most-rep}

\begin{figure}
\centering
    \includegraphics[scale=0.6]{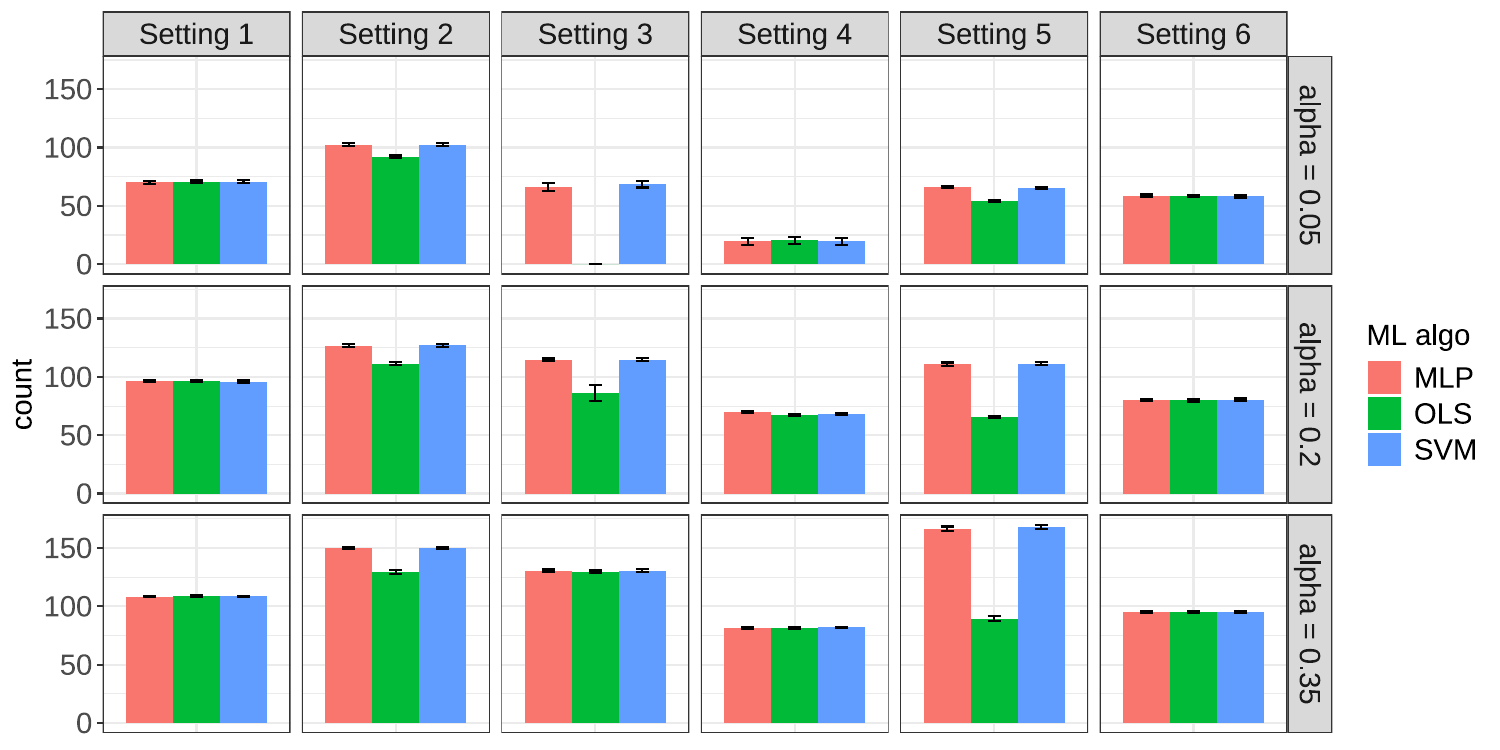}
    \caption{Average number of candidates belonging to the most-represented cluster of the CS selection set.}
    \label{fig:underrep-result-most_rep}
\end{figure}

Figure~\ref{fig:underrep-result-most_rep} shows (MC averages over simulation replicates) of $\bE\left[\max_{c \in [C]} N_c(\mathcal{R}_{\text{CS}}) \mid \mathcal{R}_{\text{CS}} \neq \emptyset\right]$, the average count of the most-represented cluster in the CS selection set $\mathcal{R}_{\text{CS}}$.

\subsection{Sharpe ratio and Markowitz objective}\label{app:add-sim-results-sharpe-markowitz}
In this section, we report additional simulation results for Settings 1 and 2 in our Sharpe ratio and Markowitz objective simulations performed in Section~\ref{sims:sharpe-markowitz} of the main text.

\subsubsection{FDR, power, and selection set size}\label{app:add-sim-results-sharpe-markowitz-fdr-power}
\begin{figure}
    \centering
    \includegraphics[scale=0.4]{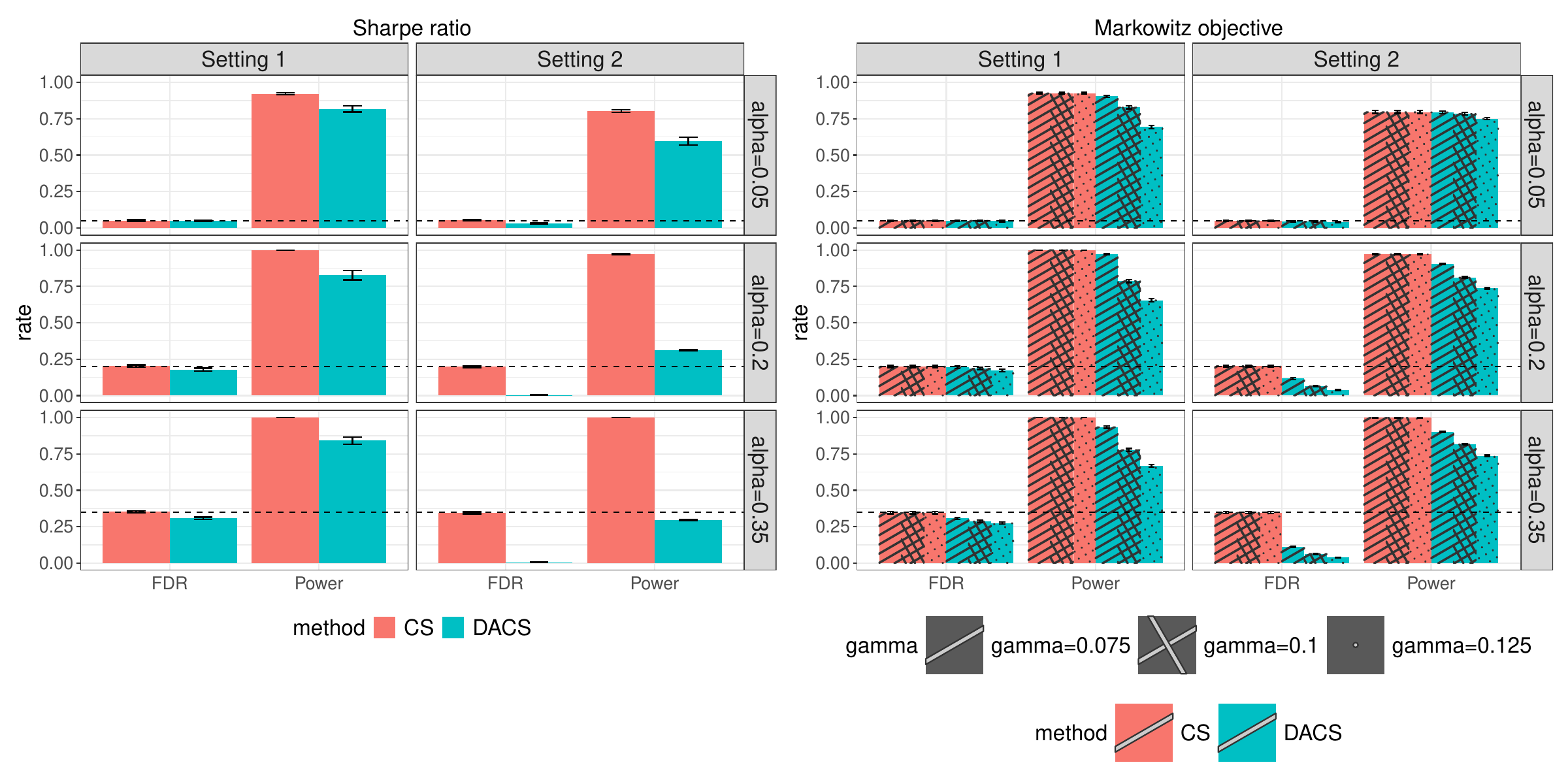}
    \caption{FDR and power under Sharpe ratio and Markowitz objective for various values of $\gamma$ (note that CS does not depend on $\gamma$, hence why its results for different $\gamma$ values are the same) for DACS (blue) and CS (red). Dashed horizontal line denotes nominal level $\alpha$.}
    \label{fig:sharpe-markowitz-fdr-power}
\end{figure}

Figure~\ref{fig:sharpe-markowitz-fdr-power} reports the FDR and power for the Sharpe ratio and Markowitz objective. Despite the fact that Proposition~\ref{relaxed-fdr} guarantees only that our method controls FDR at level $1.3\alpha$, we see here that the FDR is actually controlled below $\alpha$ for both simulation results. Unsurprisingly, the power of our method decays as $\gamma$ increases (since larger $\gamma$ indicates more preference for diversity than number of selections). In terms of power, the Sharpe ratio is more setting- and $\alpha$-dependent: in Setting 1 with $\alpha = 0.05$, it makes a comparable number of selections to CS while in Setting 2 with $\alpha =0.35$ it makes much fewer. These same remarks about power apply analogously to the number of selections of each method as shown in Figure~\ref{fig:sharpe-markowitz-numr}.

\begin{figure}
    \centering
    \includegraphics[scale=0.4]{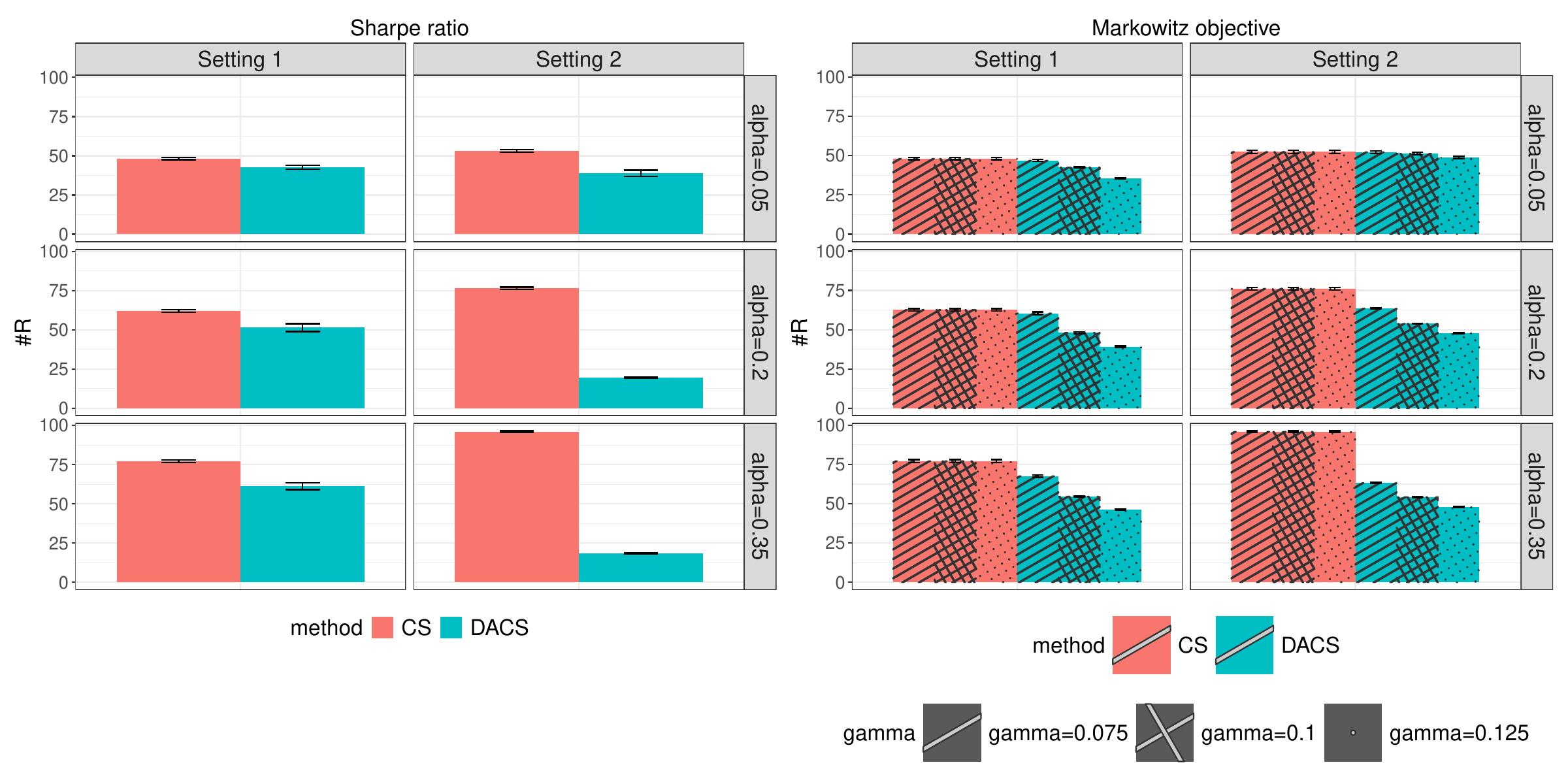}
    \caption{Number of selections made under Sharpe ratio and Markowitz objective for various values of $\gamma$ (note that CS does not depend on $\gamma$, hence why its results for different $\gamma$ values are the same) for DACS (blue) and CS (red).}
    \label{fig:sharpe-markowitz-numr}
\end{figure}

\subsubsection{Coupling/warm-starting ablation test}\label{app:warm-start-ablation}
In this section, we study the impact of the coupling/warm-starting scheme described in Section~\ref{warm-start} on both the computation time and statistical quality of the returned selection sets. We also compare to the baseline of using the MOSEK solver without coupled MC sampling or warm-starts. In all timing results, we report both the time taken just to solve the e-value optimization programs for the reward computation as well as the total time of our method (along with $\pm 1.96$ standard errors).

\begin{table}
\centering
\caption{Solver and total times for Sharpe ratio in Setting 1}
\label{tab:sharpe_times_setting1}
\begin{tabular}{lccc}
\toprule
$\alpha$ & Method & Solver times & Total times \\
\midrule
\multirow{4}{*}{$0.05$} & Warm-started PGD with coupling & $2.88 \pm 0.12$ & $10.72 \pm 0.40$ \\
 & PGD with coupling only & $3.48 \pm 0.15$ & $11.35 \pm 0.44$ \\
 & PGD without coupling or warm-starting & $3.47 \pm 0.15$ & $11.68 \pm 0.43$ \\
 & MOSEK (no coupling) & $802.03 \pm 35.21$ & $811.67 \pm 35.63$ \\
\hline
\multirow{4}{*}{$0.2$} & Warm-started PGD with coupling & $12.46 \pm 0.46$ & $45.04 \pm 1.65$ \\
 & PGD with coupling only & $18.83 \pm 0.74$ & $51.69 \pm 1.93$ \\
 & PGD without coupling or warm-starting & $18.91 \pm 0.74$ & $50.13 \pm 1.85$ \\
 & MOSEK (no coupling) & $3712.39 \pm 159.07$ & $3751.08 \pm 160.85$ \\
\hline
\multirow{4}{*}{$0.35$} & Warm-started PGD with coupling & $26.39 \pm 0.93$ & $89.60 \pm 2.98$ \\
 & PGD with coupling only & $43.06 \pm 1.61$ & $107.02 \pm 3.68$ \\
 & PGD without coupling or warm-starting & $43.25 \pm 1.61$ & $103.55 \pm 3.54$ \\
 & MOSEK (no coupling) & $6815.74 \pm 288.54$ & $6889.39 \pm 291.86$ \\
\bottomrule
\end{tabular}
\end{table}

\begin{table}
\centering
\caption{Solver and total times for Sharpe ratio in Setting 2}
\label{tab:sharpe_times_setting2}
\begin{tabular}{lccc}
\toprule
$\alpha$ & Method & Solver times & Total times \\
\midrule
\multirow{4}{*}{$0.05$} & Warm-started PGD with coupling & $2.25 \pm 0.11$ & $12.37 \pm 0.54$ \\
 & PGD with coupling only & $2.76 \pm 0.14$ & $12.91 \pm 0.57$ \\
 & PGD without coupling or warm-starting & $2.74 \pm 0.14$ & $13.08 \pm 0.55$ \\
 & MOSEK (no coupling) & $963.20 \pm 48.74$ & $975.69 \pm 49.37$ \\
\hline
\multirow{4}{*}{$0.2$} & Warm-started PGD with coupling & $12.39 \pm 0.35$ & $62.53 \pm 1.78$ \\
 & PGD with coupling only & $18.22 \pm 0.52$ & $68.53 \pm 1.95$ \\
 & PGD without coupling or warm-starting & $18.27 \pm 0.53$ & $65.93 \pm 1.87$ \\
 & MOSEK (no coupling) & $4979.44 \pm 202.97$ & $5036.46 \pm 205.41$ \\
\hline
\multirow{4}{*}{$0.35$} & Warm-started PGD with coupling & $25.04 \pm 0.63$ & $121.52 \pm 3.11$ \\
 & PGD with coupling only & $38.73 \pm 0.94$ & $135.60 \pm 3.46$ \\
 & PGD without coupling or warm-starting & $38.75 \pm 0.95$ & $129.90 \pm 3.29$ \\
 & MOSEK (no coupling) & $9124.01 \pm 363.34$ & $9233.11 \pm 367.83$ \\
\bottomrule
\end{tabular}
\end{table}

\begin{table}
\centering
\caption{Solver and total times for Markowitz objective in Setting 1}
\label{tab:markowitz_times_setting1}
\begin{tabular}{llccc}
\toprule
$\alpha$ & $\gamma$ & Method & Solver times & Total times \\
\midrule
\multirow{4}{*}{$0.05$} & \multirow{4}{*}{$0.075$} & Warm-started PGD with coupling & $15.12 \pm 1.70$ & $19.92 \pm 1.83$ \\
 &  & PGD with coupling only & $20.10 \pm 2.30$ & $24.85 \pm 2.41$ \\
 &  & PGD without coupling or warm-starting & $20.02 \pm 2.29$ & $25.19 \pm 2.39$ \\
 &  & MOSEK (no coupling) & $978.97 \pm 41.87$ & $986.87 \pm 42.17$ \\
\hdashline
\multirow{4}{*}{$0.05$} & \multirow{4}{*}{$0.1$} & Warm-started PGD with coupling & $25.68 \pm 1.68$ & $30.42 \pm 1.80$ \\
 &  & PGD with coupling only & $32.88 \pm 2.02$ & $37.61 \pm 2.15$ \\
 &  & PGD without coupling or warm-starting & $32.72 \pm 2.01$ & $37.83 \pm 2.12$ \\
 &  & MOSEK (no coupling) & $928.11 \pm 41.53$ & $935.56 \pm 41.83$ \\
\hdashline
\multirow{4}{*}{$0.05$} & \multirow{4}{*}{$0.125$} & Warm-started PGD with coupling & $29.81 \pm 1.78$ & $34.84 \pm 1.93$ \\
 &  & PGD with coupling only & $36.64 \pm 2.11$ & $41.66 \pm 2.25$ \\
 &  & PGD without coupling or warm-starting & $36.47 \pm 2.10$ & $41.81 \pm 2.23$ \\
 &  & MOSEK (no coupling) & $924.85 \pm 39.81$ & $932.26 \pm 40.11$ \\
\hline
\multirow{4}{*}{$0.2$} & \multirow{4}{*}{$0.075$} & Warm-started PGD with coupling & $169.94 \pm 8.92$ & $187.43 \pm 9.44$ \\
 &  & PGD with coupling only & $212.19 \pm 11.26$ & $229.79 \pm 11.76$ \\
 &  & PGD without coupling or warm-starting & $211.96 \pm 11.23$ & $227.81 \pm 11.65$ \\
 &  & MOSEK (no coupling) & $4774.57 \pm 190.83$ & $4803.10 \pm 191.87$ \\
\hdashline
\multirow{4}{*}{$0.2$} & \multirow{4}{*}{$0.1$} & Warm-started PGD with coupling & $190.09 \pm 9.19$ & $207.57 \pm 9.70$ \\
 &  & PGD with coupling only & $239.48 \pm 11.87$ & $257.05 \pm 12.36$ \\
 &  & PGD without coupling or warm-starting & $238.89 \pm 11.82$ & $254.66 \pm 12.23$ \\
 &  & MOSEK (no coupling) & $4560.04 \pm 208.74$ & $4585.80 \pm 209.86$ \\
\hdashline
\multirow{4}{*}{$0.2$} & \multirow{4}{*}{$0.125$} & Warm-started PGD with coupling & $196.14 \pm 8.84$ & $214.40 \pm 9.41$ \\
 &  & PGD with coupling only & $249.08 \pm 11.51$ & $267.43 \pm 12.06$ \\
 &  & PGD without coupling or warm-starting & $248.44 \pm 11.48$ & $264.89 \pm 11.95$ \\
 &  & MOSEK (no coupling) & $4529.18 \pm 187.22$ & $4554.73 \pm 188.26$ \\
\hline
\multirow{4}{*}{$0.35$} & \multirow{4}{*}{$0.075$} & Warm-started PGD with coupling & $446.91 \pm 20.06$ & $477.22 \pm 20.91$ \\
 &  & PGD with coupling only & $579.45 \pm 26.03$ & $609.95 \pm 26.87$ \\
 &  & PGD without coupling or warm-starting & $578.89 \pm 26.00$ & $605.34 \pm 26.70$ \\
 &  & MOSEK (no coupling) & $8554.19 \pm 330.97$ & $8598.73 \pm 332.59$ \\
\hdashline
 \multirow{4}{*}{$0.35$}& \multirow{4}{*}{$0.1$} & Warm-started PGD with coupling & $455.70 \pm 20.95$ & $485.68 \pm 21.96$ \\
 &  & PGD with coupling only & $585.91 \pm 27.22$ & $616.11 \pm 28.23$ \\
 &  & PGD without coupling or warm-starting & $585.73 \pm 27.25$ & $611.96 \pm 28.10$ \\
 &  & MOSEK (no coupling) & $8363.83 \pm 322.83$ & $8407.32 \pm 324.55$ \\
\hdashline
 \multirow{4}{*}{$0.35$}& \multirow{4}{*}{$0.125$} & Warm-started PGD with coupling & $468.07 \pm 19.89$ & $499.64 \pm 20.79$ \\
 &  & PGD with coupling only & $604.76 \pm 25.90$ & $636.57 \pm 26.78$ \\
 &  & PGD without coupling or warm-starting & $603.68 \pm 25.82$ & $631.24 \pm 26.56$ \\
 &  & MOSEK (no coupling) & $8243.90 \pm 331.17$ & $8286.28 \pm 332.82$ \\
\bottomrule
\end{tabular}
\end{table}

\begin{table}
\centering
\caption{Solver and total times for Markowitz objective in Setting 2}
\label{tab:markowitz_times_setting2}
\begin{tabular}{llccc}
\toprule
$\alpha$ & $\gamma$ & Method & Solver times & Total times \\
\midrule
\multirow{4}{*}{$0.05$} & \multirow{4}{*}{$0.075$} & Warm-started PGD with coupling & $3.75 \pm 0.56$ & $9.19 \pm 0.72$ \\
 &  & PGD with coupling only & $4.37 \pm 0.76$ & $9.81 \pm 0.91$ \\
 &  & PGD without coupling or warm-starting & $4.32 \pm 0.75$ & $10.11 \pm 0.90$ \\
 &  & MOSEK (no coupling) & $1110.82 \pm 46.82$ & $1119.43 \pm 47.14$ \\
\hdashline
\multirow{4}{*}{$0.05$} & \multirow{4}{*}{$0.1$} & Warm-started PGD with coupling & $11.02 \pm 1.31$ & $16.37 \pm 1.50$ \\
 &  & PGD with coupling only & $13.94 \pm 1.74$ & $19.28 \pm 1.92$ \\
 &  & PGD without coupling or warm-starting & $13.85 \pm 1.73$ & $19.56 \pm 1.89$ \\
 &  & MOSEK (no coupling) & $1081.79 \pm 46.52$ & $1090.77 \pm 46.89$ \\
\hdashline
 \multirow{4}{*}{$0.05$}& \multirow{4}{*}{$0.125$} & Warm-started PGD with coupling & $17.66 \pm 1.60$ & $23.40 \pm 1.79$ \\
 &  & PGD with coupling only & $21.97 \pm 2.01$ & $27.70 \pm 2.19$ \\
 &  & PGD without coupling or warm-starting & $21.80 \pm 1.99$ & $27.87 \pm 2.15$ \\
 &  & MOSEK (no coupling) & $1074.27 \pm 50.94$ & $1082.36 \pm 51.29$ \\
\hline
\multirow{4}{*}{$0.2$} & \multirow{4}{*}{$0.075$} & Warm-started PGD with coupling & $211.68 \pm 10.30$ & $234.21 \pm 10.78$ \\
 &  & PGD with coupling only & $278.87 \pm 13.09$ & $301.54 \pm 13.57$ \\
 &  & PGD without coupling or warm-starting & $278.54 \pm 13.05$ & $298.24 \pm 13.43$ \\
 &  & MOSEK (no coupling) & $6406.80 \pm 227.62$ & $6440.99 \pm 228.69$ \\
\hdashline
 \multirow{4}{*}{$0.2$}& \multirow{4}{*}{$0.1$} & Warm-started PGD with coupling & $232.85 \pm 9.62$ & $255.13 \pm 10.22$ \\
 &  & PGD with coupling only & $301.42 \pm 11.92$ & $323.80 \pm 12.53$ \\
 &  & PGD without coupling or warm-starting & $301.07 \pm 11.88$ & $320.56 \pm 12.38$ \\
 &  & MOSEK (no coupling) & $6117.86 \pm 236.64$ & $6150.26 \pm 237.80$ \\
\hdashline
\multirow{4}{*}{$0.2$} & \multirow{4}{*}{$0.125$} & Warm-started PGD with coupling & $232.77 \pm 9.22$ & $255.51 \pm 9.84$ \\
 &  & PGD with coupling only & $298.24 \pm 11.34$ & $321.08 \pm 11.96$ \\
 &  & PGD without coupling or warm-starting & $297.88 \pm 11.29$ & $317.75 \pm 11.80$ \\
 &  & MOSEK (no coupling) & $6101.47 \pm 236.81$ & $6132.71 \pm 238.07$ \\
\hline
\multirow{4}{*}{$0.35$} & \multirow{4}{*}{$0.075$} & Warm-started PGD with coupling & $713.57 \pm 18.02$ & $754.00 \pm 19.05$ \\
 &  & PGD with coupling only & $892.87 \pm 21.83$ & $933.61 \pm 22.86$ \\
 &  & PGD without coupling or warm-starting & $892.73 \pm 21.82$ & $927.10 \pm 22.66$ \\
 &  & MOSEK (no coupling) & $12571.45 \pm 422.80$ & $12630.86 \pm 424.72$ \\
\hdashline
 \multirow{4}{*}{$0.35$}& \multirow{4}{*}{$0.1$} & Warm-started PGD with coupling & $687.72 \pm 17.15$ & $726.97 \pm 18.18$ \\
 &  & PGD with coupling only & $859.72 \pm 21.07$ & $899.25 \pm 22.10$ \\
 &  & PGD without coupling or warm-starting & $859.51 \pm 21.05$ & $892.96 \pm 21.90$ \\
 &  & MOSEK (no coupling) & $11974.46 \pm 410.04$ & $12029.83 \pm 411.98$ \\
\hdashline
\multirow{4}{*}{$0.35$} & \multirow{4}{*}{$0.125$} & Warm-started PGD with coupling & $680.23 \pm 17.63$ & $721.35 \pm 18.80$ \\
 &  & PGD with coupling only & $854.03 \pm 21.50$ & $895.45 \pm 22.67$ \\
 &  & PGD without coupling or warm-starting & $851.64 \pm 21.32$ & $886.58 \pm 22.29$ \\
 &  & MOSEK (no coupling) & $11875.24 \pm 414.60$ & $11930.95 \pm 416.49$ \\
\bottomrule
\end{tabular}
\end{table}

Tables~\ref{tab:sharpe_times_setting1} and \ref{tab:sharpe_times_setting2} report times for Settings 1 and 2 in our Sharpe ratio simulations and Tables~\ref{tab:markowitz_times_setting1} and \ref{tab:markowitz_times_setting2} do the same for each setting considered in our simulations for the Markowitz objective. Across both simulation settings and for both diversity metrics, the PGD solver on its own (i.e., without coupling or warm-starting) is significantly faster than MOSEK: for both settings in our Sharpe ratio simulations, the speedup is always (approximately) between $67$--$76\times$ across all values of $\alpha$ while for the Markowitz objective simulations, the speedup ranges from $13\times$ to $111\times$ and seems to be diminished for larger values of $\alpha$ and $\gamma$. The warm-start heuristic in both cases sometimes yields a further reduction in computation time; we now study these reductions in comparison to the PGD solver without warm-starts or coupling (in most cases, the PGD solver with coupled MC sampling but no warm-starting is either essentially the same speed as or slower than the vanilla PGD solver which uses neither, most likely due to the overhead of constructing the coupled samples). For Setting 1 of the Sharpe ratio, the reduction in computation time is by between $8$--$13\%$ for all values of $\alpha$; in Setting 2 even though the warm-starting does reduce the total solver time, this does not translate into a significant reduction in overall computation time compared to the PGD solver with no warm-starting or coupling (though, there is still a small reduction), likely because the overhead of coupling the Monte Carlo samples cancels out the gains from warm-starting. For the Markowitz objective the warm-start further reduces computation, across all values of $\alpha$ and $\gamma$, by approximately $17$--$21\%$ in Setting 1 and $9$--$21\%$ in Setting 2.

\begin{figure}
    \centering
    \includegraphics[scale=0.4]{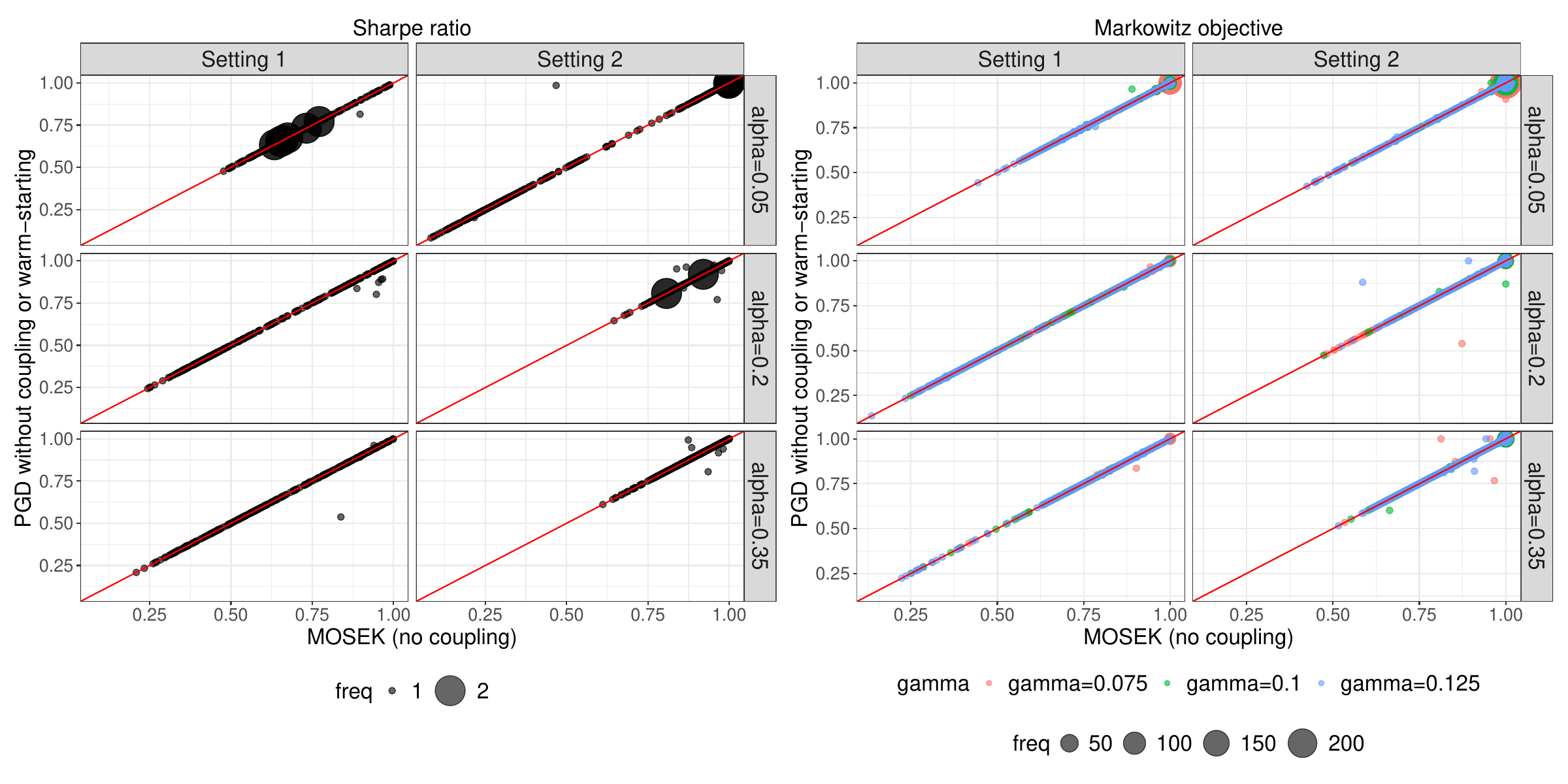}
    \caption{DACS' normalized diversity under $\hat{F}^{\varphi}_{\text{baseline}}$ for both the Sharpe ratio and Markowitz objective for various values of $\gamma$ ($x$-coordinate is the value using the MOSEK solver without coupling or warm-starts while the $y$-coordinate is the value obtained using the custom PGD solver without coupling or warm-starts). Each scatter point represents a different simulation replicate; results reported for two simulation settings and three nominal levels $\alpha$. Red line is $y=x$.}
    \label{sharpe-markowitz:comp2}
\end{figure}

\begin{figure}
    \centering
    \includegraphics[scale=0.4]{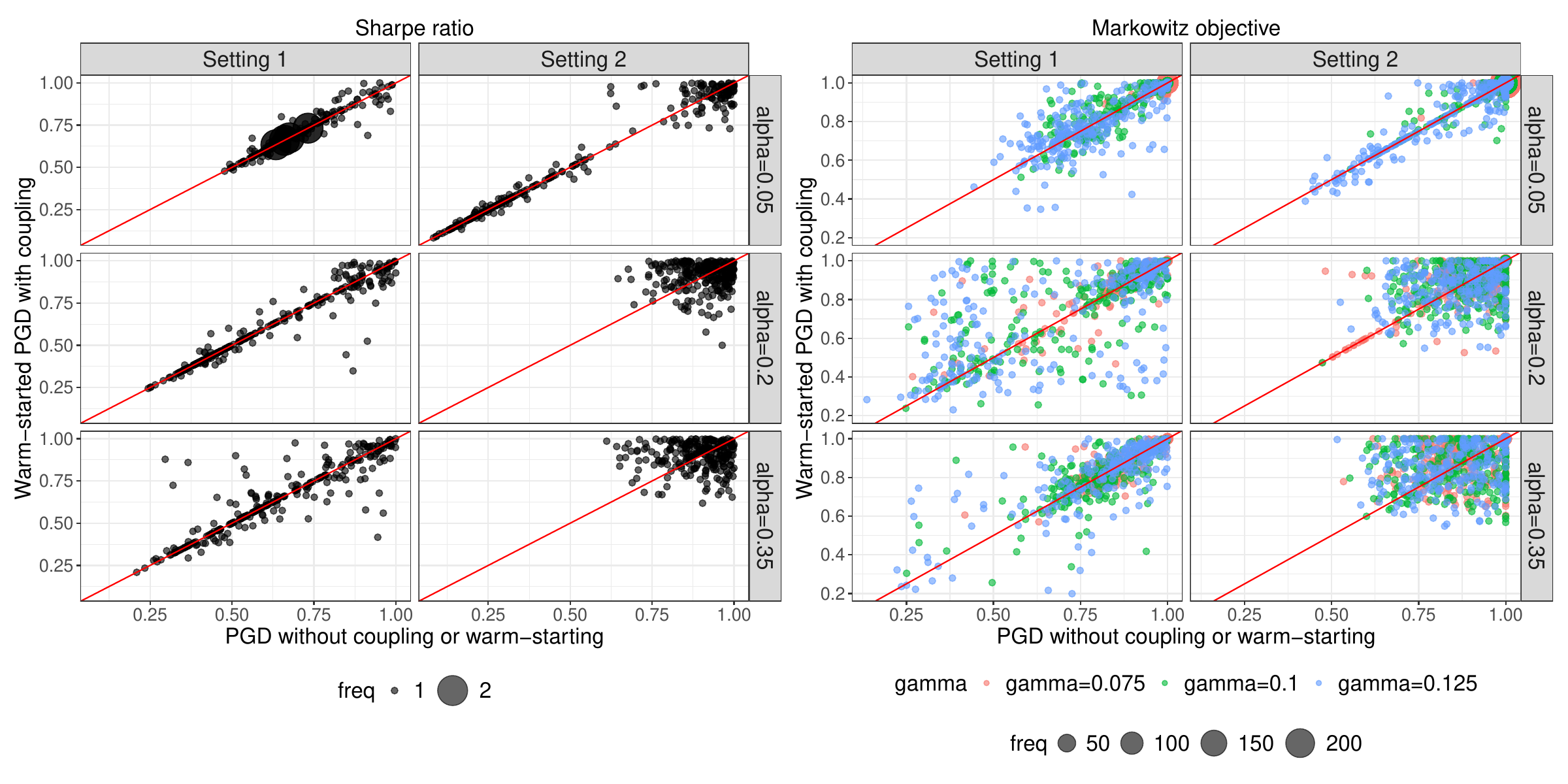}
    \caption{DACS' normalized diversity under $\hat{F}^{\varphi}_{\text{baseline}}$ for both the Sharpe ratio and Markowitz objective for various values of $\gamma$ ($x$-coordinate is the value using the custom PGD solvers without coupling or warm-starts while the $y$-coordinate is the value obtained using the custom PGD solver with coupling and warm-starts). Each scatter point represents a different simulation replicate; results reported for two simulation settings and three nominal levels $\alpha$. Red line is $y=x$.}
    \label{sharpe-markowitz:comp1}
\end{figure}

Figures~\ref{sharpe-markowitz:comp2} and \ref{sharpe-markowitz:comp1} study how the statistical quality (i.e., diversity) of the selection sets change depending on the solver that is used. Figure~\ref{sharpe-markowitz:comp2} compares the normalized diversity values when using the MOSEK solver versus our PGD solver without warm-starts or coupling. It shows that neither solver yields results that are systematically significantly better than the other. Figure~\ref{sharpe-markowitz:comp1} shows that, while coupled MC sampling can result in different results than sampling without coupling in any given simulation replicate, it does not seem to make a significant difference in statistical quality on aggregate.

\end{document}